\documentclass[a4paper]{article}
\usepackage{comment}
\usepackage[utf8]{inputenc}
\usepackage{enumerate}
\usepackage[OT4]{fontenc}
\usepackage{graphicx}
\usepackage{amssymb}
\usepackage{amsmath}
\usepackage{amsfonts}
\usepackage{amsthm}
\usepackage{authblk}
\usepackage{epstopdf}
\usepackage{algorithm}
\usepackage[noend]{algpseudocode}
\usepackage{algorithmicx}
\usepackage{fullpage}
\usepackage{todonotes}
\usepackage[hidelinks]{hyperref}
\usepackage{breakurl}
\usepackage[capitalise]{cleveref}
\usepackage{thmtools}
\usepackage{thm-restate}
\usepackage{microtype}

\ifx\shortversion\empty
	\newenvironment{shortv}{}{}
	\excludecomment{fullv}
    \newcommand{\appref}[1]{\marginpar{\tiny #1\\full version}}
\else
	
	\excludecomment{shortv}
    \newcommand{\appref}[1]{}
\fi

\urldef{\mailsa}\path|a.karczmarz@mimuw.edu.pl|

\newtheorem{theorem}{Theorem}[section]
\newtheorem{lemma}[theorem]{Lemma}
\newtheorem{observation}[theorem]{Observation}
\newtheorem{fact}[theorem]{Fact}
\newtheorem{definition}[theorem]{Definition}
\newtheorem{corollary}[theorem]{Corollary}
\newtheorem{claim}[theorem]{Claim}
\newtheorem{invariant}[theorem]{Invariant}

\theoremstyle{remark}
\newtheorem{remark}[theorem]{Remark}

\newcommand{\makeop}{\mathtt{make\_string}}
\newcommand{\dropop}{\mathtt{drop\_string}}
\newcommand{\compop}{\mathtt{compare}}
\newcommand{\actop}{\mathtt{activate}}
\newcommand{\dactop}{\mathtt{deactivate}}

\newcommand{\eqop}{\mathtt{equal}}
\newcommand{\concop}{\mathtt{concat}}
\newcommand{\splitop}{\mathtt{split}}
\newcommand{\findop}{\mathtt{find}}
\newcommand{\lcpop}{\mathtt{LCP}}
\newcommand{\layerop}{\mathtt{layer}}
\newcommand{\anchfind}{\mathtt{find}}

\newcommand{\symbols}{\mathcal{S}}
\newcommand{\sigs}{\bar{\Sigma}}

\newcommand{\lastsymb}{\mathcal{Z}}
\newcommand{\compress}{\textsc{Compress}}
\newcommand{\depth}{\textsc{Depth}}
\newcommand{\rle}{\textsc{Rle}}
\newcommand{\shrink}[1]{\ensuremath{\textsc{Shrink}_{#1}}}
\newcommand{\cshrink}[1]{\ensuremath{\overline{\textsc{Shrink}_{#1}}}}

\newcommand{\str}{w}
\newcommand{\strb}{y}

\newcommand{\hs}{h}

\newcommand{\grammar}{\mathcal{G}}
\newcommand{\mword}{B}
\newcommand{\decomp}{D}

\newcommand{\spair}[3]{\ensuremath{#1 \rightarrow #2#3}}
\newcommand{\spower}[3]{\ensuremath{#1 \rightarrow #2^{#3}}}

\newcommand{\slev}{\textit{level}}
\newcommand{\sanch}{\textit{anch}}
\newcommand{\slength}{\textit{length}}
\newcommand{\stree}{\mathcal{T}}
\newcommand{\sstr}{\textit{str}}
\newcommand{\ustree}{\mathcal{\overline{T}}}
\newcommand{\spar}[1]{\mathrm{par}_{\stree(#1)}}
\newcommand{\spars}{\mathtt{pars}}
\newcommand{\uspar}{\mathrm{par}}
\newcommand{\usleft}{\mathrm{left}}
\newcommand{\usright}{\mathrm{right}}
\newcommand{\uschild}{\mathrm{child}}
\newcommand{\ussig}{\mathrm{sig}}
\newcommand{\stlayer}{\Lambda}

\newcommand{\itroot}{\mathtt{root}}
\newcommand{\itbegin}{\mathtt{begin}}
\newcommand{\itend}{\mathtt{end}}

\newcommand{\coll}{\mathcal{W}}

\newcommand{\itparent}{\mathtt{parent}}
\newcommand{\itchild}{\mathtt{child}}
\newcommand{\itindex}{\mathtt{index}}
\newcommand{\itlevel}{\mathtt{level}}
\newcommand{\itleft}{\mathtt{left}}
\newcommand{\itrepr}{\mathtt{repr}}
\newcommand{\itright}{\mathtt{right}}
\newcommand{\itnil}{\mathtt{nil}}
\newcommand{\itanch}{\mathtt{anchor}}
\newcommand{\itdegree}{\mathtt{degree}}
\newcommand{\itfirst}{\mathtt{First}}
\newcommand{\itlast}{\mathtt{Last}}
\newcommand{\itsig}{\mathtt{sig}}

\newcommand{\itrskip}{\mathtt{rskip}}
\newcommand{\itlskip}{\mathtt{lskip}}
\newcommand{\itrext}{\mathtt{rext}}
\newcommand{\itlext}{\mathtt{lext}}

\newcommand{\stpush}{\mathtt{push}}
\newcommand{\stpop}{\mathtt{pop}}
\newcommand{\sttop}{\mathtt{top}}
\newcommand{\stempty}{\mathtt{empty}}
\newcommand{\stcomp}{\mathtt{collapse}}

\newcommand{\poly}{\mathrm{poly}}

\newcommand{\insertop}{\mathtt{insert}}
\newcommand{\deleteop}{\mathtt{delete}}
\newcommand{\connectedop}{\mathtt{connected}}
\newcommand{\updateop}{\mathtt{update}}
\newcommand{\verifyop}{\mathtt{verify}}

\newcommand{\algptr}{P}
\newcommand{\algqtr}{Q}

\newcommand{\eps}{\varepsilon}
\newcommand{\edot}{{\cdot}}

\begin{document}

\begin{titlepage}
\date{}
\title{Optimal Dynamic Strings\thanks{Work done while Paweł Gawrychowski held a post-doctoral position at  Warsaw  Center  of
Mathematics and Computer Science. Piotr Sankowski is supported by the Polish National Science Center, grant no 2014/13/B/ST6/00770.}}

\author[1]{Paweł Gawrychowski}
\author[1]{Adam Karczmarz}
\author[1]{Tomasz Kociumaka}
\author[2]{\\Jakub Łącki}
\author[1]{Piotr Sankowski}

\affil[1]{Institute of Informatics, University of Warsaw, Poland}
\affil[ ]{\texttt{[gawry,a.karczmarz,kociumaka,sank]@mimuw.edu.pl}}
\affil[2]{Sapienza University of Rome, Italy}
\affil[ ]{\texttt{j.lacki@mimuw.edu.pl}}

\maketitle
\thispagestyle{empty}

\begin{abstract}
In this paper we study the fundamental problem of maintaining
a~dynamic collection of strings under the following operations:
\begin{itemize}
\item $\concop$ -- concatenates two strings,
\item $\splitop$ -- splits a string into two at a given position,
\item $\compop$ -- finds the lexicographical order (less, equal, greater) between two strings,
\item $\lcpop$ -- calculates the longest common prefix of two strings.
\end{itemize}
We present an efficient data structure for this problem, where an update requires
only $O(\log n)$ worst-case time with high probability, with $n$ being the total length of all strings
in the collection, and a query takes constant worst-case time.
On the lower bound side, we prove that even if the only possible query is checking equality of
two strings, either updates or queries take amortized $\Omega(\log n)$ time; hence
our implementation is optimal.

Such operations can be used as a basic building block to solve other string problems.
We provide two examples. First, we can augment our data structure to provide
pattern matching queries that may locate occurrences of a specified pattern $p$ in the
strings in our collection in optimal $O(|p|)$ time, at the expense of increasing update time to $O(\log^2 n)$.
Second, we show how to maintain a history of an edited text, processing updates in $O(\log t \log \log t)$ time,
where $t$ is the number of edits, and how to support pattern matching queries against the whole
history in $O(|p| \log t \log \log t)$ time. 

Finally, we note that our data structure can be applied to test dynamic tree isomorphism
and to compare strings generated by dynamic straight-line grammars. 
\end{abstract}
\end{titlepage}
\newpage

\tableofcontents
\newpage

\section{Introduction}
Imagine a set of text files that are being edited. Each edit consists in adding or removing a single character or copy-pasting longer
fragments, possibly between different files. We want to perform fundamental queries on such
files. For example, we might be interested in checking equality of two files, finding their
first mismatch, or pattern matching, that is, locating occurrences of a given pattern in all files.
All these queries are standard operations supported in popular programs, e.g., text editors. However, modern
text editors rise new interesting problems as they often maintain the entire history of edits.
For example, in TextEdit on OS X Lion the save button does not exist anymore.
Instead, the editor stores all versions of the file that can be scrolled through using the so called timeline.
Analogously, in Google Docs one sees a file together with its whole history of edits.
A natural question arises whether we can efficiently support a \emph{find} operation,
that is whether it is possible to efficiently locate all occurrences of a given pattern in all versions of a file.

We develop improved and more efficient tools that can be used as a basic building block in algorithms on dynamic texts.
In particular, we present an optimal data structure for dynamic strings equality.
In this problem we want to maintain a collection of strings that can be concatenated, split, and tested for equality. Our algorithm
requires $O(\log n)$ worst-case time with high probability\footnote{Our algorithms always return correct answers, and the
randomization only impacts the running time.} for an update, where $n$ is the total length of the strings in the collection,
and compares two strings in $O(1)$ worst-case time, assuming $O(1)$-time arithmetic on
the lengths of the strings.

Our solution is obtained by maintaining a certain context-insensitive representation of the strings,
which is similar but conceptually simpler than what has been used in the previous works~\cite{Alstrup,Mehlhorn}
(which use the same assumptions concerning randomization and arithmetics).
This allows us to obtain better running time in the end, but getting the desired bound requires a deeper insight in comparison to previous results.
While our improvement in the update time is $O(\log^*n)$, we
believe that obtaining tight bounds for such a fundamental problem is of major interest.
Thus, we also provide a matching $\Omega(\log n)$ lower bound for the amortized complexity
of update or query, showing that our solution is indeed the final answer. This lower bound is obtained
by appropriately modifying the proof of the lower bound for dynamic path connectivity~\cite{logarithmic}.

Next, we describe how to support lexicographical comparison and computing the longest
common prefix of two strings in constant worst-case time. All time complexities remain unchanged if we
want to maintain a collection of persistent strings, that is, concatenate and split do not destroy their
arguments. The aforementioned lower bound applies to non-persistent strings, and hence also to persistent
strings. Nevertheless, in the persistent setting $n$ might be exponential with respect to the total input size,
while it is only polynomial in the lower-bound examples.
Thus, the bounds would not meet if we measured the running time as a function of the total input size.
We also show how to apply our date structure to: support pattern matching queries in dynamic string collections, 
support find operation in history of edits, compare strings generated by dynamic straight line grammars, and test dynamic
tree isomorphism. 

\subsection{Our Results}
\paragraph{Dynamic String Equality and Lexicographic Comparison}
The first result on dynamic string equality was given by Sundar and Tarjan~\cite{sundartarjan}.
They achieved amortized $O(\sqrt{n \log m} + \log m)$ time for an update, where $m$ is
the number of operations executed so far and $n$ is the total length of the strings. This was later improved
to randomized logarithmic time for the special case of repetition-free sequences~\cite{Pugh}.
The first improvement for the general case was due to Mehlhorn et al.~\cite{Mehlhorn}, who
decreased the update time to expected $O(\log^2 n)$ time. They also provided a deterministic version
working in $O(\log n (\log n + \log m\log^*m))$ time. Finally, much later, the deterministic algorithm of~\cite{Mehlhorn} was
substantially improved by Alstrup et al.~\cite{Alstrup}, who achieved $O(\log n \log^* n)$
update time.\footnote{While not explicitly stated in the paper, this bound holds with high probability, as the algorithm uses a hash table.} In all these solutions, equality of two strings can be checked in worst-case
constant time. We provide the final answer by further improving the update time to $O(\log n)$ with
high probability and showing that either an update or an equality test requires $\Omega(\log n)$ 
time, even if one allows amortization and randomization (either Las Vegas or Monte Carlo).%
\begin{shortv}%
\footnote{This extended abstract aims to present a brief overview of our results. The references to formal definitions, and detailed descriptions included in the full version are put on the margin.}
\end{shortv}

We note that it is very simple to achieve $O(\log n)$ update time for maintaining a non-persistent
family of strings under concatenate and split operations, if we allow the equality queries
to give an incorrect result with polynomially small probability.
We represent every string by a balanced search tree with characters in the leaves and every
node storing a fingerprint of the sequence represented by its descendant leaves. However, it is not clear how to make
the answers always correct in this approach (even if the time bounds should
only hold in expectation). Furthermore, it seems that both computing the longest common
prefix of two strings of length $n$ and comparing them lexicographically requires $\Omega(\log^2n)$ time in this approach.
This is a serious drawback, as the lexicographic comparison is a crucial ingredient in our applications related to pattern matching.

\paragraph{Straight-line Grammars}
A straight-line grammar (also known as a straight-line program) is a context-free grammar
generating exactly one string. Such grammars are a very convenient abstraction of text
compression. For instance, it is known that given a LZ77 representation of a text of length $n$
consisting of $z$ phrases, we can construct a straight-line grammar of size
$O(z\log(n/z))$~\cite{Charikar,Rytter} that generates the given text. Therefore, we might hope to efficiently process
compressed texts by working with their straight-line grammars.

However, operating on such
grammars is nontrivial. Even the seemingly simple task of checking if two non-terminals derive
the same string is challenging, as the derived strings might have exponential
length in terms of the size of the grammar $g$. Using combinatorial properties of strings,
Hirshfeld et al.~\cite{Hirshfeld} and Plandowski~\cite{Plandowski} independently showed
how to check in polynomial time whether two straight-line grammars describe the same string.
The time complexity of both implementations is rather high (for example, although not stated
explicitly in~\cite{Plandowski}, it seems to be $O(g^4)$). This was recently improved
by Jeż to $O(g\log n)$, where $n$ is the length of the described string.

This problem can be easily solved with a data structure for maintaining a family of persistent strings
to obtain the same $O(g\log n)$ complexity for 
preprocessing a straight-line program
for checking if two non-terminals derive the same string (and if not, computing the longest common prefix of the two derived strings) in constant time, assuming
constant-time arithmetic on the lengths of the generated strings.
While matching the best previously known time bound, our solution has the advantage of being dynamic:
we are able to efficiently add new non-terminals to the grammar and ask queries about
already existing non-terminals.
We believe that our results have further consequences for algorithmics of straight-line grammars.
See~\cite{LohreySurvey} for a survey of the area.

\paragraph{Pattern Matching in Dynamic Texts}
Finding all $occ$ occurrences of a length $\ell$ pattern in a static text can be done in $O(\ell+occ)$ time
using suffix trees, which can be constructed in linear time~\cite{McCreight:1976,SuffixTree}.
Suffix trees can be made partially dynamic by allowing prepending or appending single characters
to the text~\cite{onlinesuffix}. In the fully dynamic setting, where insertions or deletions can
occur in the middle of the text, the problem becomes much more difficult.
However, if we are only interested in inserting or deleting single characters in $O(\log n)$ time,
queries can be implemented in $O(\ell+ occ \log m + m \log \ell)$ time~\cite{Gu:1994},
where $m$ is again the number of operations executed so far.
Ferragina~\cite{Ferragina:1997} has shown how to handle insertions and deletions of blocks of text with the query
time roughly proportional to the number of updates made so far. This was soon improved to support
updates in $O(\sqrt{n})$ time and queries in $O(\ell+occ)$ time~\cite{DBLP:journals/siamcomp/FerraginaG98}.
The first polylogarithmic-time data structure for this problem was presented by Sahinalp and Vishkin~\cite{Sahinalp:1996}, who achieved $O(\log^3 n)$ time for an update and optimal
$O(\ell+occ)$ time for a query. Later, the update time was improved to
$O(\log^2n \log \log n \log^* n)$ at the expense of increasing query time to
$O(\ell+ occ + \log n \log \log n)$ by Alstrup et al.~\cite{Alstrup}. By building on our structure
for maintaining a collection of dynamic strings, we are able to further improve the update time
to $O(\log^2n)$ (we remove the $O(\log \log n \log^{*} n)$ factor) with optimal query time (we remove the $\log n \log \log n$ additive term).
We also extend pattern matching to the persistent setting, in which case the update times are preserved
and the query time becomes $O(\ell+\log^2 n + occ\log n)$.

\paragraph{Pattern Matching in the History of an Edited Text}
We consider a problem of maintaining a text that is being edited.
At each step either a character is added, some block of text is removed or moved to a different location (cut-paste).
We develop a data structure that can be updated in $O(\log t \log \log t)$ time with high probability and
supports pattern matching queries.
Such a query locates and reports first $occ$ occurrences of a length-$\ell$ pattern in the whole history of a text in chronological
order in $O(\ell \log t \log \log t + occ\log t)$ time.
To the best of our knowledge, we are the first to consider this natural problem.

\paragraph{Dynamic Tree Isomorphism}
As shown in~\cite{treei}, a data structure for maintaining a family of sequences can also be used for solving dynamic tree isomorphism problem.
In this problem, a family of trees is maintained and can be updated by adding/removing edges and adding vertices.
Moreover, each two trees can be tested for being isomorphic.
The result of~\cite{treei} can be immediately improved with the data structure of
Alstrup et al.~\cite{Alstrup}, and our result improves it further by a $O(\log^{*} n)$ factor
to decrease the update time to $O(\log^2 n)$ with high probability, where $n$ is the total number of vertices in all trees, while keeping the query time constant.

\subsection{Related Techniques}
Our structure for dynamic string equality is based on maintaining a hierarchical representation
of the strings, similarly to the previous works~\cite{Mehlhorn,Alstrup}.
In such an approach the representation should be, in a certain sense, locally consistent, meaning that
two equal strings have identical representations and the representations of two strings
can be joined to form the representation of their concatenation at a relatively small cost.
The process of creating such a representation can be imagined as parsing: breaking the string into blocks,
replacing every block by a new character, and repeating the process on the new shorter
string.

Deciding how to partition the string into blocks is very closely related to the
list ranking problem, where we are given a linked list of $n$ processors and every processor wants
to compute its position on the list. This requires resolving contention, that
is, choosing a large independent subset of the processors. A simple approach,
called the random mate algorithm~\cite{TreeContraction,Vishkin}, is to give a random bit
to every processor and select processors having bit $0$ such that their successor has bit $1$.
A more complicated (and slower by a factor of $O(\log^*n)$) deterministic solution, called deterministic
coin tossing, was given by Cole and Vishkin~\cite{DeterministicTossing}. Such symmetry-breaking
method is the gist of all previous solutions for dynamic string equality. Mehlhorn et al.~\cite{Mehlhorn}
used a randomized approach (different than random mate) and then applied deterministic coin tossing
to develop the deterministic version. Then, Alstrup et al.~\cite{Alstrup} further optimized the
deterministic solution.

The strategy of breaking a string into blocks and repeating on the shorter string has been
recently very successfully used by Jeż, who calls it the recompression method~\cite{JezRecompression},
to develop more efficient and simpler algorithms for a number of problems on compressed strings and for solving
equations on words. In particular, a straightforward consequence of his fully compressed pattern
matching algorithm~\cite{JezFully} is an algorithm for checking if two straight-line grammars
of total size $g$ describe the same string of length $n$ in time $O(g\log n)$.
However, he considers only static problems, which allows him to appropriately choose the partition 
by looking at the whole input.

Very recently, Nishimoto et al.~\cite{DBLP:journals/corr/NishimotoIIBT15} further extended
some of the ideas of Alstrup et al.~\cite{Alstrup} to design a space-efficient dynamic index.
They show how to maintain a string in space roughly proportional to its Lempel-Ziv factorization,
while allowing pattern matching queries and inserting/deleting substrings in time roughly proportional
to their lengths.

To obtain an improvement on the
work of Alstrup et al.~\cite{Alstrup}, we take a step back and notice that while partitioning
the strings into blocks is done deterministically, obtaining an efficient implementation
requires hash tables, so the final solution is randomized anyway. This suggests that, possibly, it does
not make much sense to use the more complicated deterministic coin tossing, and applying
the random mate paradigm might result in a faster solution. We show
that this is indeed the case, and the obtained structure is not only faster, but also conceptually
simpler (although an efficient implementation requires a deeper insight and solving a few new
challenges, see \cref{sec:comp} for a detailed discussion).

\subsection{Organization of the Paper}
In \cref{sec:preliminaries} we introduce basic notations.
Then, in \cref{sec:collection} we present the main ideas of our dynamic string collections
and we sketch how they can be used to give a data structure that supports equality tests in $O(1)$ time.
We also discuss differences between our approach and the previous ones (in \cref{sec:comp}).
The details of our implementation are given in the following three sections:
\cref{sec:pointers} describes iterators for traversing parse trees of grammar-compressed strings.
Then, in~\cref{sec:contexti} we prove some combinatorial facts concerning context-insensitive decompositions.
Combined with the results of previous sections, these facts let us handle equality tests and updates on the dynamic string collection in a clean and uniform way
in \cref{sec:adding}.

Next, in~\cref{sec:lb} we provide a lower bound for maintaining a family of non-persistent strings.
Then, in~\cref{sec:om}, we show how to extend the data structure developed in~\cref{sec:collection,sec:pointers,sec:contexti,sec:adding} in order to support lexicographic comparisons and longest common prefix queries.
\cref{sec:anchored} introduces some basic tools that are related to answering pattern-matching queries in the following three sections.
In \cref{sec:pm} we show that the our dynamic string collection can be extended with pattern matching queries against all strings in the data structure.
However, this comes at a cost of making the data structure non-persistent.
In \cref{sec:ppm} we address this issue presenting a persistent variant with an extra additive term in the query time.
Finally, in \cref{sec:timeline} we show how to use our data structures to support pattern matching queries in all versions of an edited text.

We conclude with content deferred from \cref{sec:preliminaries,sec:collection}: 
discussion of technical issues related to the model of computation in \cref{app:model} and a rigorous formalization of our main concepts in \cref{app:formalsymbols}.
\section{Preliminaries}\label{sec:preliminaries}
Let $\Sigma$ be a finite set of \emph{letters} that we call an \emph{alphabet}.
We denote by $\Sigma^{*}$ a set of all finite \emph{strings} over $\Sigma$,
and by $\Sigma^+$ all the non-empty strings.
Our data structure maintains a family of strings over some integer alphabet $\Sigma$.
Internally, it also uses strings over a larger (infinite) alphabet $\symbols \supseteq \Sigma$.
We say that each element of $\symbols$ is a \emph{symbol}.
In the following we assume that each string is over $\symbols$, but the exact description of the set $\symbols$ will be given later.

Let $\str = \str_1 \cdots \str_k$ be a string (throughout the paper we assume that the string indexing is 1-based).
For $1\le a \le b \le |\str|$ a word $u=\str_a \cdots \str_b$ is called a \emph{substring} of $\str$.
By $\str[a..b]$ we denote the \emph{occurrence} of $u$ in $\str$ at position $a$, called a \emph{fragment} of $\str$.
We use shorter notation $\str[a]$, $\str[..b]$, and $\str[a..]$, to denote $\str[a..a]$, $\str[1..b]$, and $\str[a..|\str|]$, respectively.
Additionally, we slightly abuse notation and use $\str[a+1..a]$ for $0 \le a \le |\str|$ to represent empty fragments. 

We say that a string $\str$ is \emph{1-repetition-free} if its every two consecutive symbols are distinct.
The concatenation of two strings $\str_1$ and $\str_2$ is denoted by $\str_1 \cdot \str_2$ or simply $\str_1 \str_2$.
For a symbol $x$ and $k \in \mathbb{Z}_{+}$, $x^k$ denotes the string of length $k$ whose every element is equal to $x$.
We denote by $\str^R$ the reverse string $w_kw_{k-1}\cdots w_1$.



To compute the \emph{run-length encoding} of a string $\str$, we compress maximal substrings that consist of equal symbols.
Namely, we divide $\str$ string into \emph{blocks}, where each block is a maximal substring of the form $a^k$ (i.e., no two adjacent blocks can be merged).
Then, each block $a^k$, where $k > 1$, is replaced by a pair $(a, k)$.
The blocks with $k=1$ are kept intact.
For example, the run-length encoding of $aabaaacc$ is $(a,2)b(a,3)(c,2)$.

We use the word RAM model with multiplication and randomization, assuming that the machine word has $\mword$ bits.
Some of our algorithms rely on certain numbers fitting in a machine word, which is a~weak
assumption because a word RAM machine with word size $\mword$ can simulate a word RAM machine
with word size $O(\mword)$ with constant factor slowdown.
A more detailed discussion of our model of computation is given in~\cref{app:model}.

\section{Overview of Our Result}\label{sec:collection}
In this section we show the main ideas behind the construction of our data structure
for maintaining a family of strings that can be tested for equality.
The full details are then provided in~\cref{sec:pointers,sec:contexti,sec:adding}.
Our data structure maintains a dynamic collection $\coll$ of non-empty strings.
The collection is initially empty and then can be extended using the following updates:
\begin{itemize}
  \item $\makeop(\str)$, for $\str\in \Sigma^+$, results in $\coll:=\coll\cup \{\str\}$.
   \item $\concop(\str_1,\str_2)$, for $\str_1,\str_2\in \coll$, results in $\coll:=\coll\cup \{\str_1\str_2\}$.
  \item $\splitop(\str,k)$, for $\str\in \coll$ and $1\leq k<|\str|$, results in
    $\coll:=\coll\cup \{\str[..k], \str[(k + 1)..]\}$.
\end{itemize}
Each string in $\coll$ has an integer \emph{handle} assigned by the update which created it.
This way the arguments to $\concop$ and $\splitop$ operations have constant size.
Moreover, we make sure that if an update creates a string which is already present in $\coll$, the original handle is returned.
Otherwise, a fresh handle is assigned to the string. We assume that handles are consecutive integers starting from $0$
to avoid non-deterministic output.
Note that in this setting there is no need for an extra $\eqop$ operation, as it can be implemented by
testing the equality of handles.

The w.h.p. running time can be bounded by $O(|\str|+\log n)$ for $\makeop$ and $O(\log n)$ for $\splitop$ and $\concop$,
where $n$ is the total length of strings in $\coll$; see~\cref{thm:data_structure} for details.

\subsection{Single String Representation}\label{sec:single_string}
In order to represent a string, we build a straight-line grammar that generates it.
\begin{definition}
A context-free grammar is a \emph{straight-line grammar} (a straight-line program) if for every non-terminal $S$
there is a single production rule $S \rightarrow \str$ (where $\str$ is a string of symbols),
and each non-terminal can be assigned a level in such a way that if $S \rightarrow \str$, then the levels of non-terminals in $\str$ are smaller than the level of $S$.
\end{definition}
It follows that a straight-line grammar generates a single string.
Let us describe how we build a grammar representing a given string $\str$.
This process uses two functions that we now introduce.

We first define a function \rle{} based on the run-length encoding of a string $\str$.
To compute $\rle(\str)$, we divide $\str$ into maximal blocks consisting of adjacent equal symbols and replace each block of size at least two with a new non-terminal. Multiple blocks corresponding to the same string are replaced by the same non-terminal.
Note that the result of $\rle{}$ is 1-repetition-free.

The second function used in building the representation of a string is $\compress_i{}$, where $i \geq 1$ is a parameter.
This function also takes a string $\str$ and produces another string.
If $\str$ is 1-repetition-free, then the resulting string is constant factor shorter (in expectation).

Let us now describe the details.
We first define a family of functions $\hs_i$ ($i\ge 1$) each of which uniformly at random assigns $0$ or $1$ to every possible symbol.
To compute $\compress_i(\str)$, we define blocks in the following way.
If there are two adjacent symbols $a_j a_{j+1}$ such that $\hs_i(a_j) = 0$ and $\hs_i(a_{j+1}) = 1$, we mark them as a block.
All other symbols form single-symbol blocks.
Then we proceed exactly like in $\rle$ function.
Each block of length two is replaced by a new non-terminal.
Again, equal blocks are replaced by the same non-terminal.

Note that the blocks that we mark are non-overlapping.
For example, consider $\str = cbabcba$ and assume that $\hs_1(a) = 1$, $\hs_1(b) = 0$, and $\hs_1(c) = 1$.
Then the division into blocks is $c|ba|bc|ba$, and $\compress_1(cbabcba) = cS_1S_2S_1$.
In a 1-repetition-free string, each two adjacent symbols are different, so they form a block with probability $\frac14$.
Thus, we obtain the following:
\begin{fact}\label{lem:compress}
If $\str$ is a 1-repetition-free string, then $\mathbb{E}(|\compress_i(\str)|) \leq \frac{3}{4} |\str| + \frac14$ for every $i \geq 1$.
\end{fact}

In order to build a representation of a string $\str$, we repeatedly apply $\rle{}$ and $\compress_i{}$, until a string of length one is obtained.
We define:
\begin{align*}
\shrink{i}(\str) & := \begin{cases}
\rle(\str) & \text{if $i$ is odd} \\
\compress_{i/2}(\str) & \text{if $i$ is even}
\end{cases}\\
\cshrink{0}(\str) & := \str\\
\cshrink{i+1}(\str) & := \shrink{i+1}(\cshrink{i}(\str))
\end{align*}
The \emph{depth} of the representation of $\str$, denoted by $\depth(\str)$, is the smallest $i$, such that $|\cshrink{i}(w)| = 1$.
We say that the representation contains $\depth(\str)+1$ \emph{levels} numbered $0, 1, \ldots, \depth(\str)$.
Observe that $\depth(\str)$ is actually a random variable, as its value depends on the choice of random bits $\hs_i(\cdot)$.
Using \cref{lem:compress}, we obtain the following upper bound on the depth:

\begin{restatable}{lemma}{smalldepth}\label{lem:small_depth}
If $\str$ is a string of length $n$ and $r\in \mathbb{R}_{\ge 0}$, then
$\mathbb{P}(\depth(\str) \leq 8(r + \ln n)) \geq 1 - e^{-r}.$
\end{restatable}
In the proof we use the following theorem.

\begin{theorem}[Theorem 7 in \cite{Lehre:2013}]\label{thm:drift}
Let $(X_t)_{t \geq 0}$ be a stochastic process over some state space $S = \{0\} \cup [x_{min}, x_{max}]$, where $x_{min} > 0$.
Suppose that there exists some $\delta$, where $0 < \delta < 1$, such that $\mathbb{E}(X_t - X_{t+1} \mid (X_0, \ldots, X_t)) \geq \delta X_t$.
Let $T:= \min\{t \mid X_t = 0\}$ be the first hitting time.
Then $$\mathbb{P}\left(T \geq \frac{\ln (X_0 / x_{min}) + r}{\delta}\right) \leq e^{-r}$$ for all $r > 0$.
\end{theorem}

\begin{proof}[Proof of \cref{lem:small_depth}]
Let $X_i$ be a random variable defined as $X_i := |\cshrink{2i}(\str)|-1$.
We have that $X_0 = |\cshrink{0}(\str)|-1 = |\str|-1 = n-1$.

By \cref{lem:compress}, $\mathbb{E}(|\cshrink{2i+2}(\str)|) \leq \frac34|\cshrink{2i+1}| + \frac14$.
Clearly, $|\rle(\str)| \leq |\str|$, which implies $|\cshrink{2i+2}(\str)| \leq \frac34|\cshrink{2i}| + \frac14$.
Thus, $\mathbb{E}(X_{i+1}) \leq \frac34 \mathbb{E}(X_i)$.
Since every call to $\compress_i$ uses different random bits $\hs_i(\cdot)$, distinct calls are independent.
Thus,  we get that $\mathbb{E}(X_t - X_{t+1} \mid (X_0, \ldots, X_t)) = \mathbb{E}(X_t - X_{t+1}) \geq \frac14 X_t$.
By applying \cref{thm:drift}, for every $r > 0$ we get
$$\mathbb{P}(T \geq 4(r + \ln n)) \leq e^{-r}.$$
Here, $T$ is the smallest index, such that $X_T = 0$, which means $|\cshrink{2T}(\str)| = 1$.
Hence, $\depth(\str) \leq 2T$.
The Lemma follows.
\end{proof}

It follows that in a collection $\coll$ of $\poly(n)$ strings of length $\poly(n)$, we have $\depth(\str) = O(\log n)$ 
for every $w\in \coll$ with high probability.

Finally, we describe how to build a straight-line grammar generating the string $\str$.
Observe that every call to $\shrink{i}$ replaces some substrings by non-terminals.
Each time a string $\strb$ is replaced by a non-terminal $S$, we add a production $S \rightarrow \strb$ to the grammar.
Once we compute a string $\cshrink{i}(\str)$ that has length $1$, the only symbol of this string is the start symbol of the grammar.

\subsection{Representing Multiple Strings}\label{sec:persistent_ds}
We would now like to use the scheme from previous section to build a representation of a collection of strings $\coll$.
Roughly speaking, we build a union $\grammar$ of the straight-line grammars representing individual strings.
Formally, $\grammar$ is not a grammar because it does not have a start symbol, but it becomes a straight-line grammar once
an arbitrary non-terminal is chosen a start symbol.
Still, we abuse the naming slightly and call $\grammar$ a grammar.

\begin{shortv}
We say that two production rules $S_1 \rightarrow \str_1$ and $S_2 \rightarrow \str_2$ are \emph{equivalent} if $\str_1 = \str_2$.
We shall make sure that $\grammar$ does not contain two distinct equivalent production rules.
Actually, we even ensure that no two symbols of $\grammar$ generate the same string.

A call to $\shrink{i}(\str)$ produces a string, in which some substrings (blocks) in $\str$ are replaced by single non-terminals.
We say that each such replacement is a \emph{collapse}.
In the process of building the representation of a single string, for each collapse we add production rules to $\grammar$.
We also assure that equal blocks are collapsed to the same non-terminal.

To build a representation of a collection of strings, we use the same approach globally for all strings.
Hence, we have a global $\grammar$ that contains a list of all previously generated production rules.
Whenever some block is collapsed and we are to add a new production rule $S \rightarrow \str$, we first check if $\grammar$ already contains an equivalent production rule $S' \rightarrow \str$.
If this is the case, the function replaces $w$ with $S'$ and does not update $\grammar$.
Otherwise, the new production rule $S \rightarrow \str$ is added to $\grammar$, and the substring $\str$ is replaced by $S$.
For simplicity, throughout this paper we assume that there is a single grammar $\grammar$ that is being built and it can be accessed from any function that we describe.
\end{shortv}

As the collection of strings $\coll$ is growing over time, we need to extend the scheme from \cref{sec:single_string} to support multiple strings in such a dynamic setting.
Since our construction is randomized, some care needs to be taken to formalize this process properly.
This is done in full details in \cref{app:formalsymbols}.
In this section we provide a more intuitive description.

Our first goal is to assure that if a particular string is generated multiple times by the same instance of the data structure,
its representation is the same each time.
In order to do that, 
we ensure that if a block of symbols $w'$ is replaced by a non-terminal during the construction of the representation of some string of $\coll$, 
a new non-terminal is created only when the block $w'$ is encountered for the first time.
To this end, we use a dynamic dictionary (implemented as a hash table) that maps already encountered blocks to the non-terminals.
Combined with our definition of $\shrink{i}$, this suffices to ensure that each non-terminal in the grammar generates a distinct string.

To keep this overview simple, we assume that the non-terminals come from an infinite set $\symbols$.
The string generated by $S\in\symbols$ is denoted by $\sstr(S)$.
The set $\symbols$ is formally defined so that it is independent from the random choices of the $h_i$ functions,
which leads to some symbols being inconsistent with those choices. 
In particular many symbols may generate the same string, but exactly one of them is consistent. 
Hence, each string $\str\in\Sigma^+$ can be assigned a unique symbol $\sstr^{-1}(\str)$ that generates $\str$.
Other symbols do not appear in the grammar we construct, i.e., $\grammar \subseteq \sstr^{-1}(\Sigma^+)\subseteq\symbols$.

We say that $\grammar$ \emph{represents} a string $w$ if $\sstr^{-1}(w) \in \grammar$.
A grammar $\grammar$ is called \emph{closed} if all symbols appearing on the right-hand side of a production rule
also appear on the left-hand side.
Our data structure maintains a grammar $\grammar$ which satisfies the following invariant:
\begin{invariant}\label{inv:consistency}
$\grammar$ is closed and represents each $w\in \coll$.
\end{invariant}

Following earlier works~\cite{Alstrup, Mehlhorn}, our data structure operates on signatures which are just integer identifiers of symbols.
Whenever it encounters a symbol $S$ missing a signature (a newly created non-terminal or a terminal encountered for the first time), it assigns a fresh signature.
At the same time it reveals bits $\hs_i(S)$ for $0< i \le \mword$.\footnote{The data structure never uses bits $\hs_i(S)$ for $i>\mword$. To be precise, it fails before they ever become necessary.}
Since the bits are uniformly random, this can be simulated by taking a random machine word.

We use $\sigs(\grammar)$ to denote the set of signatures corresponding to the symbols
that are in the grammar $\grammar$.
We distinguish three types of signatures.
We write $s\to a$ for $a\in \Sigma$ if $s$ represents a terminal symbol, %
$\spair{s}{s_1}{s_2}$ to denote that $s$ represents
a non-terminal $S$ which replaced two symbols $S_1,S_2$
(represented by signatures $s_1,s_2$, correspondingly)
in $\compress_{i}$.
Analogously, we write
$\spower{s}{s_1}{k}$ to denote that $s$ represents a non-terminal introduced by $\rle$ when
$k$ copies of a symbol $s_1$ were replaced represented by $s$.
We also set $\sstr(s)=\sstr(S)$.
\begin{shortv}
$\spair{s}{s_1}{s_2}$ if $s$ is a signature created in $\compress_{i}$, when a substring  $s_1s_2$ was collapsed,
and  $\spower{s}{s_1}{k}$ if $s$ is a signature created in $\rle$ when a substring $s_1^k$ was collapsed.
The \emph{level} of a signature $s$ is the integer $\slev(s)$, such that $s$ is created when some substring is collapsed in $\cshrink{\slev(s)}$.
Note that this value is uniquely defined for each signature.
\end{shortv}

For each signature $s\in \sigs(\grammar)$ representing a symbol $S$,
we store a record containing some attributes of the associated symbol.
More precisely, it contains the following information, which can be stored in $O(1)$ space
and generated while we assign a signature to the symbol it represents: (a) The associated production rule $\spair{s}{s_1}{s_2}$, $\spower{s}{s_1}{k}$ (if it represents a non-terminal) or the terminal symbol $a\in \Sigma$ if $s \to a$; (b) The length $\slength(s)$ of the string generated from $S$; (c) The level $\slev(s)$ of $s$ (which is the level $l$, such that $s$ could be created in the call to $\shrink{l}$); (d) The random bits $\hs_i(S)$ for $0 < i \le \mword$.

In order to ensure that a single symbol has a single signature,
we store three dictionaries which for each $s\in \sigs(\grammar)$ map $a$ to $s$ whenever $s\to a$,  $(s_1,s_2)$ to $s$ whenever $\spair{s}{s_1}{s_2}$, and $(s_1,k)$ to $s$ whenever $\spower{s}{s_1}{k}$.
Thus, from now on we use signatures instead of the corresponding symbols.

Finally, we observe that a signature representing $\sstr^{-1}(w)$ can be used to identify $w$ represented by $\grammar$.
In particular, this enables us to test for equality in a trivial way.
However, in order to make sure that updates return deterministic values, we use consecutive integers as handles to elements of $\coll$
and we explicitly store the mappings between signature and handles.

\subsection{Parse Trees}\label{sec:pt_overview}
We define two trees associated with a signature $s \in \sigs(\grammar)$ that represents the string $\str$.
These trees are not stored explicitly, but are very useful in the description of our data structure.
We call the first one of them the \emph{uncompressed parse tree} and denote it by $\ustree(s)$.
When we are not interested in the signature itself, we sometimes call the tree $\ustree(\str)$ instead of $\ustree(s)$.

Consider the strings $\cshrink{0}(\str), \cshrink{1}(\str), \ldots, \cshrink{\slev(s)}(\str)$ stacked on top of each other, in such a way that $\cshrink{0}(\str)$ is at the bottom.
We define $\ustree(s)$ to be a tree that contains a single node for each symbol of each of these strings.
Each symbol of $\cshrink{i+1}(\str)$ originates from one or more symbols of $\cshrink{i}(\str)$ and this defines how the nodes in $\ustree(s)$ are connected: a node representing some symbol is a parent of all nodes representing the symbols it originates from.
The tree $\ustree(s)$ is rooted with the root being the node representing the only symbol in $\cshrink{\slev(s)}(\str)$,
whose signature is $s$.
Moreover, the children of every node are ordered left-to-right in the natural way.
If $v$ is a node in $\ustree(s)$ distinct from the root, we define $\uspar(v)$ to be the parent node of $v$ in $\ustree(s)$.
We denote by $\uschild(v,k)$ the $k$-th child of $v$ in $\ustree(s)$.

\begin{figure}[ht]
\hspace{20pt}
\includegraphics[scale=0.6]{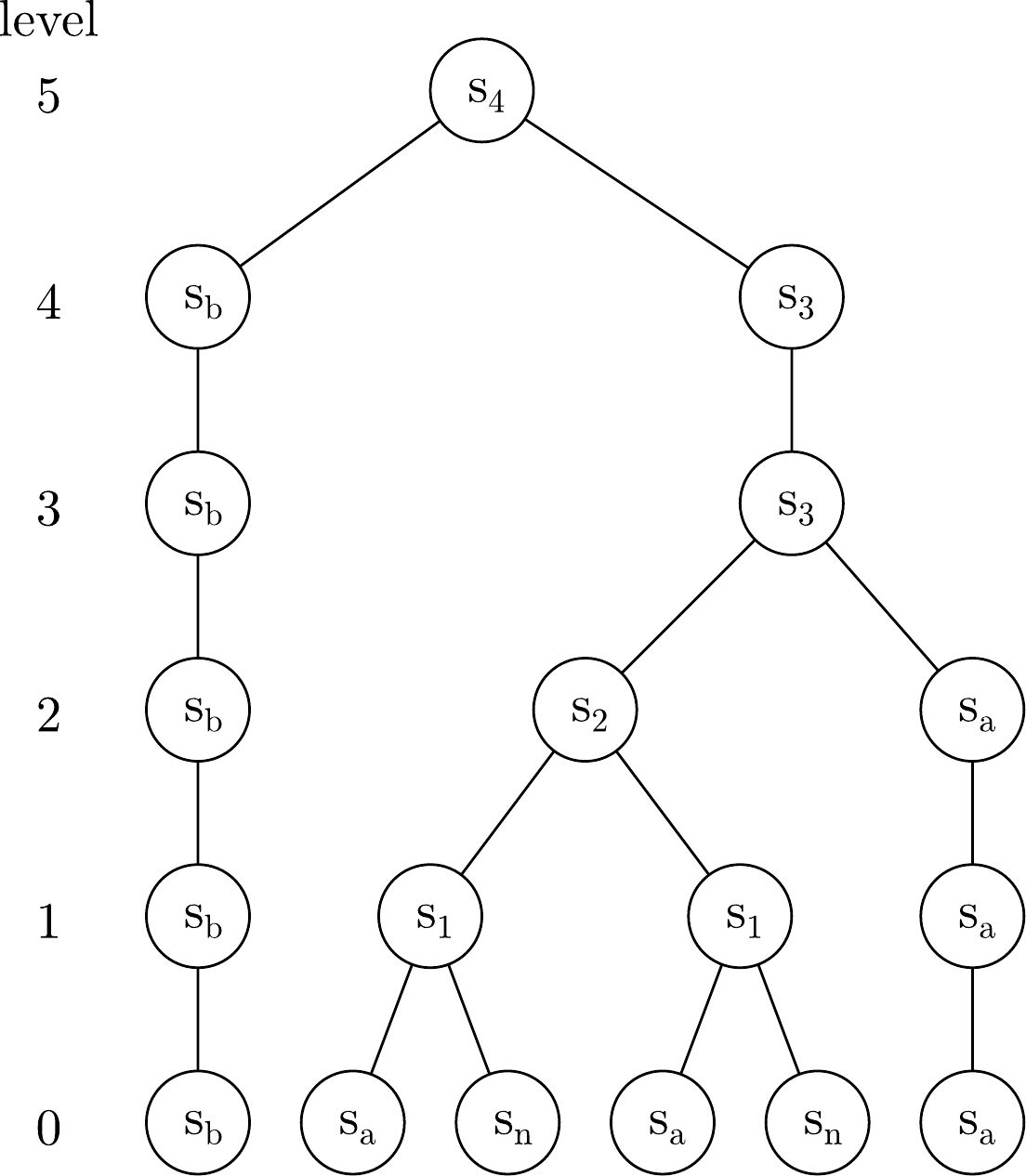}
\hfill
\includegraphics[scale=0.6]{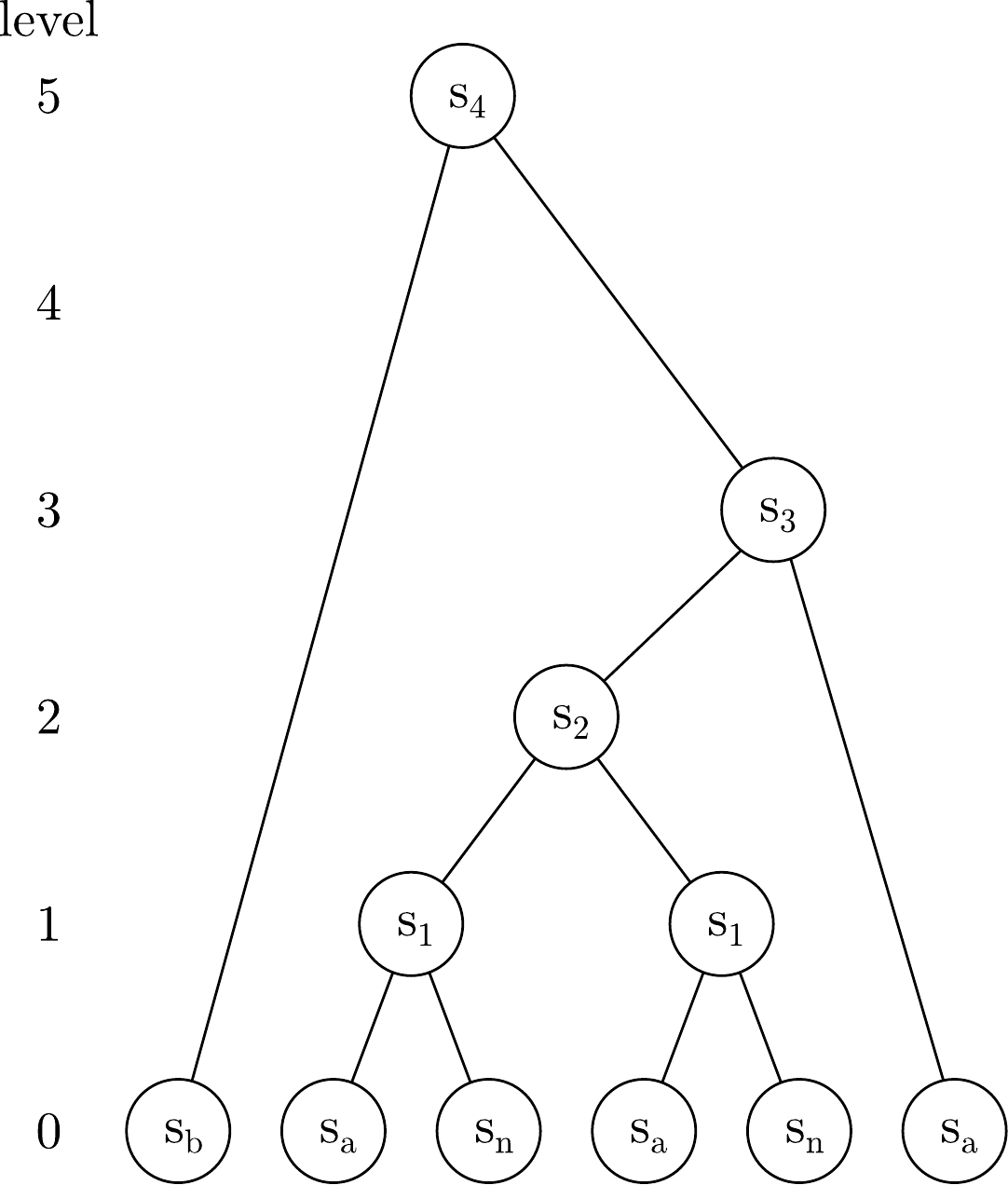}
\hspace{20pt}
\caption{Let $\str = \textrm{banana}$. Left picture shows the tree $\ustree(\str)$, whereas the right one depicts $\stree(\str)$.
The signature stored in every node is written inside the circle representing the node.
We have $\sstr(s_b) = b$, $\sstr(s_a) = a$, $\sstr(a_n) = n$, $\sstr(s_1) = an$, $\sstr(s_2) = anan$, $\sstr(s_3) = anana$, and $\sstr(s_4) = banana$.
Moreover, $s_a \rightarrow a$, $\spair{s_1}{s_a}{s_n}$, $\spower{s_2}{s_1}{2}$, $\spair{s_3}{s_2}{s_a}$ and $\spair{s_4}{s_b}{s_3}$.
Note that the parent of the rightmost leaf has level $1$ in $\ustree(\str)$, but the parent of the same leaf in $\stree(\str)$ has level $3$.
}
\label{fig:trees}
\end{figure}

Observe that $\ustree(s)$ may contain nodes with just one child.
This happens exactly when some symbol forms a single-symbol block in $\rle$ or $\compress_i$.
Hence, for a signature $s$ we also define a \emph{parse tree} $\stree(s)$, which is obtained from $\ustree(s)$ by dissolving nodes with one child (similarly as before, we sometimes use the notation $\stree(\str)$ to denote the same tree).
Observe that $\stree(s)$ is indeed a parse tree, as every node corresponds to an application of a production rule from $\grammar$.
In particular, a node $v$ with signature $s'$ has exactly two children if $\spair{s'}{s_1}{s_2}$ and $k$ children if $\spower{s'}{s_1}{k}$.
Thus, each internal node of $\stree(s)$ has at least two children, which implies that $\stree(s)$ has $O(|\str|)$ nodes, where $\str=\sstr(s)$.
Similarly to the uncompressed parse trees, each tree $\ustree(s)$ is also rooted and the children of each node are ordered.
Moreover, if $v$ is a node in $\stree(s)$ distinct from the root, we define $\spar{s}(v)$ to be the parent node of $v$ in $\stree(s)$.
See \cref{fig:trees} for an example.

Consider a tree $\stree(s)$ and its node $u$ representing a symbol with signature $\ussig(u)$.
Observe that if we know $\ussig(u)$, we also know the children of $u$, as they are determined by the production rule associated with $\ussig(u)$.
More generally, $\ussig(u)$ uniquely determines the subtree rooted at~$u$.
On the other hand, $\ussig(u)$ in general does not determine the parent of $u$.
This is because there may be multiple nodes with the same signature, even within a single tree.

We show how to deal with the problem of navigation in the parse trees.
Namely, we show that once we fix a tree $\stree(s)$ or $\ustree(s)$, we may traverse it efficiently, using the information stored in $\grammar$.
We stress that for that purpose we do not need to build the trees $\stree(s)$ and $\ustree(s)$ explicitly.

We use \emph{pointers} to access the nodes of $\ustree(s)$.
Assume we have a pointer $P$ to a node $v \in \ustree(s)$ that corresponds to the $i$-th symbol of $\cshrink{j}(\str)$.
As one may expect, given $P$ we may quickly get a pointer to the parent or child of $v$.
More interestingly, we can also get a pointer to the node $\usright(v)$ ($\usleft(v)$) that lies \emph{right} (or \emph{left}) of $v$ in $\ustree(s)$
in constant time.
The node $\usright(v)$ is the node that corresponds to the $(i+1)$-th symbol of $\cshrink{j}(\str)$,
while $\usleft(v)$ corresponds to the $(i-1)$-th symbol.
Note that these nodes may not have a common parent with $v$.
For example, consider $\ustree(\str)$ in \cref{fig:trees} and denote the level-0 nodes corresponding to the first three letters (b, a and n) by $v_b$, $v_a$ and $v_n$.
Then, $\usright(v_a) = v_n$, but also $\usleft(v_a) = v_b$, although $\uspar(v_a)\ne \uspar(v_b)$.
The pointers let us retrieve some information about the underlying nodes, including their signatures.
A pointer to the leftmost (rightmost) leaf of $\ustree(\str)$ can also be efficiently created.
The full interface of parse tree pointers is described in \cref{sec:pointers}.

In order to implement the pointers to $\ustree(\str)$ trees, we first introduce
pointers to compressed trees.
The set of allowed operations on these pointers is more limited, but sufficient
to implement the needed pointers to $\ustree(\str)$ trees
in a black-box fashion.
Each pointer to a tree $\stree(\str)$ is implemented as a stack that represents the path from the root of $\stree(s)$ to the pointed node.
To represent this path efficiently, we replace each subpath that repeatedly descends to the leftmost child by a single entry in the stack (and similarly for paths descending to rightmost children).
This idea can be traced back to a work of Gąsieniec et al.~\cite{DBLP:conf/dcc/GasieniecKPS05}, who implemented constant-time
forward iterators on grammar-compressed strings.
By using a fully-persistent stack, we are able to create a pointer to a neighbor node in constant time, without destroying the original pointer.
However, in order to do this efficiently, we need to quickly handle queries about the leftmost (or rightmost) descendant of a node $v\in\stree(\str)$
at a given depth, intermixed with insertions of new signatures into $\sigs(\grammar)$.
To achieve that, we use a data structure for finding level ancestors in dynamic trees~\cite{Alstrup:2000}.
The details of this construction are provided in~\cref{sec:pointers}.

\subsection{Comparison to Previous Approaches}\label{sec:comp}

Let us compare our construction with the approaches previously used for similar problems in~\cite{Mehlhorn} and~\cite{Alstrup}.
\cite{Mehlhorn} describes two: a deterministic one (also used in~\cite{Alstrup}) and a randomized one.

Our algorithm for $\shrink{i}$ is simpler than the corresponding procedure in~\cite{Mehlhorn}.
In particular, we may determine if there is a block boundary between two symbols, just by inspecting their values.
In the deterministic construction of~\cite{Mehlhorn}, this requires inspecting $\Theta(\log^{*} n)$ surrounding symbols.

However, the use of randomness in our construction poses some serious challenges,
mostly because the size of the uncompressed parse tree $\ustree(\str)$ can be $\Omega(n\log n)$ for a string $\str$ of length $n$ with non-negligible probability.
Consider, for example, the string $\str=(ab)^{n/2}$. With probability at least $2^{-k}$ we have $h_i(a)=h_i(b)=1$
for $0 \le i < k/2$, and consequently $\ustree(\str)$ contains at least $|\str|\cdot k$ nodes. Hence, implementing $\makeop(\str)$
in time proportional to $|\str|$ requires more care.

Another problem is that even prepending a single letter $a$ to a string $\str$ results in a string $a\cdot \str$ 
whose uncompressed parse tree $\ustree(a\cdot w)$ might differ from $\ustree(w)$ by $\Omega(\log^2n)$ nodes (with non-negligible probability).
In the sample string considered above, the strings $\cshrink{i}(a \cdot \str)$ and $\cshrink{i}(a \cdot \str)$
differ by a prefix of length $\Omega(i)$ for $0\le i < k$ with probability $4^{-k}$.
In the deterministic construction of~\cite{Mehlhorn}, the corresponding strings may only differ by $O(\log^{*} n)$ symbols.
In the randomized construction of~\cite{Mehlhorn}, the difference is of constant size.
As a result, in~\cite{Mehlhorn,Alstrup}, when the string $a \cdot \str$ is added to the data structure, the modified prefixes of (the counterpart of) $\cshrink{i}(a \cdot \str)$ can be computed explicitly, which is not feasible in our case. 

To address the first problem, our grammar is based on the compressed parse trees $\stree(\str)$ and we operate on the uncompressed trees $\ustree(\str)$
only using constant-time operations on the pointers. 
In order to keep the implementation of $\splitop$ and $\concop$ simple despite the second issue,
we additionally introduce the notion of context-insensitive decomposition, which captures the part of $\ustree(\str)$ which is preserved in the parse tree of every superstring of $\str$ (including $a \cdot \str$).
A related concept (of \emph{common sequences} or \emph{cores}) appeared in a report by Sahinalp and Vishkin~\cite{SV1995} and was later used in several
papers including a recent a work by Nishimoto et al.~\cite{DBLP:journals/corr/NishimotoIIBT15}.

Consequently, while the definition of $\shrink{i}$ is simpler compared to  the previous constructions,
and the general idea for supporting $\splitop$ and $\concop$ is similar,
obtaining the desired running time requires deeper insight and more advanced data-structural tools.

\newcommand{\itstring}{\mathit{pref}}
\newcommand{\itsuf}{\mathit{suff}}

\subsection{Context-Insensitive Nodes}\label{sec:context_short}
In this section we introduce the notion of context-insensitive nodes of $\ustree(w)$ to express
the fact that a significant part of $\ustree(w)$ is ``preserved'' in the uncompressed parse trees of superstrings of $w$.
For example, if we concatenate two strings $\str$ and $\strb$ from $\coll$, it suffices to take care of the part of  $\ustree(\str \cdot \strb)$
which is not ``present'' in $\ustree(w)$ or $\ustree(y)$. We formalize this intuition as follows.

Consider a node $v$ of $\ustree(w)$, which represents a particular fragment $w[i..j]$ of $w$.
For every fixed strings $x,y$ this fragment can be naturally identified with a fragment of the concatenation $x\edot w\edot y$.
If $\ustree(xwy)$ contains a node representing that fragment, we say that $v$ is \emph{preserved} in the \emph{extension} $x\edot w \edot y$.
If $v$ is preserved for every such extension, we call it \emph{context-insensitive}.
A weaker notion of \emph{left} and \emph{right} context-insensitivity is defined to impose a node
to be preserved in  all \emph{left} and all \emph{right extensions},  i.e., extensions with $y=\eps$ and $x=\eps$, respectively.

The following lemma captures some of the most important properties of context-insensitive nodes.
Its proof can be found in \cref{sec:contexti}, where a slightly stronger result appears as \cref{cor:context_insensitive}.

\begin{lemma}\label{lem:context_insensitive} Let $v$ be a node of $\ustree(w)$.
  If $v$ is right context-insensitive, so are nodes $\usleft(v)$, $\usleft(\uspar(v))$, and all children of~$v$.
  Symmetrically, if $v$ is left context-insensitive, so are nodes $\usright(v)$, $\usright(\uspar(v))$, and all children of~$v$.
Moreover, if $v$ is both left context-insensitive and right context-insensitive, then it is context-insensitive.
\end{lemma}

We say that a collection $L$ of nodes in $\ustree(w)$ or in $\stree(w)$ forms a \emph{layer} if every leaf has exactly one ancestor in $L$.
The natural left-to-right order on $\ustree(w)$ lets us treat every layer as a sequence of nodes.
The sequence of their signatures is called the \emph{decomposition} corresponding to the layer.
Note that a single decomposition may correspond to several layers in $\ustree(w)$, but only one of them (the lowest one) does not contain nodes with exactly one child.

If a layer in $\ustree(w)$ is composed of (left/right) context-insensitive nodes,
we also call the corresponding decomposition (left/right) context-insensitive.
The following fact relates context-insensitivity with concatenation of strings and their decompositions.
It is proven in~\cref{sec:contexti}.

\begin{restatable}{fact}{twodecompositions}\label{fact:two-decompositions}
Let $D$ be a right context-insensitive decomposition of $w$ and let $D'$ be a left-context insensitive decomposition of $w'$.
The concatenation $D\cdot D'$ is a decomposition of $w w'$. Moreover, if $D$ and $D'$ are context-insensitive, then $D\cdot D'$
is also context-insensitive.
\end{restatable}

The context-insensitive decomposition of a string may need to be linear in its size.
For example, consider a string $a^k$ for $k > 1$.
We have $|\cshrink{1}(a^k)| = 1$.
Thus, $\ustree(a^k)$ contains a root node with $k$ children.
At the same time, the root node is not preserved in the tree $\ustree(a^{k} \cdot a)$.
Thus, the smallest context-insensitive decomposition of $a^k$ consists of $k$ signatures.

However, as we show, each string $\str$ has a context-insensitive decomposition, whose run-length encoding has length $O(\depth(\str))$, which is $O(\log |\str|)$ with high probability.
This fact is shown in~\cref{sec:constructing_cid}.
Let us briefly discuss, how to obtain such a decomposition.
For simplicity, we present a construction of a right context-insensitive decomposition.

Consider a tree $\ustree(\str)$.
We start at the rightmost leaf $v_0$ of this tree.
This leaf is clearly context-insensitive (since all leaves are context-insensitive).
Then, we repeatedly move to $v_{i+1}=\usleft(\uspar(v_{i}))$.
From~\cref{lem:context_insensitive} it follows that in this way we iterate through right context-insensitive nodes of $\ustree(\str)$.
Moreover, nodes left of each $v_i$ are also right context-insensitive.
In order to build a decomposition, we start with an empty sequence $D$ and every time before we move from $v_{i}$ to $v_{i+1}$ we prepend $D$ with the sequence of signatures of all children of $\uspar(v_i)$ that are left of $v_i$, including the signature of the node $v_i$.
In each step we move up $\ustree(\str)$, so we make $O(\depth(\str))$ steps in total.
At the same time, $\uspar(v_i)$ either has at most two children or all its children have equal signatures.
Thus, the run-length encoding of $D$ has length $O(\depth(\str))$ and, by using tree pointers, can be computed in time that is linear in its size.
This construction can be extended to computing context-insensitive decompositions of a given substring of $\str$, which is used in the $\splitop$ operation.

\subsection{Updating the Collection}
In this section we show how to use context-insensitive decompositions in order to add new strings to the collection.
It turns out that the toolbox which we have developed allows us to handle $\makeop$, $\concop$ and $\splitop$ operations in a very simple and uniform way.

Consider a $\concop(\str_1, \str_2)$ operation.
We first compute (the run-length encodings of) context-insensitive decompositions of $\str_1$ and $\str_2$.
Denote them by $D_1$ and $D_2$.
By~\cref{fact:two-decompositions}, their concatenation $D$ is a decomposition of $\str_1 \cdot \str_2$.
Moreover, each signature in $D$ belongs to $\sigs(\grammar)$.
Let $L$ be the layer in $\stree(\str_1 \cdot \str_2)$ that corresponds to $D$.
In order to ensure that $\grammar$ represents $\str_1 \cdot \str_2$, it suffices to add to $\grammar$ the signatures of all nodes of $\stree(\str_1 \cdot \str_2)$ that lie above $L$.
This subproblem appears in $\makeop$, $\concop$ and $\splitop$ operations and we sketch its efficient solution below.

Consider an algorithm that repeatedly finds a node $v$ of the lowest level that lies above $L$ and substitutes in $L$ the children of $v$ with $v$.
This algorithm clearly iterates through all nodes of $\stree(\str_1 \cdot \str_2)$ that lie above $L$.
At the same time, it can be implemented to work using just the run-length encoding of the decomposition $D$ corresponding to $L$ instead of $L$ itself.
Moreover, its running time is only linear in the length of the run-length encoding of $D$ and the depth of the resulting tree. 
For details, see~\cref{lem:compressdec}.

Thus, given the run-length encoding of decomposition of $\str_1 \cdot \str_2$ we can use the above algorithm to add to $\grammar$ the signatures of all nodes of $\stree(\str_1 \cdot \str_2)$ that lie above $L$.
The same procedure can be used to handle a $\splitop$ operation (we compute a context-insensitive decomposition of the prefix and suffix and run the algorithm on each of them) or even a $\makeop$ operation.
In the case of $\makeop$ operation, the sequence of letters of $\str$ is a trivial decomposition of $\str$, so by using the algorithm, we can add it to the data structure in $O(|\str|+\depth(\str))$ time.
Combined with the efficient computation of the run-length encodings of context-insensitive decompositions, this gives an efficient way of handling $\makeop$, $\concop$ and $\splitop$ operations.
See~\cref{sec:conclusions} for precise theorem statements.

\section{Navigating through the Parse Trees}\label{sec:pointers}
In this section we describe the notion of \emph{pointers} to trees $\stree(s)$ and $\ustree(s)$ that has been introduced in~\cref{sec:pt_overview} in detail.
Although in the following sections we mostly make use of the pointers
to $\ustree(s)$ trees, the pointers to compressed parse trees
are essential to obtain an efficient implementation
of the pointers to uncompressed trees.

Recall that a pointer is a handle that can be used to access some node of a tree $\stree(s)$ or $\ustree(s)$ in constant time.
The pointers for trees $\stree(s)$ and $\ustree(s)$ are slightly different and we describe them separately.

The pointer to a node of $\stree(s)$ or $\ustree(s)$ can be created for any string represented by $\grammar$ and, in fact, for any existing signature $s$.
Once created, the pointer points to some node $v$ and cannot be moved.
The information on the parents of $v$ (and, in particular, the root $s$)
is maintained as a part of the internal state of the pointer.
The state can possibly be shared with other pointers.

In the following part of this section, we first describe the pointer interface.
Then, we move to the implementation of pointers to $\stree(s)$ trees.
Finally, we show that by using the pointers to $\stree(s)$ trees in a black-box fashion, we may obtain pointers to $\ustree(s)$ trees.

\subsection{The Interface of Tree Pointers}
We now describe the pointer functions we want to support.
Let us begin with pointers to the uncompressed parse trees.
Fix a signature $s\in\sigs(\grammar)$ and
let $\str=\sstr(S)$, where $s$ corresponds to the symbol $S$.
First of all, we have three primitives for creating pointers:
\begin{itemize}
  \item $\itroot(s)$ -- a pointer to the root of $\ustree(s)$.
  \item $\itbegin(s)$ ($\itend(s)$) -- a pointer to the leftmost (rightmost) leaf of $\ustree(s)$.
\end{itemize}
We also have six functions for navigating through the tree.
These functions return appropriate pointers or $\itnil$ if the respective nodes do not exist.
Let $P$ be a pointer to $v\in \ustree(s)$.
\begin{itemize}
  \item $\itparent(P)$ -- a pointer to $\uspar(v)$.
  \item $\itchild(P,k)$ -- a pointer to the $k$-th child of $v$ in $\ustree(s)$.
  \item $\itright(P)$ ($\itleft(P)$) -- a pointer to the node $\usleft(v)$ ($\usright(v)$, resp.).
  \item $\itrskip(P,k)$ ($\itlskip(P,k)$) --
    let $v=v_1,v_2,\ldots$ be the nodes to the right (left) of $v$
    in the $\itlevel(P)$-th level in $\ustree(s)$, in the left to right (right to left) order
    and let $K$ be the largest integer such that
    all nodes $v_1,\ldots,v_K$ correspond to the same signature as $v$.
    If $k\leq K$ and $v_{k+1}$ exists, return a pointer to $v_{k+1}$,
    otherwise return $\itnil$.

\end{itemize}
Note that the nodes $\usright(v)$ and $\usleft(v)$ might not have a common parent with the node $v$.

Additionally, we can retrieve information about the pointed node $v$ using several functions:
\begin{itemize}
  \item $\itsig(P)$ -- the signature corresponding to $v$, i.e., $\ussig(v)$.
  \item $\itdegree(P)$ -- the number of children of $v$.
  \item $\itindex(P)$ -- if $v \neq s$, an integer $k$ such that $P=\itchild(\itparent(P,k))$, otherwise $0$.
  \item $\itlevel(P)$ -- the level of $v$, i.e., the number of edges on the path from $v$ to the leaf of $\ustree(s)$ (note that this may not be equal to $\slev(\itsig(P))$, for example if $v$ has a single child that represents the same signature as $v$).
  \item $\itrepr(P)$ -- a pair of indices $(i,j)$ such that $v$ represents $w[i..j]$.
  \item $\itrext(P)$ ($\itlext(P)$) -- let $v=v_1,v_2,\ldots$ be the nodes to the right (left) of $v$
    in the $\itlevel(P)$-th level in $\ustree(s)$, in the left to right (right to left) order.
    Return the maximum $K$ such that
    all nodes $v_1,\ldots,v_K$ correspond to the same signature as $v$.
\end{itemize}

Moreover, the interface provides a way to check if two pointers $P$, $Q$ to
the nodes of the same tree $\ustree(s)$ point to the same node:
$P=Q$ is clearly equivalent to $\itrepr(P)=\itrepr(Q)$.

The interface of pointers for trees $\stree(s)$ is more limited.
The functions $\itroot$, $\itbegin$, $\itend$, $\itparent$, $\itchild$,
$\itsig$, $\itdegree$, $\itindex$ and $\itrepr$ are defined
analogously as for uncompressed tree pointers.
The definitions of $\itright$ and $\itleft$ are slightly different.
In order to define these functions, we introduce the notion of
the \emph{level-$l$ layer} of $\stree(s)$.
The level-$l$ layer $\stlayer_s(l)$ (also denoted by $\stlayer_s(l)$)
is a list of nodes $u$ of $\stree(s)$ such that $\slev(\ussig(u))\leq l$
and either $\ussig(u)=s$ or $\slev(\ussig(\spar{s}(u)))>l$, ordered from the leftmost to the rightmost nodes.
In other words, we consider a set of nodes of level $l$ in $\ustree(s)$ and then replace each node with its first descendant (including itself) that belongs to $\stree(s)$.
This allows us to define $\itleft$ and $\itright$:
\begin{itemize}
  \item $\itright(P,l)$ ($\itleft(P,l)$) -- assume $v\in \stlayer_s(l)$.
    Return a pointer to the node to the right (left) of $v$ on $\stlayer_s(l)$, if such node exists.
\end{itemize}
For example, assume that $P$ points to the leftmost leaf in $\stree(\str)$ in \cref{fig:trees}.
Then $\itright(P, 0)$ is the pointer to the leaf representing the first $a$, but $\itright(P, 2)$ is the pointer to the node representing the string $anan$.

\begin{remark}The functions $\itright$ and $\itleft$ operating on $\stree(s)$ trees
(defined above)
are only used for the purpose of implementing the pointers to uncompressed
parse trees.
\end{remark}
\subsection{Pointers to Compressed Parse Trees}
As we later show, the pointers to $\ustree(s)$ trees can be implemented using pointers to $\stree(s)$ in a black-box fashion.
Thus, we first describe the pointers to trees $\stree(s)$.

\begin{lemma}\label{lem:ptr}
Assume that the grammar $\grammar$ is growing over time.
We can implement the entire interface of pointers to trees $\stree(s)$ where
$s\in\sigs(\grammar)$
in such a way that all operations take $O(1)$ worst-case time.
The additional time needed to update the shared pointer state is also $O(1)$ per creation of a new signature.
\end{lemma}
Before we embark on proving this lemma, we describe auxiliary
data structures constituting the shared infrastructure for pointer
operations.
Let $s\in\sigs(\grammar)$ and let $\stree_{\geq l}(s)$ be the tree $\stree(s)$
with the nodes of levels less than $l$ removed.
Now, define $\itfirst(s,l)$ and $\itlast(s,l)$ to be the signatures
corresponding to the leftmost and rightmost
leaves of $\stree_{\geq l}(s)$, respectively.
We aim at being able to compute functions $\itfirst$ and $\itlast$ in
constant time, subject to insertions of new signatures
into $\sigs(\grammar)$.

\newcommand{\bmask}{\mathtt{b}}
\newcommand{\bcnt}{\mathtt{bcnt}}

\begin{lemma}
  We can maintain a set of signatures $\sigs(\grammar)$ subject
to signature insertions and queries $\itfirst(s,l)$ and
$\itlast(s,l)$, so that both insertions and queries take constant time.
\end{lemma}

\begin{proof}
For brevity, we only discuss the implementation of $\itfirst$, as $\itlast$ can
be implemented analogously.
We form a forest $F$ of rooted trees on the set $\sigs(\grammar)$, such that each
$x\in\Sigma$ is a root of some tree in $F$.
Let $s\in\sigs(\grammar)\setminus\Sigma$.
If $s\to s_l s_r$, then $s_l$ is a parent of $s$ in $F$ and
if $s\to s_p^k$, where $k\in\mathbb{N}$, then $s_p$ is
a parent of $s$.
For each $s\in\sigs(\grammar)$ we also store a $2\mword$-bit mask $\bmask(s)$ with
the $i$-th bit set if some (not necessarily proper) ancestor of $s$ in $F$ has
level equal to $i$.
The mask requires a constant number of machine words.
When a new signature $s'$ is introduced, a new leaf has to be added
to $F$.
Denote by $s_p$ the parent of $s'$ in $F$.
The mask $\bmask(s')$ can be computed in $O(1)$ time: it is equal
$\bmask(s_p)+2^{\slev(s')}$.
We also store a dictionary $\bcnt$ that maps each introduced value $\bmask(s)$
to the number of bits set in $\bmask(s)$.
Note that since $\slev(s')>\slev(s_p)$, $\bcnt(\bmask(s'))=\bcnt(\bmask(s_p))+1$.

Now, it is easy to see that $\itfirst(s,l)$ is the highest (closest to the root)
ancestor $a$ of $s$ in $F$, such that $\slev(a)\geq l$.
In order to find $a$, we employ the \emph{dynamic level ancestor}
data structure of Alstrup et al.~\cite{Alstrup:2000}.
The data structure allows us to maintain a dynamic forest subject
to leaf insertions and queries for the $k$-th nearest ancestor
of any node.
Both the updates and the queries are supported in worst-case $O(1)$ time.
Thus, we only need to find such number $k$ that $a$ is
the $k$-th nearest ancestor of $s$ in $F$.
Let $\bmask'(s)$ be the mask $\bmask(s)$ with the bits numbered from $l$ to $2\mword-1$
cleared.
The mask $\bmask'(s)$ can be computed in constant time with standard bitwise operations.
Note that $\bmask'(s)=\bmask(a')$ for some ancestor $a'$ of $s$
and as a result $\bcnt$ contains the key $\bmask'(s)$.
Hence, $\bcnt(\bmask'(s))$ is the number of bits set in $\bmask'(s)$.
Also, $\bmask'(s)$ denotes the number of ancestors of $s$ at levels less than $l$.
Thus $k=\bcnt(\bmask(s))-\bcnt(\bmask'(s))$ is the number of ancestors
of $s$ at levels not less than $l$.
Consequently, the $(k-1)$-th nearest ancestor of $s$ is the needed node $a$.
\end{proof}

Equipped with the functions $\itfirst$ an $\itlast$, we now move to the proof
of \cref{lem:ptr}.
\begin{proof}[Proof of \cref{lem:ptr}]
  Let $w$ be such that $\stree(w)=\stree(s)$.
  Recall that we are interested in pointers to $\stree(s)$ for some fixed
$s\in\sigs(\grammar)$.
A functional stack will prove to be useful in the following description.
For our needs, let the functional stack be a data structure storing
values of type $\mathcal{A}$. Functional stack has the following interface:
\begin{itemize}
  \item $\stempty()$ -- return an empty stack
  \item $\stpush(S,val)$ -- return a stack $S'$ equal to $S$ with a value $val\in\mathcal{A}$ on top.
  \item $\stpop(S)$ -- return a stack $S'$ equal to $S$ with the top value removed.
  \item $\sttop(S)$ -- return the top element of the stack $S$.
\end{itemize}
The arguments of the above functions are kept intact.
Such stack can be easily implemented as a functional list \cite{Okasaki:1999}.

Denote by $r$ the root of $\stree(s)$.
The pointer $P$ to $v$ is implemented as a functional stack containing
some of the ancestors of $v$ in $\stree(s)$, including $v$ and $r$,
accompanied with some additional information (a $\itnil$ pointer
is represented by an empty stack).
The stack is ordered by the levels of the ancestors, with $v$
lying always at the top.
Roughly speaking, we omit the ancestors being the internal nodes
on long subpaths of $r\to v$ descending to the leftmost (rightmost)
child each time. Gąsieniec et al.~\cite{DBLP:conf/dcc/GasieniecKPS05} applied a similar idea
to implement forward iterators on grammar-compressed strings (traversing only the leaves of the parse tree).

More formally, let $A_v$ be the set of ancestors of $v$ in $\stree(s)$,
including $v$.
The stack contains an element per each $a\in A_v\setminus (X_l\cup X_r)$, where
$X_l$ ($X_r$ respectively) is a set of such ancestors $b\in A_v\setminus\{r\}$ that $b$ simultaneously:
\begin{itemize}
  \item[(1)] contains $v$ in the subtree rooted at its leftmost (rightmost resp.) child,
  \item[(2)] is the leftmost (rightmost) child of $\spar{s}(b)$.
\end{itemize}
The stack elements come from the set
$\sigs(\grammar)\times(\{\texttt{L},\texttt{R},\perp\}\cup\mathbb{N}_+)\times \mathbb{N}$.
Let for any node $u$ of $\stree(s)$, $\delta(u)$ denote an
integer such that $u$ represents the substring of $w$ starting at $\delta(u)+1$.
If $a=r$, then the stack contains the entry $(s,\perp,0)$.
If $a$ is the leftmost child of its parent in $\stree(s)$, then the stack
contains the entry $(\ussig(a),\texttt{L},\delta(a))$.
Similarly, for $a$ which is the rightmost child of $\spar{s}(a)$,
the stack contains $(\ussig(a),\texttt{R},\delta(a))$.
Otherwise, $a$ is the $i$-th child of some $k$-th power signature, where
$i\in(1,k)$, and the stack contains $(\ussig(a),i,\delta(a))$ in this case.
Note that the top element of the stack is always of the form $(\ussig(v),*,*)$.

Having fixed the internal representation of a pointer $P$, we are
now ready to give the implementation of the $\stree(s)$-pointer interface.
In order to obtain $\itsig(P)$, we take the first coordinate of the top element of $P$.
The $\itdegree(P)$ can be easily read from $\grammar$ and $\itsig(P)$.

Let us now describe how $\itindex(P)$ works.
If $P$ points to the root, then the return value is clearly $0$.
Otherwise, let the second coordinate of the top element be $y$.
We have that $y \in \{\texttt{L}, \texttt{R}, 2, 3, 4, \ldots\}$.
If $y \not\in \{\texttt{L}, \texttt{R}\}$, we may simply return $y$.
Otherwise, if $y = \texttt{L}$, $\itindex(P)$ returns $1$.
When $y = \texttt{R}$, we return $\itdegree(\itparent(P))$.

The third ``offset'' coordinate of the stack element $(\ussig(v),y,\delta(v))$ is only
maintained for the purpose of computing the value $\itrepr(P)$, which
is equal to $(\delta(v)+1,\delta(v)+\slength(\ussig(v)))$.
We show that the offset $\delta(v)$ of any entry depends
only on $v$, $y$ and the previous preceding stack element $(\ussig(v'),*,\delta(v'))$
(the case when $v$ is the root of $\stree(s)$ is trivial).
Recall that $v'$ is an ancestor of $v$ in $\stree(s)$.
\begin{itemize}
\item If $y=\texttt{L}$, then on the path $v'\to v$ in $\stree(s)$
  we only descend to the leftmost children, so $\delta(v)=\delta(v')$.
\item If $y=j$, then $\ussig(v')\to s_1^k$, $j\in(1,k)$ and
$v'$ is the parent of $v$ in $\stree(s)$.
Hence, $\delta(v)=\delta(v')+(j-1)\cdot\slength(s_1)$ in this case.
\item If $y=\texttt{R}$, then on the path $v'\to v$ in $\stree(s)$
  we only descend to the rightmost children, so
  $\delta(v)=\delta(v')+\slength(\ussig(v'))-\slength(\ussig(v))$.
\end{itemize}
For the purpose of clarity, we hide the third coordinate of the stack element
in the below discussion, as it can be computed according to the above
rules each time $\stpush$ is called.
Thus, in the following we assume that the stack contains
elements from the set $\sigs(\grammar)\times(\{\texttt{L},\texttt{R},\perp\}\cup\mathbb{N}_+)$.

The implementation of the primitives $\itroot$, $\itbegin$ and $\itend$
is fairly easy: $\itroot(s)$ returns $\stpush(\stempty(),\allowbreak (s,\perp))$,
$\itbegin(s)$ returns the pointer
$\stpush(\itroot(s),(\itfirst(s,0),\texttt{L}))$,
whereas $\itend(s)$ returns $\stpush(\itroot(s),\allowbreak(\itlast(s,0),\texttt{R}))$.

It is beneficial to have an auxiliary
function $\stcomp$.
If the stack $P$ contains $(s_1,y)$ as the top element
and $(s_2,y)$ as the next element, where $y\in \{\texttt{L},\texttt{R}\}$,
$\stcomp(P)$ returns a stack $P'$ with $(s_2,y)$ removed, i.e.
$\stpush(\stpop(\stpop(P)),\sttop(P))$.
Now, to execute $\itchild(P,k)$, we compute the $k$-th child $w$ of $v$
from $\grammar$ and return $\stcomp(\stpush(P,(\ussig(w),z))$,
where $z=\texttt{L}$ if $k=1$, $z=\texttt{R}$ if $v$ has exactly
$k$ children and $z=k$ otherwise.

The $\itparent(P)$ is equal to $\stpop(P)$ if $\sttop(P)=(s_v,k)$,
for an integer $k$ and $s_v=\ussig(v)$.
If, however, $\sttop(P)=(s_v,\texttt{L})$, then the actual parent
of $v$ might not lie on the stack.
Luckily, we can use the $\itfirst$ function: if $(s_u,*)$ is the second-to top stack
entry, then the needed parent is $\itfirst(s_u,\slev(s_v)-1)$.
Thus, we have $\itparent(P)=\stcomp(\stpush(\stpop(P),\allowbreak(\itfirst(s_u,\slev(s_v)-1),\texttt{L})))$
in this case.
The case when $\sttop(P)=(s_v,\texttt{R})$ is analogous -- we use $\itlast$
instead of $\itfirst$.

To conclude the proof, we show how to implement $\itright(P,l)$ (the implementation
of $\itleft(P,l)$ is symmetric).
We first find $Q$ -- a pointer to the nearest ancestor $a$ of $v$ such that $v$ is not in
the subtree rooted at $a$'s rightmost-child.
To do that, we first skip the ``rightmost path'' of nearest ancestors
of $v$ by computing a pointer $P'$ equal to $P$ if $\sttop(P)$ is
not of the form $(*,\texttt{R})$ and $\stpop(P)$ otherwise.
The pointer $Q$ can be now computed by calling $\itparent(P')$.
Now we can compute the index $j$ of the child $b$ of $a$ such that
$b$ contains the desired node: if $\sttop(P')=(*,\texttt{L})$, then $j=2$, and if $\sttop(P')=(*,k)$, then $j=k+1$.
The pointer $R$ to $b$ can be obtained by calling $\stcomp(\stpush(Q,\itsig(\itchild(Q,j)),y))$,
where $y$ is equal to $\texttt{R}$ if $b$ is the rightmost child
of $a$ and $j$ otherwise.
The exact value of the pointer $\itright(P,l)$ depends on whether
$b$ has level no more than $l$.
If so, then $\itright(P,l)=R$.
Otherwise, $\itright(P,l)=\stpush(R,(\itfirst(\itsig(R),l),\texttt{L}))$.

To sum up, each operation on pointers take worst-case $O(1)$ time, as they
all consist of executing a constant number of functional stack operations.
\end{proof}

\subsection{Pointers to Compressed Parse Trees}
We now show the implementation of pointers to uncompressed parse trees.

\begin{lemma}\label{lem:uptr}
Assume that the grammar $\grammar$ is growing over time.
We can implement the entire interface of pointers to trees $\ustree(s)$ where
$s\in\sigs(\grammar)$
in such a way that all operations take $O(1)$ worst-case time.
The additional worst-case time needed to update the shared pointer state is also $O(1)$ per creation of new signature.
\end{lemma}

\begin{proof}
To implement the pointers to the nodes of a tree $\ustree(s)$ we use pointers to the nodes of $\stree(s)$ (see \cref{lem:ptr}).
Namely, the pointer $\overline{P}$ to a tree $\ustree(s)$ is a pair $(P, l)$, where $P$ is a pointer to $\stree(s)$ and $l$ is a level.

Recall that $\ustree(s)$ can be obtained from $\stree(s)$ by dissolving nodes with one child.
Thus, a pointer $\overline{P}$ to a node $u$ of $\ustree(s)$ is represented by a pair $(P, l)$, where:
\begin{itemize}
\item If $u\in\stree(s)$, $P$ points to $u$. Otherwise, it points to the first descendant
  of $u$ in $\ustree(s)$ that is also a node of $\stree(s)$,
\item $l$ is the level of $u$ in $\ustree(s)$.
\end{itemize}

Thus, $\itlevel(\overline{P})$ simply returns $l$.
The operations $\itroot(s)$ and $\itbegin(s)$ can be implemented directly, by calling the corresponding operations on $\stree(s)$.
Observe that the desired nodes are also nodes of $\stree(s)$.
Moreover, thanks to our representation, the return value of $\itsig(\overline{P})$ is $\itsig(P)$.

To run $\itparent(\overline{P})$ we consider two cases.
If the parent of $\overline{P}$ is a node of $\stree(s)$, then $\itparent(\overline{P})$ returns $(\itparent(P), l+1)$.
Otherwise, it simply returns $(P, l+1)$.
Note that we may detect which of the two cases holds by inspecting the level of $\itparent(P)$.

The implementation of $\itchild(\overline{P}, k)$ is similar.
If the node pointed by $\overline{P}$ is also a node of $\stree(s)$, we return $(\itchild(P, k), l-1)$.
Otherwise, $k$ has to be equal to $1$ and we return $(P, l-1)$.

If $\itparent(\overline{P})$ does exist and is a node of $\stree(s)$, the return value of $\itindex(\overline{P})$ is the same as the return value of $\itindex(P)$.
Otherwise, the parent of the node pointed by $\overline{P}$ either has a single child or does not exist, so $\itindex(\overline{P})$ returns $1$.

In the function $\itlext$ we have two cases.
If $l$ is odd, then $\itlext(\overline{P})=0$, as the odd
levels of $\ustree(s)$ do not contain adjacent equal signatures.
The same applies to the case when $P$ points to the root of $\stree(s)$.
Otherwise, the siblings of $v$ constitute a block of equal signatures
on the level $l$.
Thus, we return $\itindex(\overline{P})-1$.
The implementation of $\itrext$ is analogous.

$\itright(\overline{P})$ returns $(\itright(P, l), l)$ and $\itleft(\overline{P})$ returns $(\itleft(P, l), l)$.

Finally, $\itright(\overline{P})$ can be used to implement $\itrskip(\overline{P},k)$ ($\itlskip$ is symmetric).
If $k=\itrext(\overline{P})+1$, then $\itrskip(\overline{P},k)=\itright(\itrskip(\overline{P},k-1))$.
Also, $\itrskip(\overline{P},0)=\overline{P}$ holds in a trivial way.
In the remaining case we have $k\in[1,\itrext(\overline{P})]$ and it follows that $\itsig(\itparent(\overline{P}))\to \itsig(\overline{P})^q$
for $q=\itindex(\overline{P})+\itrext(\overline{P})$.
Thus, we return $\itchild(\itparent(\overline{P}),\itindex(\overline{P})+k))$.
\end{proof}

\begin{remark}
In the following we sometimes use the pointer notation
to describe a particular node of $\ustree(s)$.
For example, when we say ``node $P$'', we actually
mean the node $v\in\ustree(s)$ such that $P$ points to $v$.
\end{remark}

\section{Context-Insensitive Nodes}\label{sec:contexti}
In this section we provide a combinatorial characterization of context-insensitive nodes,
introduced in \cref{sec:context_short}.
Before we proceed, let us state an important property of our way of parsing,
which is also applied in \cref{sec:om} to implement longest common prefix queries.

\begin{lemma}\label{lem:lcpd}
Let $i\in \mathbb{Z}_{> 0}$, $w,w'\in \sigs^*$ and let $p$ be the longest common prefix of $w$ and $w'$.
There exists a decomposition $p=p_\ell p_r$ such that $|\rle(p_r)|\le 1$ and $\shrink{i}(p_\ell)$ is the longest common prefix
of $\shrink{i}(w)$ and $\shrink{i}(w')$.
\end{lemma}
\begin{proof}
Observe that the longest common prefix of $\shrink{i}(w)$ and $\shrink{i}(w')$
expands to a common prefix of $w$ and $w'$. We set $p_\ell$ to be this prefix and $p_r$ to be the corresponding
suffix of $p$. If $p_r = \eps$ we have nothing more to prove and thus suppose $p_r\ne \eps$.

Note that $p_r$ starts at the last position of $p$ before which both
$\shrink{i}(w)$ and $\shrink{i}(w')$ place block boundaries.
Recall that $\shrink{i}$ places a block boundary between two positions only based on the characters on these positions.
Thus, since the blocks starting with $p_r$ are (by definition of $p_\ell$) different in $w$ and $w'$,
in one of these strings, without loss of generality in $w$, this block extends beyond $p$.
Consequently, $p_r$ is a proper prefix of a block. Observe, however,
that a block may contain more than one distinct character only if its length is exactly two.
Thus, $|\rle(p_r)|= 1$ as claimed.
\end{proof}

Next, let us recall the notions introduced in \cref{sec:context_short}. Here, we provide slightly more formal definitions.
For arbitrary strings $w$, $x$, $y$ we say that $x\edot w\edot y$ (which is formally a triple $(x,w,y)$)
is \emph{extension} of $w$. We often identify the extension with the underlying concatenated string $xwy$,
remembering, however, that a particular occurrence of $w$ is distinguished.
If $x=\eps$ or $y=\eps$, we say that $x\edot w\edot y$ is a \emph{right extension} or \emph{left extension}, respectively.

Let us fix an extension $x\edot w \edot y$ of $w$. Consider a node $v$ at level $l$ of $\ustree(w)$.
This node represents a particular fragment $w[i..j]$ which has a naturally corresponding fragment of the extension.
We say that a node $v'$ at level $l$ of $\ustree(xwy)$ is a \emph{counterpart} of $v$ with respect to $x\edot w \edot y$
(denoted $v \approx v'$) if $\ussig(v)=\ussig(v')$ and $v'$ represents the fragment of $xwy$ which naturally corresponds to $w[i..j]$.

\begin{lemma}\label{lem:close}
If a level-$l$ node $v\in \ustree(w)$ has a counterpart $v'$ with respect to a right extension $\eps \edot w \edot y$,
then $\uschild(v,k)\approx \uschild(v',k)$ for any integer $k$, $\usleft(v)\approx \usleft(v')$, and $\usleft(\uspar(v))\approx \usleft(\uspar(v'))$.
Moreover, if $\ussig(v)\ne \ussig(\usleft(v))$, then $\uspar(\usleft(v))\approx \uspar(\usleft(v'))$.
In particular, the nodes do not exist for $v$ if and only if the respective nodes do not exist for $v'$.

Analogously, if $v$ has a counterpart $v'$ with respect to a left extension  $x \edot w \edot \eps$,
then $\uschild(v,k)\approx \uschild(v',k)$ for any integer $k$, $\usright(v)\approx \usright(v')$, and $\usright(\uspar(v))\approx \usright(\uspar(v'))$.
Moreover, if $\ussig(v)\ne \ussig(\usright(v))$, then $\uspar(\usright(v))\approx \uspar(\usright(v'))$.
\end{lemma}

\begin{proof}
Due to symmetry of our notions, we only prove statements concerning the right extension.
First, observe that $\uschild(v,k)\approx \uschild(v',k)$ simply follows from the construction
of a parse tree and the definition of a counterpart.

Next, let us prove that $\usleft(v)\approx \usleft(v')$. If $l=0$, this is clear since
all leaves of $\ustree(w)$ have natural counterparts in $\ustree(wy)$.
Otherwise, we observe that $\usleft(u)=\usleft(\uspar(\uschild(u,1)))$ holds for every $u$.
This lets us apply the inductive assumption
for the nodes $\uschild(v,1)\approx \uschild(v',1)$ at level $l-1$ to conclude that $\usleft(v)\approx \usleft(v')$.

This relation can be iteratively deduced for all nodes to the left of $v$ and $v'$,
and thus $v$ and $v'$ represent symbols of $\cshrink{l}(w)$ and $\cshrink{l}(wy)$ within the longest
common prefix $p$ of these strings.
\cref{lem:lcpd} yields a decomposition $p=p_\ell p_r$ where $p':=\shrink{l+1}(p_\ell)$ is the longest common prefix of $\cshrink{l+1}(w)$
and $\cshrink{l+1}(wy)$ while $|\rle(p_r)|\le 1$. Clearly, nodes at level $l+1$ representing symbols within this longest
common prefix are counterparts.
Thus, whenever $u$ represents a symbol within $p_\ell$ and $u'$ is its counterpart with respect to $\eps \edot w \edot y$,
we have $\uspar(u)\approx \uspar(u')$.

Observe that $\ussig(v)\ne \ussig(\usleft(v))$ means that $\usleft(v)$ represents a symbol within $p_\ell$, and this immediately
yields $\uspar(\usleft(v))\approx \uspar(\usleft(v'))$.
Also, both in $\ustree(w)$ and in $\ustree(wy)$ nodes representing symbols in $p_r$ have a common parent, so
$\usleft(\uspar(v))\approx \usleft(\uspar(v'))$ holds irrespective of $v$ and $v'$ representing symbols within $p_\ell$ or within $p_r$.
\end{proof}

\begin{lemma}\label{lem:lr}
Let $v\in \ustree(w)$. If $v$ has counterparts with respect to both $x\edot w \edot \eps$ and $\eps \edot w \edot y$, then $v$
has a counterpart with respect to $x \edot w \edot y$.
\end{lemma}
\begin{proof}
We proceed by induction on the level $l$ of $v$. All nodes of level $0$ have natural counterparts with respect to any extension,
so the statement is trivial  for $l=0$, which lets us assume $l>0$.
Let $v^L$ and $v^R$ be the counterparts of $v$ with respect to $x\edot w \edot \eps$ and $\eps \edot w \edot y$, respectively.
Also, let $v_1,\ldots,v_k$ be the children of $v$.
Note that these nodes have counterparts $v^L_1,\ldots,v^L_k$ with respect to $x\edot w \edot \eps$ and
$v^R_1,\ldots,v^R_k$ with respect to $\eps \edot w \edot y$.
Consequently, by the inductive hypothesis they have counterparts $v'_1,\ldots,v'_k$ with respect to $x\edot w \edot y$.

Observe that $v^L_i \approx v'_i$ when $xwy$ is viewed as an extension $\eps \edot xw\edot y$ of $xw$.
Now, \cref{lem:close} gives $\usright(v^L)=\usright(\uspar(v^L_i))\approx \usright(\uspar(v'_i))$
and thus nodes $v'_1,\ldots,v'_k\in \ustree(xwy)$ have a common parent $v'$. We shall prove
that $v'$ is the counterpart of $v$ with respect to $x \edot w \edot y$. For this it suffices to prove
that it does not have any children to the left of $v'_1$ or to the right of $v'_k$.
However, since $\usright(v^L)\approx \usright(v')$ and $v^L_k$ is the rightmost child of $v^L$,
$v'_k$ must be the rightmost child of $v'$.
Next, observe that $v^R_i \approx v'_i$ when $xwy$ is viewed as an extension $x \edot wy \edot \eps$ of $wy$.
Hence, \cref{lem:close} implies that $\usleft(v^R)=\usleft(\uspar(v^R_i))\approx \usleft(\uspar(v'_i))$,
i.e., $\usleft(v^R) \approx \usleft(v')$. As $v^R_1$ is the leftmost child of $v^R$, $v'_1$ must be the leftmost child of $v'$.
\end{proof}

A node $v\in \ustree(w)$ is called \emph{context-insensitive} if it is preserved in any extension of $v$,
and \emph{left} (\emph{right}) \emph{context-insensitive}, if it is preserved in any left (resp. right) extension.
The following corollary translates the results of \cref{lem:close,lem:lr}
in the language of context-insensitivity. Its weaker version is stated in \cref{sec:context_short} as \cref{lem:context_insensitive}.

\begin{corollary}\label{cor:context_insensitive} Let $v$ be a node of $\ustree(w)$.
\begin{enumerate}[(a)]
\item\label{it:left} If $v$ is right context-insensitive, so are nodes $\usleft(v)$, $\usleft(\uspar(v))$, and all children of~$v$.
Moreover, if $\ussig(v)\ne \ussig(\usleft(v))$, then $\uspar(\usleft(v))$ is also right context-insensitive.
\item\label{it:right} If $v$ is left context-insensitive, so are nodes $\usright(v)$, $\usright(\uspar(v))$, and all children of~$v$.
Moreover, if $\ussig(v)\ne \ussig(\usright(v))$, then $\uspar(\usright(v))$ is also left context-insensitive.
\item\label{it:both} If $v$ is both left context-insensitive and right context-insensitive, it is context-insensitive.
\end{enumerate}
\end{corollary}

We say that a collection $L$ of nodes in $\ustree(w)$ or in $\stree(w)$ forms a \emph{layer} if every leaf has exactly one ancestor in $L$.
Equivalently, a layer is a maximal antichain with respect to the ancestor-descendant relation.
Note that the left-to-right order of $\ustree(w)$ gives a natural linear order on $L$, i.e., $L$ can be treated as a sequence of nodes.
The sequence of their signatures is called the \emph{decomposition} corresponding to the layer.
Observe a single decomposition may correspond to several layers in $\ustree(w)$, but only one of them does not contain nodes with exactly one child.
In other words, there is a natural bijection between decompositions and layers in $\stree(w)$.

We call a layer (left/right) context-insensitive if all its elements are (left/right) context-insensitive.

We also extend the notion of context insensitivity to the underlying decomposition.
The decompositions can be seen as sequences of signatures.
For two decompositions $D$, $D'$ we define their concatenation $D\cdot D'$ to be
the concatenation of the underlying lists.
The following fact relates context-insensitivity with concatenation of words and their decompositions.

\twodecompositions*

\begin{proof}
Let $L$ and $L'$ be layers in $\ustree(w)$ and $\ustree(w')$ corresponding to $D$ and $D'$, respectively.
Note that $ww'$ can be seen as a right extension $\eps \edot w \edot w'$ of $w$ and as a left extension $w\edot w' \edot \eps$ of $w'$.
Thus, all nodes in $L\cup L'$ have counterparts in $\ustree(ww')$ and these counterparts clearly form a layer.
Consequently, $D\cdot D'$ is a decomposition of $ww'$.
To see that $D\cdot D'$ is context-insensitive if $D$ and $D'$ are, it suffices to note that any extension $x \edot ww' \edot y$
of $ww'$ can be associated with extensions $x \edot w \edot w'y$ of $w$ and $xw \edot w' \edot y$ of $w'$.
\end{proof}

\begin{fact}\label{fct:layer}
Let $v,v'\in \stree(w)$ be adjacent nodes on a layer $L$.
If $v$ and $v'$ correspond to the same signature $s$, they are children of the same parent.
\end{fact}
\begin{proof}
  Let $l=\slev(s)$ and note that both $v$ and $v'$ both belong to $\Lambda_w(l)$,
i.e., they represent adjacent symbols of $\cshrink{l}(w)$.
However, $\shrink{l+1}$ never places a block boundary between two equal symbols.
Thus, $v$ and $v'$ have a common parent at level $l+1$.
\end{proof}

\section{Adding New Strings to the Collection}\label{sec:adding}

\subsection{Constructing Context-Insensitive Decompositions}\label{sec:constructing_cid}
Recall that the run-length encoding partitions a string into maximal blocks (called \emph{runs}) of equal characters
and represents a block of $k$ copies of symbol $S$ by a pair $(S,k)$, often denoted as $S^k$.
For $k=1$, we simply use $S$ instead of $(S,1)$ or $S^1$. 
In \cref{sec:single_string}, we defined the $\rle$ function operating on symbols as a component of our parse scheme.
Below, we use it just as a way of compressing a sequence of signatures.
Formally, the output of the $\rle$ function on a sequence of items
is another sequence of items,
where each maximal block of $k$ consecutive items $x$ is replaced with a single item $x^k$.

We store run-length encoded sequences as linked lists. This way we can create a new sequence consisting of $k$ copies
of a given signature $s$ (denoted by $s^k$) and concatenate two sequences $A$, $B$ (we use
the notation $A\cdot B$), both in constant time.
Note that concatenation of strings does not directly correspond to concatenation of the underlying lists:
if the last symbol of the first string is equal to the first symbol of the second string, two blocks need to be merged.

Decompositions of strings are sequences of symbols, so we may use $\rle$ to store them space-efficiently.
As mentioned in \cref{sec:context_short}, this is crucial for context-insensitive decompositions.
Namely, string $w$ turns out to have a context-insensitive decomposition $D$ such that $|\rle(D)|=O(\depth(w))$.
Below we give a constructive proof of this result.

\begin{algorithm}[ht]
\begin{algorithmic}[1]
\Function{$\layerop$}{$s$}
\State $\algptr := \itbegin(s)$
\State $\algqtr := \itend(s)$
\State $S = T = \eps$
\While{$P\neq Q$ \textbf{and} $\itparent(\algptr) \neq \itparent(\algqtr)$}
\If{$\itsig(P)\ne \itsig(\itright(P))$ \textbf{and} $\itparent(P)=\itparent(\itright(P))$}\label{alg:layer:if1}
\State $P' := P$ \Comment{Pointers $P'$ and $Q'$ are set for the purpose of analysis only.}
\State $P : =\itparent(P)$
\Else
\State $P' := \itrskip(P,\itrext(P))$
\State	$S := S \cdot \itsig(P)^{\itrext(P)}$
\State	$P := \itright(\itparent(P))$
\EndIf
\If{$\itsig(Q)\ne \itsig(\itleft(Q))$ \textbf{and} $\itparent(Q)=\itparent(\itleft(Q))$}\label{alg:layer:if2}
\State $Q':= Q$
\State $Q : =\itparent(Q)$
\Else
\State $Q':=\itlskip(Q,\itlext(Q))$
\State	$T := \itsig(Q)^{\itlext(Q)} \cdot T$
\State	$Q := \itleft(\itparent(Q))$
\EndIf
\EndWhile
\If{$\itsig(P)=\itsig(Q)$}
\State \Return $S \cdot \itsig(P)^{\itindex(Q)-\itindex(P)+1} \cdot T$
\Else
\State \Return $S \cdot \itsig(P)\cdot \itsig(Q) \cdot T$
\EndIf
\EndFunction
\end{algorithmic}
\caption{Compute a context-insensitive decomposition of a string $\str$ given by a signature $s$.}
\label{alg:layer}
\end{algorithm}

\begin{lemma}\label{lem:layerimpl}
Given $s\in \sigs(\grammar)$, one can compute the run-length encoding of a context-insensitive decomposition of $w=\sstr(s)$ in $O(\depth(w))$ time.
\end{lemma}

\begin{proof}
The decomposition is constructed using procedure $\layerop$ whose implementation is given as \cref{alg:layer}.
Let us define $\itright^k$ as the $k$-th iterate of $\itright$ (where $\itright(\itnil)=\itnil$)
and $\itleft^k$ as the $k$-th iterate of $\itleft$.
Correctness of \cref{alg:layer} follows from the following claim.

\begin{claim}
Before the $(l+1)$-th iteration of the \textbf{while} loop,
$P$ points to a left context-insensitive node at level $l$ and $Q$ points to a right context-insensitive node at level $l$.
Moreover $Q=\itright^m(P)$ for some $m\geq0$ and $S \cdot (\itsig(\itright^0(P)),\ldots,\itsig(\itright^m(P))) \cdot T$
is a context-insensitive decomposition of $w$.
\end{claim}
\begin{proof}
Initially, the invariant is clearly satisfied with $m=|w|-1\geq0$ because all leaves are context-insensitive.
Thus, let us argue that a single iteration of the \textbf{while} loop preserves the invariant.
Note that the iteration is performed only when $m > 0$.

Let $P_1$ and $P_2$ be the values of $P$ and $Q$ before the $(l+1)$-th iteration
of the \textbf{while} loop.
Now assume that the loop has been executed for the $(l+1)$-th time.
Observe that $P'=\itright^{m_p}(P_1)$ and $Q'=\itleft^{m_q}(Q_1)$ for some $m_p,m_q\geq0$
and from $\itparent(P_1)\neq\itparent(Q_1)$ it follows that $m_p+m_q\leq m$.
Also, $S$ and $T$ are extended so that
$S \cdot (\itsig(\itright^0(P')),\ldots,\itsig(\itright^{m-m_p-m_q}(P')) \cdot T$ is equal to the decomposition
we got from the invariant. Moreover, observe that $P'$ points to the leftmost child of the new value of $P$,
and $Q'$ points to the rightmost child of the new~$Q$. Consequently, after these values are set, we have $Q=\itright^{m'}(P)$ for some $m'\ge 0$
and $S \cdot (\itsig(\itright^0(P)),\ldots,\itsig(\itright^{m'}(P))) \cdot T$ is a decomposition of $w$.

Let us continue by showing that the new values of $P$ and $Q$ satisfy the first part of the invariant.
Note that $m'\ge 0$ implies $P$ and $Q$ are not set to $\itnil$.
In the \textbf{if} block at line~\ref{alg:layer:if1}, \cref{cor:context_insensitive}(\ref{it:right}) implies
that $\itparent(P)=\itparent(\itright(P))$ indeed points to a left context-insensitive at level $l+1$.
The same holds for $\itright(\itparent(P))$ in the \textbf{else} block.
A symmetric argument shows that $Q$ is set to point to a right context-insensitive node at level $i+1$.

Along with \cref{cor:context_insensitive}(\ref{it:both}) these conditions imply
 that all nodes $\itright^i(P)$ for $0 \le i \le m''$ are context-insensitive and thus the new representation is context-insensitive.
\end{proof}

To conclude the proof of correctness of \cref{alg:layer}, we observe that after leaving the main loop
the algorithm simply constructs the context-insensitive decomposition mentioned in the claim.
Either $P$ points at the root of $\ustree(s)$ or $\itparent(P)=\itparent(Q)$.
If $P$ points at the root or $\itsig(P)=\itsig(Q)$, we only need to add a single
run-length-encoded entry between $S$ and $T$.
Otherwise, $\itparent(P)$ has exactly two children $P$ and $Q$ with
different signatures, so we need to add two entries
$\itsig(P)$ and $\itsig(Q)$ between $S$ and $T$.
\end{proof}

We conclude with a slightly more general version of \cref{lem:layerimpl}.
\begin{lemma}\label{lem:layerimpl2}
Given $s\in \sigs(\grammar)$ and lengths $|x|,|y|$ such that $\sstr(s)=xwy$, one can compute in time $O(\depth(xwy))$
a run-length-encoded context-insensitive decomposition of $w$.
\end{lemma}
\begin{proof}
The algorithm is the same as the one used for  \cref{lem:layerimpl}.
The only technical modification is that the initial values of $P$ and $Q$ need to be set to leaves representing the first and the last position of $w$
in $xwy$. In general, we can obtain in $O(\depth(xwy))$ time a pointer to the $k$-th leftmost leaf of $\ustree(xwy)$.
This is because each signature can be queried for its length. Note that for nodes with more than two children we cannot scan
them in a linear fashion, but instead we need to exploit the fact that these children correspond to fragments of equal length
and use simple arithmetic to directly proceed to the right child.

Finally, we observe that by \cref{lem:close}, the subsequent values of $P$ and $Q$ in the implementation on $\ustree(xwy)$
are counterparts of the values of $P$ and $Q$ in the original implementation on $\ustree(w)$.
\end{proof}

\begin{remark}\label{remark:left-context}
In order to obtain a left context-insensitive decomposition we do not need to maintain the list $T$ in \cref{alg:layer} and it suffices to set $Q=\itparent(Q)$ at each iteration. Of course, \cref{lem:layerimpl2} needs to be restricted to left extensions
$x \edot w \edot \eps$. A symmetric argument applies to right context-insensitive decompositions.
\end{remark}

\subsection{Supporting Operations}
We now use the results of previous sections to describe the process of updating the collection
when $\makeop$, $\splitop$, and $\concop$ operations are executed.
Recall that we only need to make sure that the grammar $\grammar$
represents the strings produced by these updates.
The core of this process is the following algorithm.

\begin{algorithm}
\begin{algorithmic}[1]
\Function{$layer$}{$L$}
\While{$|L| > 1$}
\State $v := $ a lowest-level node of $\stree(\str)$ among all proper ancestors of $L$\label{line:ancestor}
\State $L := L \setminus \{\textrm{children of } v\} \cup \{v\}$
\EndWhile
\EndFunction
\end{algorithmic}
\caption{Iterate through all nodes above a layer $L$}
\label{alg:layers}
\end{algorithm}

\begin{lemma}
\cref{alg:layers} maintains a layer $L$ of $\stree(\str)$.
It terminates and upon termination, the only element of $L$ is the root of $\stree(\str)$.
Moreover, line~\ref{line:ancestor} is run exactly once per every proper ancestor of the initial $L$ in $\stree(\str)$
\end{lemma}

\begin{proof}
Consider line~\ref{line:ancestor} of \cref{alg:layers}.
Clearly, $v \not\in L$ and every child of $v$ belongs to $L$.
Recall that every node in $\stree(w)$, in particular $v$, has at least two children.
Thus, by replacing the fragment of $L$ consisting of all children of $v$ with $v$ we obtain a new layer $L'$, such that $|L'| < |L|$.
Hence, we eventually obtain a layer consisting of a single node.
From the definition of a layer we have that in every parse tree there is only one single-node layer and its only element is the root of $\stree(w)$.
The Lemma follows.
\end{proof}

We would like to implement an algorithm that is analogous to \cref{alg:layers}, but operates on decompositions, rather than on layers.
In such an algorithm in every step we replace some fragment $d$ of a decomposition $\decomp$ with a signature $s$ that generates $d$.
We say that we \emph{collapse} a production rule $s \rightarrow d$.

In order to describe the new algorithm, for a decomposition $\decomp$ we define the set of \emph{candidate production rules} as follows.
Let $s_1^{\alpha_1},\ldots, s_k^{\alpha_k}$ be the run-length encoding of a decomposition $\decomp$.
We define two types of candidate production rules.
First, for every $i \in \{1, \ldots, k\}$ we have rules $s \rightarrow s_i^{\alpha_i}$.
Moreover, for $i \in \{1, \ldots, k-1\}$ we have $s \rightarrow s_is_{i+1}$.
The \emph{level} of a candidate production rule $s \rightarrow s_1'\ldots s_m'$ is the level of the signature $s$.

The following lemma states that collapsing a candidate production rule of minimal level corresponds to executing one iteration of the loop in \cref{alg:layers}.
\begin{lemma}\label{lem:cpr}
Let $L$ be a layer of $\stree(w)$ and let $D$ be the corresponding decomposition.
Consider a minimum-level candidate production rule $s \rightarrow s'_1\ldots s'_m$ for $\decomp$.
Let $v'_1,\ldots,v'_m$ be the nodes of $\stree(w)$ corresponding to $s'_1,\ldots,s'_m$.
These nodes are the only children of a node $v$ which satisfies $\ussig(v)=s$.
\end{lemma}
\begin{proof}
Let $l = \slev(s)$ and let $v$ be a minimum-level node above $L$.
Observe that all children of $v$ belong to $L$ and the sequence of their signatures forms a substring of $D$.
By \cref{fct:layer} this substring is considered while computing candidate productions, and the level of the production
is clearly $l$.
Consequently, all nodes above $L$ have level at least $l$.
In other words, every node of $\stlayer_w(l-1)$ has an ancestor in $L$.
Since $\slev(s'_i)<l$, we have in particular $v_i'\in \stlayer_w(l-1)$,
i.e., the corresponding signatures $s'_1,\ldots,s'_m$ form a substring of $\cshrink{l-1}(w)$.
We shall prove that they are a single block formed by $\shrink{l}$.
It is easy to see that no block boundary is placed between these symbols.
If $l$ is even, we must have $m=2$ and $\shrink{l}$ simply cannot form larger blocks.
Thus, it remains to consider odd $l$ when $s'_i=s'$ for a single signature $s'$.
We shall prove that the block of symbols $s'$ in $\cshrink{l-1}(w)$ is not any longer. By \cref{fct:layer}
such a block would be formed by siblings of $v'_i$. However, for each of them an ancestor must belong to $L$.
Since their proper ancestors coincide with proper ancestors of $v'_i$, these must be the sibling themselves.
This means, however, that a block of symbols $s'$ in $D$ corresponding to $s'_1,\ldots,s'_m$ was not maximal,
a contradiction.
\end{proof}

\begin{lemma}\label{lem:compressdec}
Let $L$ be a layer in $\stree(\str)$ and $\decomp$ be its corresponding decomposition.
Given a run-length encoding of $\decomp$ of length $d$, we may implement an algorithm that updates $\decomp$ analogously to \cref{alg:layers}.
It runs in $O(d + \depth(\str))$ time.
The algorithm fails if $\depth(w) > 2\cdot \mword$.
\end{lemma}

\begin{proof}
By \cref{lem:cpr}, it suffices to repeatedly collapse the candidate production rule of the smallest level.
We maintain the run-length encoding of the decomposition $\decomp$ as a doubly linked list $Y$, whose every element corresponds to a signature and its multiplicity.
Moreover, for $i = 1, \ldots, 2\mword$ we maintain a list $P_i$ of candidate production rules of level $i$.
The candidate production rules of higher levels are ignored.
With each rule stored in some $P_i$ we associate pointers to elements of $Y$, that are removed when the rule is collapsed.
Observe that collapsing a production rule only affects at most two adjacent elements of $Y$, since we maintain a run-length encoding of the decomposition.

Initially, we compute all candidate production rules for the initial decomposition $\decomp$.
Observe that this can be done easily in time proportional to $d$.
Then, we iterate through the levels in increasing order.
When we find a level $i$, such that the list $P_i$ is nonempty, we collapse the production rules from the list $P_i$.
Once we do that, we check if there are new candidate production rules and add them to the lists $P_i$ if necessary.
Whenever we find a candidate production rule $s \rightarrow \strb$ to collapse, we first check if $\grammar$ already contains a signature $s_2$ associated with an equivalent production rule $s_2 \rightarrow \strb$ for some $s_2$.
If this is the case, the production rule from $\grammar$ is collapsed.
Otherwise, we collapse the candidate production rule, and add signature $s$ to $\grammar$.
Note that this preserves \cref{inv:consistency}.
Also, this is the only place where we modify $\grammar$.

Note that the candidate production rules are never removed from the lists $P_i$.
Thus, collapsing one candidate production rule may cause some production rules in lists $P_j$ ($j > i$) to become obsolete (i.e., they can no longer be collapsed).
Hence, before we collapse a production rule, we check in constant time whether it can still be applied, i.e., the affected signatures still exist in the decomposition.

As we iterate through levels, we collapse minimum-level candidate production rules.
This process continues until the currently maintained decomposition has length $1$.
If this does not happen before reaching level $2\mword$, the algorithm fails.
Otherwise, we obtain a single signature that corresponds to a single-node layer in $\stree(\str)$.
The only string in such a layer is the root of $\stree(\str)$, so the decomposition of length $1$ contains solely the signature of the entire string $\str$.
As a result, when the algorithm terminates successfully, the signature of $\str$ belongs to $\sigs(\grammar)$.
Hence, the algorithm indeed adds $\str$ to $\grammar$.

Regarding the running time, whenever the list $Y$ is modified by collapsing a production rule, at most a constant number of new candidate production rules may appear and they can all be computed in constant time, directly from the definition.
Note that to compute the candidate production rule of the second type, for two distinct adjacent signatures $s_1$ and $s_2$ in $\decomp$ we compute the smallest index $i > \max(\slev(s_1), \slev(s_2))$, such that $\hs_{i/2}(s_1) = 0$ and $\hs_{i/2}(s_2) = 1$.
This can be done in constant time using bit operations, because the random bits $\hs_i(\cdot)$ are stored in machine words.
Since we are only interested in candidate production rules of level at most $2\mword$, we have enough random bits to compute them.

Every time we collapse a production rule, the length of the list $Y$ decreases.
Thus, this can be done at most $d$ times.
Moreover, we need $O(d)$ time to compute the candidate production rules for the initial decomposition.
We also iterate through $O(\depth(\str))$ lists $P_i$.
In total, the algorithm requires $O(d + \depth(\str))$ time.
\end{proof}

\begin{corollary}\label{cor:decomposition-to-sig}
Let $\decomp$ be a decomposition of a string $w$.
Assume that for every signature $s$ in $\decomp$ we have $s \in \sigs(\grammar)$.
Then, given a run-length encoding of $\decomp$ of length $d$, we may add $w$ to $\grammar$ in $O(d + \depth(w))$ time.

The algorithm preserves \cref{inv:consistency} and fails if $\depth(w)>2\mword$.
\end{corollary}

Using \cref{cor:decomposition-to-sig} we may easily implement $\makeop$, $\splitop$ and $\concop$ operations.

\begin{lemma}\label{lem:makeop}
Let $w$ be a string of length $n$.
We can execute $\makeop(w)$ in $O(n + \depth(\str))$ time.
This operation fails if $\depth(w) > 2\mword$.
\end{lemma}

\begin{proof}
Observe that $\str$ (treated as a sequence) is a decomposition of $\str$.
Thus, we may compute its run-length encoding in $O(n)$ time and then apply \cref{cor:decomposition-to-sig}.
\end{proof}

\begin{lemma}\label{lem:splitop}
Let $\str$ be a string represented by $\grammar$ and $1 \leq k \leq |\str|$.
Then, we can execute a $\splitop(\str, k)$ operation in $O(\depth(\str) + \depth(\str[1..k]) + \depth(\str[k+1..]))$ time.
This operation fails if $\max(\depth(\str[1..k]), \depth(\str[k+1..]) > 2\mword$.
\end{lemma}

\begin{proof}
Using \cref{lem:layerimpl2} we compute the run-length-encoded decompositions of $\str[1..k]$ and $\str[k+1..]$ in $O(\depth(\str))$ time.
Then, it suffices to apply \cref{cor:decomposition-to-sig}.
\end{proof}

\begin{lemma}\label{lem:concop}
Let $\str_1$ and $\str_2$ be two strings represented by $\grammar$.
Then, we can execute a $\concop(\str_1, \str_2)$ operation in $O(\depth(\str_1) + \depth(\str_2) +  \depth(\str_1 \cdot \str_2))$ time.
The operation fails if $\depth(\str_1 \cdot \str_2) > 2\mword$.
\end{lemma}

\begin{proof}
We use \cref{lem:layerimpl} to compute the run-length-encoded context-insensitive decompositions of $\str_1$ and $\str_2$ in $O(\depth(\str_1) + \depth(\str_2))$ time.
By \cref{fact:two-decompositions} the concatenation of these decompositions is a decomposition of $\str_1 \cdot \str_2$.
Note that given run-length encodings of two strings, we may easily obtain the run-length encoding of the concatenation of the strings.
It suffices to concatenate the encodings and possibly merge the last symbol of the first string with the first symbol of the second one.
Once we obtain the run-length encoding of a decomposition of $\str_1 \cdot \str_2$, we may apply \cref{cor:decomposition-to-sig} to complete the proof.

Note that for the proof, we only need a right-context insensitive decomposition of $\str_1$ and a left-context insensitive decomposition of $\str_2$.
Thus, we apply the optimization of \cref{remark:left-context}.
However, it does not affect the asymptotic running time.
\end{proof}

\subsection{Conclusions}\label{sec:conclusions}

The results of this section are summarized in the following theorem.

\begin{theorem}\label{thm:data_structure}
There exists a data structure which maintains a family of strings $\coll$ and supports $\concop$ and $\splitop$ in $O(\log n +\log t)$
time, $\makeop$ in $O(n + \log t)$ time, where $n$ is the total length of strings involved in the operation and $t$
is the total input size of the current and prior updates (linear for each $\makeop$, constant for each $\concop$ and $\splitop$).
An update may fail with probability $O(t^{-c})$ where $c$ can be set as an arbitrarily large constant.
The data structure assumes that total length of strings in $\coll$ takes at most $\mword$ bits in the binary representation.
\end{theorem}

\begin{proof}
Consider an update which creates a string $\str$ of length $n$.
By \cref{lem:small_depth}, we have that $$\mathbb{P}(\depth(\str) \leq 8(c \ln t + \ln n)) \geq 1 - e^{c \ln t} = 1 - 1/t^{c}.$$
Hence, an update algorithm may fail as soon as it notices that $\depth(\str) > 8(c \ln t + \ln n)$.
The total length of strings in $\coll$ is at least $\max(n,t)$, and thus $8(c \ln t + \ln n)\le 8(c+1)\mword\ln2$.
We extend the machine word to $\mword'=4(c+1)\mword\ln 2 = O(B)$ bits, which results in a constant factor overhead in the running time.

This lets us assume that updates implemented according to \cref{lem:makeop,lem:splitop,lem:concop} do not fail.
By \cref{lem:makeop}, $\makeop$ runs in $O(n+\depth(\str))$ time, and, by \cref{lem:splitop,lem:concop}, both $\concop$ and $\splitop$ run in $O(\log n +\depth(\str))$ time.
Since we may terminate the algorithms as soon as $\depth(\str)$ turns out to be larger than $8(c \ln t + \ln n)$,
we can assume that $\makeop(\str)$ runs in $O(n+\log t)$ time, while $\concop$ and $\splitop$ run in $O(\log n +\log t)$ time.
\end{proof}

Note that the failure probability in \cref{thm:data_structure} is constant for the first few updates. 
Thus, when the data structure is directly used as a part of an algorithm, the algorithm could fail with constant probability.
However, as discussed in \cref{app:model}, restarting the data structure upon each failure is a sufficient mean to eliminate failures 
still keeping the original asymptotic running time with high probability with respect to the total size of updates (see~\cref{lem:restart}).
Below, we illustrate this phenomenon in a typical setting of strings with polynomial lengths.

\begin{corollary}\label{cor:persistent}
Suppose that we use the data structure to build  a dynamic collection of $n$ strings, each of length at most $\poly(n)$.
This can be achieved by a Las Vegas algorithm whose running time with high probability is $O(N+M\log n)$ where $N$
is the total length of strings given to $\makeop$ operations and $M$ is the total number of updates ($\makeop$, $\concop$, and $\splitop$).
\end{corollary}

\begin{proof}
We use \cref{thm:data_structure} and restart the data structure in case of any failure.
Since the value of $n$ fits in a machine word (as this is the assumption in the word RAM model),
by extending the machine word to $O(B)$ bits, we may ensure that the total length of strings in $\coll$ (which is $\poly(n)$)
fits in a machine word.  Moreover, $t=\poly(n)$, so $\log t=O(\log n)$, i.e., the running time of $\makeop(|w|)$ is $O(|w|+\log n)$
and of $\splitop$ and $\concop$ is $O(\log n)$.
We set $c$ in \cref{thm:data_structure} as a sufficiently large constant
so that \cref{lem:restart} ensures that the overall time bound holds with the desired probability.
\end{proof}

\section{Dynamic String Collections: Lower Bounds}\label{sec:lb}

In this section we prove a lower bound for any data structure maintaining a dynamic collection
of strings under operations $\makeop(w)$, $\concop(w_1,w_2)$ and $\splitop(w,k)$ and
answering equality queries $\eqop(w_1,w_2)$. The strings are not persistent, that is, both $\concop(w_1,w_2)$
and $\splitop(w,k)$ destroy their arguments and return new strings. Furthermore,
$\eqop$ only needs to support comparing strings of length one.
In such a setting, if we are maintaining $\poly(n)$ strings of length at most $\poly(n)$,
the amortized complexity of at least one operations among $\concop$, $\splitop$, $\eqop$ is
$\Omega(\log n)$. The lower bound applies to any Monte Carlo structure returning correct answers
w.h.p.

The lower bound is based on a reduction from dynamic connectivity. A well-known
result by P\v{a}tra\c{s}cu and Demaine~\cite{logarithmic} is that any data structure maintaining
an undirected graph under operations $\insertop(u,v)$, $\deleteop(u,v)$ and answering queries
$\connectedop(u,v)$ either needs $\Omega(\log n)$ amortized time for an update or
$\Omega(\log n)$ amortized time for a query, where $\insertop(u,v)$ inserts a new edge $(u,v)$,
$\deleteop(u,v)$ deletes an existing edge $(u,v)$ and $\connectedop(u,v)$ checks if
$u$ and $v$ are in the same component.
\begin{shortv}
The lower bound is based on a reduction from dynamic connectivity. A well-known
result by P\v{a}tra\c{s}cu and Demaine~\cite{logarithmic} is that any data structure maintaining
an undirected graph under operations $\insertop(u,v)$, $\deleteop(u,v)$ and answering queries
$\connectedop(u,v)$ either needs $\Omega(\log n)$ amortized time for an update or
$\Omega(\log n)$ amortized time for a query.%
\end{shortv}
It can be assumed that the graph is
a collection of disjoint paths and the lower bound applies to any Monte Carlo structure returning
correct answers with high probability. For the reduction we represent every path with
a string, so that connectivity can be checked by inspecting one character of an appropriate string.

Because our reduction is not completely black-box,
we start with a brief review of the proof presented in ~\cite{logarithmic}.

\subsection{Lower Bound for Dynamic Connectivity}
To prove a lower bound on dynamic connectivity, we consider the following problem:
maintain $\sqrt{n}$ permutations $\pi_1,\pi_2,\ldots,\pi_{\sqrt{n}}\in S_{\sqrt{n}}$ on
$\sqrt{n}$ elements. An update $\updateop(i,\pi)$ sets $\pi_i := \pi$ and a query
$\verifyop(i,\pi)$ checks whether $\pi = p_i$,
where $p_i := \pi_{i} \circ \pi_{i-1} \circ \ldots \circ \pi_{1}$.
We assume that $n$ is a power of two and consider sequences
of $\sqrt{n}$ pairs of operations, each pair being $\updateop(i,\pi)$
with $i\in [1,\sqrt{n}]$ and $\pi\in S_{\sqrt{n}}$ chosen uniformly at random and $\verifyop(i,\pi)$ with
$i\in [1,\sqrt{n}]$ chosen uniformly at random and $\pi = \pi_{i} \circ \pi_{i-1} \circ \ldots \circ \pi_{1}$.
That is, the $\verifyop$ is asked to prove a tautology, and must gather enough information
to certify it by probing cells. Then, if the word size is $\Theta(\log n)$,
it can be proved that the expected total number of cell probes (and hence also the total time
complexity) must be $\Omega(n\log n)$ through an entropy-based argument. The essence
of this argument is that if we consider $\ell$ indices $q_1 < q_2 < \ldots q_\ell$
and execute at least one update $\updateop(i,\pi)$ with $i\in [q_{j-1}+1,q_j]$ for
every $j=2,3,\ldots,\ell$, then all partial sums $p_{q_j}$ are independent random
variables uniformly distributed in $S_{\sqrt{n}}$.

We encode the problem of maintaining such $\sqrt{n}$ permutations
as an instance of dynamic connectivity for disjoint paths. For every $i=0,1,\ldots,\sqrt{n}$
we create a layer of $\sqrt{n}$ nodes $v_{i,1},v_{i,2},\ldots,v_{i,\sqrt{n}}$. Then, for every
$i=1,2,\ldots,\sqrt{n}$ and $j=1,2,\ldots,\sqrt{n}$ we connect $v_{i-1,j}$ with $v_{i,\pi_i(j)}$.
Clearly, the resulting graph is a collection of disjoint paths. Furthermore,
$\verifyop(i,\pi)$ can be implemented by asking all $\sqrt{n}$ queries of the form
$\connectedop(v_{0,j},v_{i,\pi(j)})$. $\updateop(i,\pi)$ requires first removing all edges
between nodes from the $(i-1)$-th and $i$-th layer and then connecting them
according to $\pi$. In total, the whole sequence can be processed with $2n$ updates
and $n$ queries, and therefore the amortized complexity of $\insertop$,
$\deleteop$ or $\connectedop$ is $\Omega(\log n)$ on some input. The lower bound
immediately applies also to Las Vegas algorithms. It can be also argued that it
applies to Monte Carlo algorithms with error probability $n^{-c}$ for sufficiently
large $c$, see Section 6.4 of~\cite{logarithmic}.

\subsection{Reduction}
Now we explain how to further encode the instance of dynamic connectivity for disjoint paths
as maintaining a dynamic collection of strings. Intuitively, we represent every path with
a string.
The $j$-th string $s_j$ describes the path connecting $v_{0,j}$ and $v_{\sqrt{n},p_{\sqrt{n}}(j)}$.
More precisely, we maintain the invariant that $s_j = p_0(j) p_1(j) \ldots p_{\sqrt{n}}(j)$,
where every $p_i(j)$ is a separate character. Additionally, for every $j$ we prepare
a one-character string $c_j$ such that $c_j\leq c_k$ for $j\neq k$.

To implement $\connectedop(v_{0,j},v_{i,\pi(j)})$ we need to check if the $(i+1)$-th character of $s_j$
is equal to $\pi(j)$. We split $s_j$ into two parts $s'_js''_js'''_j$ such that $|s'_j|=i$, $|s''_j|=1$
and $|s'''_j|=\sqrt{n}-i$
and compare $s''_j$ with the previously prepared string $c_{\pi(j)}$. Then,
we merge $s'_j$, $s''_j$ and $s'''_j$ to recover the original $s_j$ and maintain the invariant.

The difficulty lies in implementing $\updateop(i,\pi)$. We need to replace all edges between
the $(i+1)$-th and $i$-th layer. We first split every $s_j$ into two parts $s'_j s''_j$ such that
$|s'_j| = i$ and $|s''_j| = \sqrt{n}-i+1$. Then we would like to create the new $s_j$
by concatenating $s'_j$ and $s''_{j'}$ such that $s'_j[i]=k$ and $s''_{j'}[1]=\pi(k)$
(in other words, $k = \pi_{i-1}(j)$). However, while we can extract every such $k$,
we are not able to find the corresponding $j'$ for every $j$ by only comparing
strings in our collection (and without augmenting the set of allowed operations).
However, recall that every $\pi$ is chosen uniformly at random. We take $j' = \pi(j)$,
or in other words set $\pi_i = \pi \circ p^{-1}_{i-1}$ instead of $\pi_i = \pi$.
It is easy to implement such update in our setting, because after choosing $\pi$
uniformly at random we just need to concatenate $s'_j$ and $s''_{\pi(j)}$ to form
the new $s_j$, for every $j=1,2,\ldots,\sqrt{n}$, but the meaning of $\updateop(i,\pi)$
is different now. Nevertheless, the expected total
number of cell probes must still be $\Omega(n\log n)$ by the same argument. That is,
if we consider $\ell$ indices $q_1 < q_2 < \ldots q_\ell$
and call $\updateop(i,\pi)$ with $i\in [q_{j-1}+1,q_j]$ for
every $j=2,3,\ldots,\ell$, then all $p_{q_j}$ are independent random variables uniformly
distributed in $S_{\sqrt{n}}$. So the amortized complexity of $\concop$, $\splitop$ or
$\eqop$ must be $\Omega(\log n)$. Again, the lower bound immediately
applies to Las Vegas algorithms and, by the same reasoning as in the lower bound for
dynamic connectivity, to Monte Carlo algorithms with error probability $n^{-c}$
for sufficiently large $c$.

\begin{theorem}
\label{thm:polylowerbound}
For any data structure maintaining a dynamic collection of $\poly(n)$ non-persistent strings of length
$\poly(n)$ subject to operations $\makeop$, $\concop$, $\splitop$, and
$\eqop$ queries, correct with high probability, the amortized complexity of
$\concop$, $\splitop$, or $\eqop$ is $\Omega(\log n)$.
\end{theorem}

\section{Dynamic String Collections with Order Maintenance}\label{sec:om}

\newcommand{\prevst}{\mathtt{prev}}
\newcommand{\succst}{\mathtt{succ}}
\newcommand{\findst}{\mathtt{find}}
\newcommand{\insertst}{\mathtt{insert}}
\newcommand{\nilst}{\mathtt{nil}}
\newcommand{\sub}{\subseteq}
\newcommand{\sm}{\setminus}
\newcommand{\utrie}{\overline{T}}
\newcommand{\trie}{T}

\newcommand{\tval}{\mathtt{val}}
\newcommand{\tlink}{\mathtt{link}}
\newcommand{\tleaf}{\mathtt{leaf}}
\newcommand{\tsub}{\mathtt{sub}}
\newcommand{\tlca}{\mathtt{lca}}
\newcommand{\troot}{\mathtt{root}}
\newcommand{\tdepth}{\mathtt{depth}}
\newcommand{\tup}{\mathtt{up}}
\newcommand{\tdown}{\mathtt{down}}
\newcommand{\tmake}{\mathtt{make\_explicit}}
\newcommand{\taddleaf}{\mathtt{add\_leaf}}

In this section we describe how to extend the data structure developed in \cref{sec:collection} to support the following two types of queries
in constant time:
\begin{itemize}
  \item $\compop(\str_1,\str_2)$, for $\str_1,\str_2\in \coll$,  checks whether $\str_1<\str_2$,
    $\str_1=\str_2$, or $\str_1>\str_2$.
  \item $\lcpop(\str_1,\str_2)$, for $\str_1,\str_2\in \coll$, returns the length of the longest common prefix
    of $\str_1$ and $\str_2$.
\end{itemize}

In order to support these queries, we maintain the strings from $\coll$ in the lexicographic order.
The overall idea is similar to the one sketched in~\cite{Alstrup}.
However, due to the reasons already described in \cref{sec:comp}, our implementation requires some more insight.

For each non-negative integer $i$, we maintain a \emph{trie} $T_i$ that contains $\cshrink{i}(\str)$ for each string $\str$ in the data structure
with $\depth(w)\ge i$.
The tree is stored in a compacted form, that is we compress long paths consisting of nodes with exactly one child.
Since an edge of a trie may correspond to a long substring of some $\cshrink{i}(\str)$, roughly speaking, we store in it a pointer to a tree $\ustree(\str)$ that allows us to access the sequence of symbols represented by the edge.
Moreover, we maintain links between corresponding nodes in $T_i$ and $T_{i-1}$.
In order to update the tries, we develop pointers for traversing them, with similar objectives to that behind the pointers to $\ustree(\str)$ trees.

Whenever a string is added to our data structure, we update the tries $T_i$.
At the same time, we discover the string from $\coll$ that lies just before the inserted string in the lexicographically sorted list of all strings from $\coll$.
By combining this with an order-maintenance data structure~\cite{Dietz:1987,Bender:2002}, we are able to answer $\compop$ queries in constant time.
Moreover, by using a data structure for answering lowest common ancestor queries~\cite{Cole:2005} on the trie $T_0$, we are able to answer $\lcpop$ queries in constant time.
This idea was not used in~\cite{Alstrup}, where $\lcpop$ queries are answered in logarithmic time.

Let us now describe the algorithm more formally.
Suppose that a grammar $\grammar$ is maintained so that it satisfies \cref{inv:consistency}
and so that it may only grow throughout the algorithm.
Then, we may maintain a set of strings $\coll\sub \Sigma^+$ so that each string $w\in \coll$
is represented by a signature $s\in \sigs(\grammar)$ subject to the following operations:
\begin{itemize}
  \item queries $\compop(w_1,w_2)$ and $\lcpop(w_1,w_2)$ for every $w_1,w_2\in \coll$,
  \item update $\insertop(s)$ for  $s\in \sigs(\grammar)$  inserting the corresponding string $w=\sstr(s)$ to $\coll$.
\end{itemize}
Our solution supports queries in constant time while an update requires $O(\depth(w)+\log |\coll|)$ time.
Note that in order to implement the dynamic string collections, we store the order-maintenance component
for our collection $\coll$ of strings. Consequently, we call $\insertop$ for all strings
constructed with $\splitop$, $\concop$, and $\makeop$ operations.

\subsection{Simple $\lcpop$ Implementation}
As the technical side of our construction is quite complex, we begin with a simple implementation of the $\lcpop(w_1,w_2)$ operation,
which runs in $O(\depth(w_1)+\depth(w_2))$ time.
In the efficient implementation of the order-maintenance component, the $\insertop(s)$ operation uses similar ideas
to compute the longest prefix of $\sstr(s)$ which is also a prefix of some string in the collection $\coll$.
This involves dealing with many technical issues,  and thus it is easier to understand the key combinatorial idea by analyzing the $\lcpop(w_1,w_2)$ implementation described below.

The algorithm is based on~\cref{lem:lcpd}.
For a pointer $P$ to some $\ustree(s)$, we denote by $\itstring(P)$ the prefix of $\cshrink{\itlevel(P)}(\str)$ that ends immediately before the symbol represented by the node pointed by $P$
and by $\itsuf(P)$ the corresponding suffix of $\cshrink{\itlevel(P)}(\str)$.

\begin{algorithm}
\begin{algorithmic}[1]
\Function{$\lcpop$}{$s_1$, $s_2$}

\State $\algptr_1 := \itroot_{\ustree}(s_1)$
\State $\algptr_2 := \itroot_{\ustree}(s_2)$

\While{$\itlevel(\algptr_1) > \itlevel(\algptr_2)$} \label{alg:samelevel_beg}
\State $\algptr_1 := \itchild(\algptr_1, 1)$
\EndWhile

\While{$\itlevel(\algptr_2) > \itlevel(\algptr_1)$}
\State $\algptr_2 := \itchild(\algptr_2, 1)$
\EndWhile \label{alg:samelevel_end}

\For{$i=\itlevel(\algptr_1)$ \textbf{downto} $0$} \label{alg:lcp:for}
    \If{$\itsig(P_1) = \itsig(P_2)$}\label{alg:lcp:cond}
    	\State $k := \min(\itrext(P_1),\itrext(P_2))$\label{alg:lcp:k}
    	\State $P_1 = \itrskip(P_1,k)$
    	\State $P_2 = \itrskip(P_2,k)$
     \EndIf
    \If{$P_1 = \itnil$ \textbf{or} $P_2=\itnil$}
    \State \Return $\min(\slength(s_1),\slength(s_2))$
   	\EndIf
    \If{$i > 0$}\label{alg:lcp:lcp}
    	\State $P_1 = \itchild(P_1,1)$
    	\State $P_2 = \itchild(P_2,1)$
    \EndIf
\EndFor
\State $(i,i) := \itrepr(P_1)$
\State \Return $i-1$ \label{line:return}
\EndFunction
\end{algorithmic}
\caption{Pseudocode of $\lcpop$}
\label{alg:lcp}
\end{algorithm}

\begin{lemma}\label{lem:lcp}
  Let $\str_1$ and $\str_2$ be two strings represented by $\grammar$ given by the signatures $s_1,s_2\in\sigs(\grammar)$.
Then, we can execute an $\lcpop(s_1, s_2)$ operation in $O(\depth(\str_1) + \depth(\str_2))$ time.
\end{lemma}

\begin{proof}
The algorithm traverses the trees $\ustree(s_1)$ and $\ustree(s_2)$ top-down, and at every level $j$ it finds the longest common prefix of $\cshrink{j}(\str_1)$ and $\cshrink{j}(\str_2)$.
The pseudocode is given as \cref{alg:lcp}.
Let us now discuss the details.

The goal of lines~\ref{alg:samelevel_beg}--\ref{alg:samelevel_end} is to make both pointers point to nodes of the same level,
which is the minimum of $\depth(w_1)$ and $\depth(w_2)$.
At the beginning of each iteration of the \textbf{for} loop starting in line~\ref{alg:lcp:for}, the following invariants are satisfied.
\begin{enumerate}
  \item $\itlevel(P_1)=\itlevel(P_2)=i$,
  \item $\itstring(P_1)=\itstring(P_2)$, and
  \item the longest common prefix of $\itsuf(P_1)$ and $\itsuf(P_2)$ contains at most one distinct character.
\end{enumerate}
Moreover, we claim that at line~\ref{alg:lcp:lcp} of each iteration, $\itstring(P_1)=\itstring(P_2)$
is the longest common prefix of $\cshrink{i}(w_1)$ and $\cshrink{i}(w_2)$.

The first invariant is satisfied at the beginning of the first iteration, because the \textbf{while} loops ensures $\itlevel(\algptr_1) = \itlevel(\algptr_2)$. It is easy to see that each iteration of the \textbf{for} loop preserves it, as it decreases both $\itlevel(\algptr_1)$ and $\itlevel(\algptr_2)$ by one (unless $i=0$).
Let us now move to the other two invariants.

Before the first iteration, $\itstring(\algptr_1)=\itstring(\algptr_2)=\eps$ and $|\itsuf(P_1)|=1$ or $|\itsuf(P_2)|=1$
because $P_1=\itroot_{\ustree}(w_1)$ or $P_2=\itroot_{\ustree(w_2)}$.
Thus, the invariants are satisfied.

It remains to show that they are preserved in every iteration.
If $\itsig(P_1)\ne \itsig(P_2)$ at line~\ref{alg:lcp:cond}, then $\itsuf(P_1)$ and $\itsuf(P_2)$
start with different symbols and thus $\itstring(P_1)=\itstring(P_2)$ is already the longest
common prefix of $\cshrink{i}(w_1)$ and $\cshrink{i}(w_2)$.
Otherwise, the longest common prefix of $\itsuf(P_1)$ and $\itsuf(P_2)$ consists (by the invariant) of exactly
one distinct character, and the prefix can be easily seen to be $s^k$
where $s=\itsig(P_1)=\itsig(P_2)$ and $k$ is computed in line~\ref{alg:lcp:k}.
In the following two lines we shift $P_1$ and $P_2$ to move $s^k$ from the suffix to the prefix.

If $P_1=\itnil$ as a result of this operation, then $\cshrink{i}(w_1)$ can be easily seen to be a prefix of $\cshrink{i}(w_2)$,
and thus $w_1$ is a prefix of $w_2$.
Symmetrically, if $P_2=\itnil$, then $w_2$ is a prefix of $w_1$. In both cases we correctly report the final answer.

Thus, at line~\ref{alg:lcp:lcp} $\itstring(P_1)=\itstring(P_2)$ indeed is the longest common prefix of $\cshrink{i}(w_1)$ and $\cshrink{i}(w_2)$.
In particular, for $i=0$ this already allows determining the answer as we do after leaving the \textbf{for} loop.
Otherwise, we infer from Lemma~\ref{lem:lcpd} that the longest common prefix $p$ of $\cshrink{i-1}(w_1)$
and $\cshrink{i-1}(w_2)$ can be represented as $p=p_\ell p_r$ where $\shrink{i}(p_\ell) = \itstring(P_1)=\itstring(P_2)$
and $|\rle(p_r)|\le 1$. We move pointers $P_1$ and $P_2$ to the leftmost children
so that after this operation $p_\ell=\itstring(P_1)=\itstring(P_2)$. The longest
common prefix of $\itsuf(P_1)$ and $\itsuf(P_2)$ is then $p_r$ and thus, as claimed, it
consists of at most one distinct character.

Clearly, the running time of the entire function is proportional to the sum of heights of $\ustree(s_1)$ and $\ustree(s_2)$.
\end{proof}

\subsection{Tries}
Recall that a trie is a rooted tree whose nodes correspond to prefixes of strings in a given family of strings $F$.
If $\alpha$ is a node, the corresponding prefix is called the \emph{value} of the node and denoted $\tval(\alpha)$.
The node $\alpha$ whose value is $\tval(\alpha)$ is called the \emph{locus} of $\tval(\alpha)$.

The parent-child relation in the trie is defined so that the root is the locus of $\eps$,
while the parent~$\beta$ of a node $\alpha$ is the locus of $\tval(\alpha)$ with the last character removed.
This character is the \emph{label} of the edge between $\beta$ and $\alpha$.
If $\beta$ is a ancestor of $\alpha$, then $\tval(\beta,\alpha)$ denotes the concatenation
of edge labels on the path from $\beta$ to $\alpha$. Note that $\tval(\alpha)=\tval(r,\alpha)$ where $r$ is the root.
For two nodes $\alpha,\beta$ of a trie, we denote their lowest common ancestor by $\tlca(\alpha,\beta)$.

A node $\alpha$ is \emph{branching} if it has at least two children and \emph{terminal} if $\tval(\alpha)\in F$.
A \emph{compressed trie} is obtained from the underlying trie by dissolving all nodes except
the root, branching nodes, and terminal nodes. Note that this way we compress paths of nodes with single children
and thus the number of remaining nodes becomes bounded by $2|F|$.

In many applications, including this work, \emph{partially compressed} tries are constructed,
where some \emph{extra} nodes are also important and thus they are not dissolved.
In general, we refer to all preserved nodes of the trie as \emph{explicit} (since they are stored explicitly)
and to the dissolved ones as \emph{implicit}.
Edges of a (partially) compressed trie correspond to paths in the underlying tree
and thus their labels are strings in $\Sigma^+$. Typically, only the first character of these labels is stored explicitly
while the others are assumed to be accessible in a problem-specific implicit way.

\subsection{Overview}
Let $\depth(\coll)=\max_{w\in \coll} \depth(w)$. For $0\le i \le \depth(\coll)$, let $\coll_i = \{w \in \coll : \depth(w)\ge i\}$
and let $\utrie_i(\coll)$ be the trie of strings $\cshrink{i}(w)$ for $w\in \coll_i$.
We write $\utrie_i$ instead of $\utrie_i(\coll)$ if the set $\coll$ is clear from context.
If $w\in \coll_i$, then by $\tleaf_i(w)$ we denote the terminal node $\alpha$ of $\utrie_i(\coll)$ for which
$\tval(\alpha)=\cshrink{i}(w)$.

Observe that for each $i>0$ and each node $\alpha\in \utrie_i$ there exists a unique node $\beta\in \utrie_{i-1}$
such that $\tval(\alpha)=\shrink{i}(\tval(\beta))$. We denote it by $\beta = \tlink(\alpha)$.
Note that $\tlink(\tleaf_i(w))=\tleaf_{i-1}(w)$ and that $\tlink$ preserves the ancestor-descendant relation:

\begin{observation}\label{obs:link}
Let $\alpha,\beta$ be nodes of $\utrie_i$ for $i>0$. If $\beta$ is a (proper) ancestor of $\alpha$, then $\tlink(\beta)$
is a (proper) ancestor of $\tlink(\alpha)$.
\end{observation}

A node $\alpha\in \utrie_i$ is a \emph{down-node} if it is root, branching, or terminal.
A node $\beta\in \utrie_{i}$ is an \emph{up-node} if $\beta=\tlink(\alpha)$ for some down-node
$\alpha\in \utrie_{i+1}$.
We define the (partially) compressed versions $\trie_i(\coll)$ of tries $\utrie_i(\coll)$ in which the explicit
nodes are precisely the up-nodes and the down-nodes. We also assume that each down-node $\alpha\in \utrie_{i}$
for $i>0$ explicitly stores a link to $\tlink(\alpha)$.

Our implementation is going to maintain the tries $\trie_i(\coll)$ subject to $\insertop(w)$  operations.
The actual goal of this is to store $\trie_0$ and augment it with auxiliary components helpful in answering the queries.
In short, to handle $\lcpop$ queries, we use a data structure of Cole and Hariharan~\cite{Cole:2005} for
lowest common ancestor queries in dynamic trees.
For $\compop$ queries, we use an order-maintenance data structure~\cite{Dietz:1987,Bender:2002}
applied for the Euler tour of~$\trie_0$.

%

With the aid of these component answering queries is simple, so our main effort lies in updating the data structure upon insertions $\insertop(w)$.
This is achieved in two phases. First, for each level $i\le \depth(w)$ we determine the \emph{mount} of $w$ in $\utrie_i(\coll)$
which is defined as the lowest node $\beta_i$ of $\utrie_i$ whose value $\tval(\beta_i)$
is a prefix of $\cshrink{i}(w)$. It is easy to see that in $\utrie_{i}(\coll\cup\{w\})$, $\beta_i$ is an ancestor of $\tleaf_i(w)$  (improper if $\cshrink{i}(w)$ is a prefix of $\cshrink{i}(w')$ for some $w'\in \coll$)
and that all internal nodes on the path from $\beta_i$ to $\tleaf_i(w)$ do not exist in $\utrie_i(\coll)$.
Consequently, after determining the mounts, we simply turn each $\trie_i(\coll)$ into $\trie_i(\coll\cup\{w\})$
by adding a new leaf as a child of the mount (unless the mount is already $\tleaf_i(w)$).

Since some of the mounts might be implicit in $\trie_i(\coll)$, we introduce a notion of \emph{trie pointers} to traverse the uncompressed tries
$\utrie_i(\coll)$, so that a trie pointer $\pi$ points to a node $\beta \in \utrie_i$.  Conceptually, they play a similar role to that of pointers to nodes of uncompressed parse trees; see \cref{sec:pointers}.

Let us describe the interface of trie pointers.
First of all, we have a primitive $\troot(i)$ to create a new pointer to the root of $\trie_i$.
Central parts of the implementation are two functions to navigate the tries starting from a node $\beta$ pointed by $\pi$:
\begin{itemize}
  \item $\tdown(\pi,  s, k)$ -- returns a pointer to a node $\alpha$ such that $\tval(\alpha)=\tval(\beta)s^{k'}$
  and $k'\le k$ is largest possible, as well as the optimal $k'$.
  \item $\tlink(\pi)$ -- returns a pointer to $\tlink(\beta)$.
\end{itemize}

Before we proceed with describing the interface for updating $\trie_i$, let us introduce a notion binding trie pointers with pointers to uncompressed parse trees. Let $\beta \in \utrie_i$.
Observe that if $\tval(\beta)$ is a (proper) prefix of $\cshrink{i}(w)$ for some string $w$,
then $\ustree(w)$ contains a unique node $v$ such that $\itstring(v)=\tval(\beta)$.
We say that a pointer $P$ to $\ustree(w)$ is \emph{associated} with $\beta$ (or with a trie pointer $\pi$ pointing to $\beta$) if it points to such a node $v$.  If $\tval(\beta)=\cshrink{i}(w)$, then
there is no such node $v$ and thus we assume that $\itnil$ is associated with $\beta$.

We have two ways of modifying the (partially) compressed tries:
\begin{itemize}
  \item \textbf{Internal node insertion.} $\tmake(\pi)$ -- given a trie pointer $\pi$ to a node $\beta \in \utrie_i$, make $\beta$ explicit in $\trie_i$.
  Returns a reference to $\beta\in \trie_i$.
  \item \textbf{Leaf insertion.} $\taddleaf(\pi,P)$ -- given a trie pointer $\pi$ to the mount of $w$ in $\utrie_i$, and an associated pointer $P$ to $\ustree(w)$,
update $\trie_i$ to represent $\coll\cup\{w\}$. Assumes the mount is already explicit in $\trie_i$. Returns a reference to $\tleaf_i(w)$.
\end{itemize}
Note that these operations may result in $\trie_i$ not being equal to $\trie_i(\coll)$ for any particular $\coll$ (since, e.g., $\trie_i(\coll)$
has a precisely defined set of explicit nodes.
However, we assume that the trie pointer operations (in particular $\tdown$ and $\tlink$) are allowed only if $\trie_i=\trie_i(\coll)$
and (in case of the $\tlink$ operation) $\trie_{i-1}=\trie_{i-1}(\coll)$ for a particular set $\coll$. Moreover, we assume
that any modification of $\trie_i$ invalidates all pointers to~$\utrie_i$ except for pointers to proper descendants of the node pointed by $\pi$.

The following result is proved in \cref{sec:trierepr,sec:triepointers}.
\begin{theorem}\label{thm:trie_pointers}
Tries $\trie_i$ can be implemented so that
internal node insertions, leaf insertions, as well as
$\troot$ and $\tlink$ operations run in worst-case $O(1)$ time, while
$\tdown$ operation works in worst-case $O(\log |\coll|)$ time.
Moreover, any alternating sequence of length $t$ of $\tdown$ and $\tlink$ operations (where the result of one operation is an argument of the subsequent) takes $O(t+\log |\coll|)$ time.
\end{theorem}
\noindent
In \cref{sec:trie_extra} we take care about the auxiliary components in $\trie_0$
and prove the following result:
\begin{lemma}\label{lem:trie_extra}
Auxiliary components of $\trie_0$ can be implemented so that operations $\compop(w,w')$ and $\lcpop(w,w')$ take constant time,
while internal node and leaf insertions require $O(\log |\coll|)$ time.
\end{lemma}

\subsection{$\insertop$ Implementation}
In this section we implement the $\insertop$ operation which adds a string $w$ to the family $\coll$.
This string is given by a signature $s\in \sigs(\grammar)$ representing it.

Recall that the main difficulty lies in determining the mounts $\beta_i$ of $w$ in each trie $\utrie_i(\coll)$
for $0\le i \le \depth(w)$.
The following two results following from \cref{lem:lcpd} provide the required combinatorial background.
\begin{corollary}\label{cor:lcpd}
Let $i>0$, $w,w'\in \coll_i$, $\alpha=\tlca(\tleaf_{i-1}(w),\tleaf_{i-1}(w')))$,
and $\beta = \tlink(\tlca(\tleaf_{i}(w),\allowbreak \tleaf_{i}(w')))$.
Then $\beta$ is an ancestor of $\alpha$ and $|\rle(\tval(\beta,\alpha))|\le 1$.
\end{corollary}
\begin{proof}
  Let $u=\cshrink{i-1}(w)$, $u'=\cshrink{i-1}(w')$,
and let $p$ be the longest common prefix of $u$ and $u'$.
Note that $p = \tval(\alpha)$.
By \cref{lem:lcpd}, we have a factorization $p=p_\ell p_r$
such that $|\rle(p_r)|\le 1$ and $\shrink{i}(p_\ell)$ is the longest common prefix of $\shrink{i}(u)$ and $\shrink{i}(u')$.
Observe that the latter implies that $\tval(\beta)=p_\ell$ and thus that $\beta$ is an ancestor of $\alpha$.
Moreover, $\tval(\beta,\alpha)=p_r$ which yields the second part of the claim.
\end{proof}

\begin{lemma}\label{lem:inscorr}
For $0\le i\le \depth(w)$, let $\beta_i$ be the mount of $w$ in $\utrie_i(\coll)$.
Then $\beta_i$ is a descendant of $\tlink(\beta_{i+1})$ and $|\rle(\tval(\tlink(\beta_{i+1}),\beta_i))|\le 1$.
\end{lemma}
\begin{proof}
Observe that if $\tval(\beta)$ is a prefix of $\cshrink{i+1}(w)$, then $\tval(\tlink(\beta))$
must be a prefix of $\cshrink{i}(w)$. Consequently, $\beta_i$ is indeed a descendant of $\tlink(\beta_{i+1})$.

To prove that $|\rle(\tval(\tlink(\beta_{i+1}),\beta_i))|\le 1$, consider tries built for $\coll\cup \{w\}$.
If $\coll=\emptyset$, then all $\beta_i$ are root nodes, hence $\tlink(\beta_{i+1})=\beta_i$
and the condition clearly holds.
Otherwise, for some $w'\in \coll$ we have  $\beta_{i}=\tlca(\tleaf_{i}(w),\tleaf_{i}(w'))$.
Let $\alpha = \tlca(\tleaf_{i+1}(w),\tleaf_{i+1}(w'))$.
Note that $\beta_{i+1}$ is a descendant of $\alpha$ and, by \cref{obs:link}, $\tlink(\beta_{i+1})$ is consequently a descendant of $\tlink(\alpha)$.
Moreover, \cref{cor:lcpd} implies that $|\rle(\tval(\tlink(\alpha),\beta_i))|\le 1$
and thus $|\rle(\tval(\tlink(\beta_{i+1}),\beta_i))|\le 1$, as claimed.
\end{proof}

\begin{algorithm}
\begin{algorithmic}[1]
\Function{$\insertop$}{$s$}
		\State $L := \slev(s)$
		\State $P'_{L} := \itroot_{\ustree}(s)$
		\State $\pi'_{L} := \troot(L)$
	\For{$i = L$ \textbf{downto} $0$}\label{alg:insert:for}	
		\If{$P'_i \ne \itnil$}\label{alg:insert:if}
    \State $(\pi_i,k) := \tdown(\pi'_i, \itsig(P'_i), \itrext(P'_i))$
		\State $P_i := \itrskip(P'_i,k)$
	\Else
    \State $\pi_i := \pi'_i$
    \State $P_i := \itnil$
	\EndIf
	\If{$i > 0$}
		\State $\pi'_{i-1} := \tlink(\pi_i)$
		\If{$P_i \ne \itnil$}\; $P'_{i-1} := \itchild(P_i, 1)$
		\Else\; $P'_{i-1} := \itnil$
		\EndIf
	\EndIf
\EndFor
\For{$i = 0$ \textbf{to} $L$}
	\State $\beta'_i := \tmake(\pi'_i)$
	\State $\beta_i := \tmake(\pi_i)$
	\State $\tleaf_i(s) := \taddleaf(\pi_i,P_i)$
	\If {$i > 0$}
		\State $\tlink(\beta_i):=\beta'_{i-1}$
		\State $\tlink(\tleaf_i(s)):=\tleaf_{i-1}(s)$
	\EndIf
\EndFor
\EndFunction
\end{algorithmic}
\caption{Pseudocode of $\insertop$ implementation.}\label{alg:insert}
\end{algorithm}

We conclude this section with the actual implementation
and the proof of its correctness.
We extend our notation so that $\tleaf_i(s)=\tleaf_i(w)$
for $s\in\sigs(\grammar)$ representing $w$.

\begin{lemma}\label{lem:trie_insert}
Let $s\in\sigs(\grammar)$ be a signature representing a string $w$.
\Cref{alg:insert} provides a correct implementation of $\insertop(s)$
working in $O(\depth(w)+\log |\coll|)$ time.
\end{lemma}
\begin{proof}
Let us start with the first part of the algorithm, including the \textbf{for} loop at line~\ref{alg:insert:for}.
Let $\beta_i$ be the mount of $w$ in $\utrie_i$.
We claim that the following conditions are satisfied:
 \begin{enumerate}[(a)]
   \item\label{it:one} $\pi_i$ and $\pi'_i$ point to nodes of $\utrie_i$,
   \item\label{it:two} $P'_i$ is a pointer to $\ustree(w)$ associated with $\pi_i'$ and $\pi'_i$ points to an ancestor $\beta'_i$ of $\beta_i$ such that $|\rle(\tval(\beta'_i,\beta_i))|\le 1$.
     \item\label{it:three} $P_i$ is a pointer to $\ustree(w)$ associated with $\pi_i$ and $\pi_i$ points to the mount $\beta_i$.
 \end{enumerate}

Observe that (\ref{it:one}) is clear from definition of each $\pi_i$ and $\pi'_i$.
The remaining two points are proved inductively.
We have $\itstring(P'_L)=\tval(\beta'_L)=\eps$ and $\tval(\beta'_L,\beta_L)=\tval(\beta_L)\in \{\eps,s\}$, where $s$ is the signature representing $w$.
Consequently, condition (\ref{it:two}) is indeed satisfied for $i=L$.
Next, observe that the \textbf{if} block at line~\ref{alg:insert:if} makes (\ref{it:two}) imply (\ref{it:three}).
Indeed, observe that $\tval(\beta'_i,\beta_i)$ is exactly the longest prefix of $\itsuf(P'_i)$ which can be expressed
as $\tval(\beta'_i,\gamma)$ for some descendant $\gamma\in \utrie_i$ of $\beta'_i$.
Consequently, pointers $\pi_i$ and $P_i$ satisfy the claimed conditions.
Finally, if $i>0$ we make $\pi'_{i-1}$ point to $\tlink(\beta_i)$ and $P'_i$ as the associated pointer to $\ustree(w)$.
By \cref{lem:inscorr} condition (\ref{it:two}) is satisfied for $i-1$.
Note that \cref{thm:trie_pointers} implies that this part of the algorithm runs in $O(\depth(w)+\log |\coll|)$
time since we have an alternating sequence of $O(\depth(w))$ subsequent $\tlink$ and $\tdown$ operations.

Correctness of the second part of the algorithm follows from the definition of mounts $\beta_i$.
Its running time is also $O(\depth(w)+\log |\coll|)$ by \cref{thm:trie_pointers,lem:trie_extra}.
\end{proof}

\subsection{Implementation of Tries}\label{sec:trierepr}
In this section we give a detailed description of how the tries $\trie_i$ are represented
and what auxiliary data is stored along with them.
We also show how this is maintained subject to $\tmake$ and $\taddleaf$ updates.

\subsubsection{Trie Structure}

Each edge of $\trie_i$ is represented by a tuple consisting of references to the upper endpoint $\alpha$, the lower endpoint $\beta$,
as well as a pointer $P$ such that $\itsig(P)$ is the first symbol of $\tval(\alpha,\beta)$ and $P$ is associated with $\alpha$.
Each explicit node $\alpha\in \trie_i$ stores a reference to the edge from the parent
and references to all the edges to the children of $\alpha$.
The latter are stored in a dictionary (hash table) indexed by the first character of the edge label.
If $\alpha$ has exactly one child, a reference to the edge to that child is also stored outside of the dictionary
so that the edge can be accessed without knowledge of the first character of its label.

Additionally, each $\alpha\in \trie_i$ stores its depth $\tdepth(\alpha)$, which is defined as the length
of the (unique) string $w$ such that $\tval(\alpha)=\cshrink{i}(w)$.
Note that this can be seen as the weighted length of $\tval(\alpha)$ in which the weight of a symbol $s$ is $\slength(s)$.
In particular, at level $i=0$ we have $\tdepth(\alpha)=|\tval(\alpha)|$.


\newcommand{\tanc}{\mathtt{anc}}
\newcommand{\tdesc}{\mathtt{desc}}
For a node $\alpha\in\utrie_i$ we define $\tanc(\alpha)$ ($\tdesc(\alpha)$) as the nearest ancestor (resp. descendant)
of $\alpha$ which is explicit in $\trie_i$. If $\alpha$ is explicit, we have $\tanc(\alpha)=\tdesc(\alpha)=\alpha$.


\begin{lemma}\label{lem:structure}
The trie structure can be updated in constant time subject to internal node insertion $\tmake(\pi)$ and leaf insertion $\taddleaf(\pi,P)$.
\end{lemma}
\begin{proof}
First, let us focus on internal node insertion. Let $\beta\in\utrie_i$ be the node pointed by $\pi$.
Let $\pi = (\alpha,P)$ and let the edge from $\alpha$ to its parent be represented as $(\gamma,\alpha,P')$.
Note that this edge needs to be split into two edges: $(\gamma,\beta,P')$ and $(\beta,\alpha,P)$.
We thus update a reference to the parent edge at $\alpha$
and a reference to the appropriate child edge at $\gamma$.
Recall that the latter is stored in the dictionary under the $\itsig(P')$ key.

Finally, we need to create the representation of $\beta$.
This is mainly a matter of storing the references to the newly created edges.
The lower one is kept in the dictionary at $\beta$ under key $\itsig(P)$.
We also store $\tdepth(\beta)$ which can be easily determined from $\itrepr(P)$.
Consequently, internal node insertion can indeed be implemented in constant time.

Let us proceed with leaf insertion. If $P=\itnil$, there is nothing to do since $\beta=\tleaf_i(w)$.
Otherwise, we need to create a new node $\lambda=\tleaf_i(w)$ and a new edge $(\beta,\lambda,P)$.
References to that edge need to be stored in $\tleaf_i(w)$ and in the dictionary at $\beta$.
In the latter we insert it under $\itsig(P)$ key. By our assumption,
no such edge has previously existed. The last step is to store $\tdepth(\tleaf_i(w))=|w|$.
\end{proof}

\newcommand{\tpath}{\mathtt{path}}
\subsubsection{Navigation on Uniform Paths.}
Let $\Pi$ be a path in $\trie_i$ from an explicit node $\gamma$ to its proper explicit descendant $\alpha$.
We say that $\Pi$ is \emph{$s$-uniform} if $\tval(\gamma,\alpha)=s^k$ for some symbol $s$ and integer $k$.
An $s$-uniform path is \emph{maximal} if it cannot be extended to a longer $s$-uniform path.
Observe that maximal uniform paths are edge-disjoint.
If $\alpha$ is a child of $\gamma$ in $\trie_i$ we say that the edge from $\gamma$ to $\alpha$ is $s$-uniform.
Detection of uniform edges is easy:
\begin{fact}\label{fct:unicheck}
Given a representation $(\gamma,\alpha,P)$ of an edge of $\trie_i$, we can check in $O(1)$ time if
the edge is uniform.
\end{fact}
\begin{proof}
Note that the edge could only be $s$-uniform for $s=\itsig(P)$.
Observe that this is the case if and only if
$$\tdepth(\alpha)-\tdepth(\gamma) \le \itrext(P)\cdot \slength(s).$$
This is due to the fact that $\tval(\gamma,\alpha)$ is a prefix of $\itsuf(P)$.
\end{proof}

\newcommand{\tlocate}{\mathtt{locate}}
With a path $\Pi$ in $\trie_i$ we associate the uncompressed path $\overline{\Pi}$ in $\utrie_i$
with the same endpoints.
For each uniform path $\Pi$ we would like to efficiently locate in $\Pi$ any node of $\overline{\Pi}$.
More precisely, we consider the following $\tlocate_{\Pi}(d)$ operation: if $\overline{\Pi}$ contains
a node $\beta$ with $\tdepth(\beta)=d$, compute $\tdesc(\beta)$ and $\tanc(\beta)$; otherwise return $\nilst$.
The path $\Pi$ can be specified by any edge belonging to it.

We introduce a function $$\phi(\beta)=\min(|\tval(\tanc(\beta),\beta))|,|\tval(\beta,\tdesc(\beta))|)$$
to measure hardness of a particular query.

\begin{lemma}\label{lem:st}
For all maximal uniform paths $\Pi$ in $\trie_i$ we can maintain data structures which support
$\tlocate_{\Pi}(d)$ queries in $O(1+\min(\phi(\beta),\log |\Pi|))$ time if there is a sought node $\beta$
and in $O(1)$ time otherwise.
Moreover, the data structures (as well the collection of maximal uniform paths) can be maintained
in constant time subject to internal node insertions and leaf insertions.
\end{lemma}
\begin{proof}
First, observe that the depths of all nodes on $\overline{\Pi}$ form an arithmetic progression whose difference is $\slength(s)$.
Thus, it suffices to know the depths $\tdepth(\gamma)$ and $\tdepth(\alpha)$ of endpoints and the length $\slength(s)$
to test if the sought node exists.
Consequently, for each path $\Pi$ we shall store these three quantities. If each uniform edge stores a pointer to the data structure responsible for the path containing that edge, we can make sure  that the path $\Pi$ can indeed by specified by one of the member edges.

To support queries with non-$\nilst$ answers, we shall maintain the depths of all nodes of $\Pi$ in two forms:
a finger search tree~\cite{BrodalLMTT03} and a dictionary (hash table) mapping a value to a finger pointing to the corresponding element of the search tree.
This way, we can find $\tanc(\beta)$ or $\tdesc(\beta)$
in $O(1+\phi(\beta))$ time by checking if the dictionary has an entry at $d\pm k \cdot \slength(s)$ for subsequent integers $k$ starting from $0$.
Once we detect $\tanc(\beta)$ or $\tdesc(\beta)$, the other can be determined in constant time since finger search tree supports
constant-time iteration.

Our query algorithm performs this procedure but stops once we reach $k\ge \log |\Pi|$. In this case it falls back to ordinary search
in the finger search tree, which clearly succeeds in $O(\log |\Pi|)$ time.
Consequently, the query algorithm runs in $O(1+\min(\phi(\beta),\log |\Pi|))$ time, as claimed.

Thus, it suffices to describe how to support internal node insertions and leaf insertions.
First, let us focus on internal node insertions. If the newly created node $\beta$ subdivides a uniform edge,
we shall insert its depth to the data structure of the corresponding path $\Pi$.
Since we know $\tanc(\beta)$ and $\tdesc(\beta)$, i.e., the predecessor and the successor of $\tdepth(\beta)$ in the finger search tree,
this element insertion takes constant time.

Otherwise, we need to check if the two newly created edge are uniform and, if so, if they extend an already existing uniform paths
or constitute new ones. All this can be easily accomplished using \cref{fct:unicheck}.
A new path $\Pi$ contains two nodes, so its data structure can be easily created in constant time.
The extension of an already existing path, on the other hand, reduces to insertion of a new depth larger or smaller than
all depths already stored. Again, finger search tree supports such insertions in constant time.

Finally, for leaf insertions we create one edge, and if it is uniform, we need to update our data structures.
We proceed exactly as in the second case for internal node insertions, since this edge either extends
an already existing path or constitutes a new one.
\end{proof}

\subsection{Implementation of Operations on Trie Pointers}\label{sec:triepointers}
In this section we implement $\tdown$ and $\tlink$ operations.
Recall that these operations require trees $\trie_i$ to be $\trie_i(\coll)$
for a fixed collection $\coll$.
In particular, we insist on explicit nodes being precisely the down-nodes and the up-nodes.

\subsubsection{Combinatorial Toolbox}

For a node $\alpha\in \utrie_i$, let us define $\tsub(\alpha)$ as the set of all strings $w\in \coll_i$
for which $\tleaf_i(w)$ is a descendant of $w$.

\begin{fact}\label{fct:link}
Let $i>0$ and let $\alpha,\beta$ be nodes of $\utrie_i$ such that $\alpha$
is a proper descendant of $\beta$.
The last character of $\tval(\tlink(\beta))$ is not the same as the first character of $\tval(\tlink(\beta),\tlink(\alpha))$.
Moreover, $|\rle(\tval(\tlink(\beta),\tlink(\alpha)))|\ge |\tval(\beta,\alpha)|$
\end{fact}
\begin{proof}
Let $\alpha'=\tlink(\alpha)$, $\beta'=\tlink(\beta)$, $w_{\alpha}=\tval(\alpha')$
and $w_{\beta}=\tval(\beta')$. The fact that $\shrink{i}(w_{\beta})$ is a proper prefix of $\shrink{i}(w_{\alpha})$
also means that $\shrink{i}$ places a block boundary in $w_{\alpha}$ just after $w_{\beta}$.
Since a boundary is never placed between two equal adjacent characters, this implies that the last character of $\tval(\tlink(\beta))$
is indeed different from the first character of $\tval(\tlink(\beta),\tlink(\alpha))$.

The second part of the claim easily follows from $\tval(\beta,\alpha)=\shrink{i}(\tval(\tlink(\beta),\tlink(\alpha)))$
and $|\shrink{i}(w)|\le |\rle(w)|$ for any integer $i$.
\end{proof}

\begin{lemma}\label{lem:nodown}
Let $\beta\in \utrie_i$ for $i>0$ and let $\alpha$ be the nearest descendant of $\beta$
which is a down-node.
The path $\tlink(\beta) \leadsto \tlink(\alpha)$ has no internal explicit nodes.
\end{lemma}
\begin{proof}
If $\beta$ is a down-node, then $\alpha=\beta$ and thus the claim holds trivially.
Consequently, we shall assume that $\beta$ is not a down-node.
Let $\beta'=\tlink(\beta)$ and $\alpha'=\tlink(\alpha)$.
Let us also define $\gamma$ as the nearest ancestor of $\beta$ being a down-node,
and $\gamma'=\tlink(\gamma)$.

By \cref{obs:link}, $\gamma'$ is a proper ancestor of $\beta'$ and $\beta'$ of $\alpha'$.
Moreover, by \cref{fct:link}, the last character of $\tval(\gamma',\beta')$ is distinct from the first character of $\tval(\beta',\alpha')$.

For a pair of fixed strings $w_1,w_2\in \coll$ let us define $\delta = \tlca(\tleaf_{i-1}(w_1),\tleaf_{i-1}(w_2))$ and
$\delta'=\tlink(\tlca(\tleaf_i(w_1),\tleaf_i(w_2)))$.
We shall prove that neither $\delta$ nor $\delta'$ is an internal node of $\beta'\leadsto \alpha'$.
Since any down-node is of the former form and any up-node of the latter form, this shall conclude the proof.
Note that (by \cref{cor:lcpd}) $\delta'$ is an ancestor of $\delta$.

First, suppose that $w_1,w_2 \in \tsub(\alpha)$. Then, $\tlca(\tleaf_i(w_1),\tleaf_i(w_2))$
is a descendant of $\alpha$ and thus, by \cref{obs:link}, $\delta'$ is a descendant of $\alpha'$.
Consequently, indeed neither $\delta$ nor $\delta'$ can be a proper ancestor of $\alpha'$.

Thus, one of the strings $w_1$, $w_2$ (without loss of generality $w_1$) does not belong to $\tsub(\alpha)$.
Let us fix a string $w\in \tsub(\alpha)$ and define $\rho = \tlca(\tleaf_i(w),\tleaf_i(w_1))$ as well as $\rho'=\tlca(\tleaf_{i-1}(w),\tleaf_{i-1}(w_1))$.
Note that $\rho$ is a proper ancestor of $\alpha$ and, since the path $\gamma\leadsto \alpha$ contains no internal down-nodes,
an ancestor of $\gamma$. Consequently, by \cref{obs:link}, $\tlink(\rho)$ is an ancestor of $\gamma'$.
By \cref{cor:lcpd}, $\tlink(\rho)$ is an ancestor of $\rho'$ and $|\rle(\tval(\tlink(\rho),\rho'))|\le 1$.
Since the last character of $\tval(\gamma',\beta')$ is distinct from the first character of $\tval(\beta',\alpha')$,
this implies that $\rho'$ is an ancestor of $\beta'$. Since $w\in \tsub(\alpha')$, this yields
that no ancestor of $\tleaf_{i-1}(w_1)$ is an internal node of the $\beta'\leadsto \alpha'$ path.
This in particular holds for both $\delta$ and $\delta'$.
\end{proof}

\begin{lemma}\label{lem:oneup}
Let $\gamma,\alpha$ be down-nodes in $\utrie_i$ such that $\gamma$ is a proper ancestor of $\alpha$
and the path $\gamma \leadsto \alpha$ contains no internal down-nodes.
This path contains at most one internal up-node.
\end{lemma}
\begin{proof}
Suppose that $\beta'\in \utrie_{i+1}$ is a down-node such that $\beta=\tlink(\beta')$
is an internal up-node on the path $\gamma \leadsto \alpha$.
Let $w,w'\in \coll$ be strings such that $\beta' = \tlca(\tleaf_{i+1}(w),\tleaf_{i+1}(w'))$.
Since $\beta$ is not a down-node, $\beta'$ is a proper ancestor of $\tleaf_{i+1}(w)$.
Let $\delta = \tlca(\tleaf_{i}(w),\tleaf_{i}(w'))$. By~\cref{cor:lcpd},
$\delta$ is a descendant of $\beta$ and $|\rle(\tval(\beta,\delta))|\le 1$.
Since $\delta$ is a down-node, this means that $\delta$ is a descendant of $\alpha$
and, in particular, $|\rle(\tval(\beta,\alpha))|=1$.

By \cref{fct:link} applied for $\beta'$ and $\tleaf_{i+1}(w)$, the last character of $\tval(\beta)$
(i.e., the last character of $\tval(\gamma,\beta)$)
is distinct from the first character of $\tval(\beta,\tleaf_{i}(w))$ (i.e., the first character of $\tval(\beta,\alpha)$).
Consequently, $\beta$ is uniquely determined as the highest internal node of $\gamma \leadsto \alpha$
satisfying $|\rle(\tval(\beta,\alpha))|=1$.
\end{proof}

\subsubsection{Implementation Details}

A trie pointer $\pi$ to a node $\beta\in\utrie_i$ is implemented as a pair $(\tdesc(\beta),P)$,
where $P\neq\itnil$ if and only if $\beta\in\trie_i$, i.e. $\beta$ is an explicit node.
Otherwise, $P$ points to a node of $\ustree(w)$ associated with $\beta$, for some $w\in\tsub(\beta)$.

\begin{lemma}\label{lem:link}
Given a trie pointer $\pi$, the pointer $\tlink(\pi)$
can be computed in $O(1)$ time.
\end{lemma}
\begin{proof}
Let $\beta\in \utrie_i$ be the node represented by $\pi$ and let $\pi=(\tdesc(\beta),P)$.
If $\beta$ is a down-node, then a reference to $\tlink(\beta)$ is stored at $\beta$ and $\tlink(\beta)$ is explicit, therefore $(\tlink(\beta), \itnil)$ can be returned as the resulting trie pointer.

If $\beta$ is an up-node, a pointer $P$ to an associated node of $\ustree(w)$ for some $w\in \coll$
can be read from the representation of the edge going from $\beta$ to its only child.
Thus, we shall assume that $\beta$ is not explicit and thus $P\ne \itnil$.

Let $\alpha$ be the nearest descendant of $\beta$ which is a down-node.
By \cref{lem:oneup}, $\alpha=\tdesc(\beta)$ or $\alpha$ is the only child of $\tdesc(\beta)$ in $\trie_i$.
In either case, $\alpha$ can be determined in constant time and so can be $\alpha':=\tlink(\alpha)$,
because down-nodes explicitly store their link pointers.

Let $\beta'=\tlink(\beta)\in \utrie_{i-1}$ be the sought (possibly implicit) node.
By \cref{lem:nodown}, $\tpath(\beta',\alpha')$ path contains no explicit internal nodes.
Thus, either $\tdesc(\beta')=\alpha'$ or $\beta'$ is the (explicit) parent of $\alpha'$ in $\trie_{i-1}$.
Note that the latter holds exactly if the parent of $\alpha'$ is at depth exactly $|\itstring(P)|$.
Since nodes of $\trie_{i-1}$ store their depth and since $|\itstring(P)|$ can be easily retrieved from $\itrepr(P)$,
this condition can be checked in constant time. We return $(\tdesc(\beta'),\itchild(P,1))$ in either case.
\end{proof}

\begin{lemma}\label{lem:down}
  Given a trie pointer $\pi$, a signature $s\in \sigs(\grammar)$, and an integer $k$, the pair $(\pi',k')=\tdown(\pi,s,k)$
can be computed in $O(1+\min(\phi(\beta'),\log |\coll|))$ time where $\beta'$ is the node $\beta'$ pointed by $\pi'$
and $\phi(\beta')=\min(|\tval(\tanc(\beta'),\beta')|,|\tval(\beta',\tdesc(\beta'))|)$.
\end{lemma}
\begin{proof}
Let $\beta\in \utrie_i$ be the node represented by $\pi$ and let $\pi=(\tdesc(\beta),P)$.
First, suppose that $\beta$ is explicit.
If $\beta$ does not have a child $\gamma$ with edge $(\beta,\gamma,Q)$ such that $s=\itsig(Q)$, we return $(\pi',k')=(\pi,0)$.
Otherwise, the dictionary lets us retrieve $\gamma$ and an associated pointer $Q$.
First, suppose that the edge is not $s$-uniform (which can be detected using \cref{fct:unicheck}).
In this case $\tdesc(\beta')=\gamma$ and it is not hard to see that $k'=\itrext(Q,k)$
and $\pi'=(\gamma,\itrskip(Q,k'))$.

If the edge is uniform, we extend it to the maximal $s$-uniform path $\Pi$
and run the \linebreak
$\tlocate_{\Pi}(\tdepth(\beta)+k\cdot \slength(s))$ operation of \cref{lem:st}.
Let $\gamma'$ be the lower endpoint of $\Pi$. If $\beta'$ is not a proper descendant
of $\gamma$, this way we locate $\tanc(\beta')$ and $\tdesc(\beta')$.
The associated pointer can be easily located as
$\itrskip(Q,\lfloor\frac{\tdepth(\beta)}{\slength(s)}\rfloor+k-\lfloor\frac{\tdepth(\tanc(\beta'))}{\slength(s)}\rfloor)$,
where $Q$ is the pointer stored in an edge $(\tanc(\beta'),\tdesc(\beta'),Q)$.

Otherwise, we recursively proceed from $\gamma'$ trying to follow $\lfloor\frac{\tdepth(\beta)}{\slength(s)}\rfloor+k-\lfloor\frac{\tdepth(\gamma')}{\slength(s)}\rfloor$
copies of $s$. This time, we are however guaranteed that reaching a child of $\gamma'$ is impossible and
thus the previous case immediately yields a result.

Finally, suppose that $\beta$ is implicit and let $\gamma=\tdesc(\beta)$.
First, we check if $\beta'$ is an ancestor of $\gamma$.
This happens if $k':=\min(k,\itrext(P))$ does not exceed $\frac{\tdepth(\gamma)-\tdepth(\beta)}{\slength(s)}$.
Since $\tdepth(\gamma)$ is stored and $\tdepth(\beta)=|\itstring(P)|$ can be retrieved from $\itrepr(P)$,
this condition can be checked in constant time. If the result is positive, we return $(\pi',k')$
where $\pi'=(\gamma,\itrext(P,k'))$.
Otherwise, we recursively proceed trying to follow $k-\frac{\tdepth(\gamma)-\tdepth(\beta)}{\slength(s)}$ copies
of $s$ starting from $\gamma$. As before, one of the previous cases are used for $\gamma$,
so the procedure terminates in at most two steps.
\end{proof}

\begin{lemma}\label{lem:amort}
An alternating sequence of $\tlink$ and $\tdown$ operations (such that the result of a given operation
is the argument of the subsequent one) runs in $O(t+\log|\coll|)$ time where $t$ is the length of the sequence.
\end{lemma}
\begin{proof}
Let us define two potential functions.
To analyze situation after $\tdown$ (before $\tlink$) we use
 $$\phi(\beta)=\min(|\tval(\tanc(\beta),\beta)|,|\tval(\beta,\tdesc(\beta))|),$$
 while after $\tlink$ (before $\tdown$)
 $$\phi'(\beta)=\min(|\rle(\tval(\tanc(\beta),\beta))|,|\rle(\tval(\beta,\tdesc(\beta)))|).$$

We shall prove that no single operation decreases the potential by more than one
and that the slow, superconstant-time case of $\tdown$ is required only if the potential is
increased from $0$ or $1$ to at least a value proportional to the running time.

Let us analyze the $\tdown$ operation. First, suppose that $\phi'(\beta)>1$.
Then, $\beta$ is implicit and, since $|\rle(\tval(\beta,\tdesc(\beta)))|>1$,
the resulting node $\beta'$ is also an implicit ancestor of $\tdesc(\beta)$.
Moreover,
$$|\tval(\beta',\tdesc(\beta')|\ge |\rle(\tval(\beta',\tdesc(\beta')))|\ge |\rle(\tval(\beta,\tdesc(\beta)))|-1$$
and
$$|\tval(\tanc(\beta'),\beta')|\ge |\rle(\tval(\tanc(\beta'),\beta'))|\ge |\rle(\tval(\tanc(\beta),\beta))|,$$
i.e., $\phi(\beta')\ge \phi(\beta)-1$ as claimed.
Finally, recall that $\tdown$ works in constant time in this case.

Thus, it suffices to prove that the if $\tdown$ in \cref{lem:down} requires superconstant time,
it the running time of $\tdown$ is bounded by the potential growth.
 This is, however, a straightforward consequence of the formulation of \cref{lem:down}.

 Consequently, it suffices to analyze $\tlink$. We shall prove that $\phi'(\tlink(\beta))\ge \phi(\beta)-1$ for every $\beta\in \utrie_i$ and $i>0$.
 If $\beta$ is explicit, there is nothing to prove.
 Let $\alpha$ be the nearest descendant of $\beta$ which is a down-node, and let $\gamma$ be the lowest ancestor of $\beta$
 whose parent (in $\utrie_i$) is a down-node. Also, let $\alpha'=\tlink(\alpha)$, $\beta'=\tlink(\beta)$, and $\gamma'=\tlink(\gamma)$.
 By \cref{lem:nodown} applied to $\gamma$, there is no explicit inner node on the $\gamma'\leadsto \alpha'$ path.
 Thus, by the second part of \cref{fct:link}, we have
 $$|\rle(\tval(\tanc(\beta'),\beta'))|\ge|\rle(\tval(\gamma',\beta'))|\ge |\tval(\gamma,\beta)|\ge |\tval(\tanc(\beta),\beta)|-1$$
 and
 $$|\rle(\tval(\beta',\tdesc(\beta')))|\ge |\rle(\tval(\beta',\alpha'))|\ge |\tval(\beta,\alpha)|\ge |\tval(\beta,\tdesc(\beta))|.$$
 Consequently, $\phi'(\beta')\ge \phi(\beta)-1$, as claimed.
\end{proof}

Note that \cref{lem:structure,lem:st,lem:amort} immediately yield \cref{thm:trie_pointers}.

\subsection{Auxiliary Components for Answering Queries}\label{sec:trie_extra}
In this section we give a detailed description of auxiliary data structures stored with $\trie_0$
in order to provide constant-time answers to $\lcpop$ and $\compop$ queries.

First of all, for each $w\in \coll$ we store a reference to $\tleaf_0(w)$ so that given a handle of $w$
we can get to the corresponding terminal node in constant time.
Observe that $\lcpop(w_1,w_2)=\tdepth(\tlca(\tleaf_0(w_1),\allowbreak\tleaf_0(w_2)))$.
Since every explicit node stores its depth, answering $\lcpop$ queries reduces to $\tlca$ queries on $\trie_0$.
Consequently, we maintain for $T_0$ a data structure of Cole and Hariharan~\cite{Cole:2005} for LCA queries in dynamic trees.
It allows $O(1)$ time queries and supports constant-time internal node insertions and leaf insertions as one of the several possible update operations.

\newcommand{\tbeg}{\mathtt{beg}}
\newcommand{\tend}{\mathtt{end}}

Supporting $\compop$ queries is a little bit more challenging. For this, we treat $\utrie_0$ as an ordered trie (based on the order on $\Sigma$)
and we maintain an Euler tour the trie. It is formally a list $L$ in which every explicit node $\alpha$ corresponds to exactly two elements $\tbeg(\alpha)$ and $\tend(\alpha)$. The order on $L$ is such that a description of a subtree rooted at $\alpha$ starts with $\tbeg(\alpha)$, ends with $\tend(\alpha)$ and between these elements contains the description of subtrees of $\alpha$, ordered from left to right.
In this setting, $\tval(\alpha)<\tval(\beta)$ if and only if $\tbeg(\alpha)$ precedes $\tbeg(\beta)$ in $L$.
Hence, in order to answer $\compop$ queries it suffices to store an order-maintenance data structure~\cite{Dietz:1987,Bender:2002} for $L$.

Such a data structure allows inserting a new element to the list if its predecessor or successor is provided.
Thus, when we add a new node to $\trie_0$, we need to determine such a predecessor or successor. For this, in each explicit node of $\trie_0$
we maintain an AVL tree storing references to edges going to the children. The AVL is ordered by the first characters of edge labels.
Hence, whenever we modify the hash table storing these references, we also need to update the AVL tree. This happens $O(1)$ times
for each node insertion, and thus takes $O(\log |\coll|)$ time. Now, for a newly added node $\beta$
the predecessor of $\tbeg(\beta)$ can be determined as follows: If $\beta$ has a preceding sibling $\alpha$, it is $\tend(\alpha)$.
Otherwise, it is $\tbeg(\gamma)$ where $\gamma$ is the parent of $\beta$. A successor of $\tend(\beta)$ can be retrieved analogously.
Determining the preceding or succeeding sibling of $\beta$ is implemented using the AVL tree of children of $\gamma$
and takes $O(\log |\coll|)$ time.
Consequently, the component for $\compop$ allows $O(1)$-time queries and requires $O(\log |\coll|)$ time per update of $\trie_0$.

The discussion above proves \cref{lem:trie_extra}.

\subsection{Conclusions}\label{sec:trie_conc}

\begin{theorem}\label{thm:persistent_compare}
The data structure described in \cref{sec:collection} can be augmented to support $\lcpop$ and $\compop$
queries in worst-case constant time (without failures). The running time of updates $\makeop$, $\concop$, and $\splitop$ increases by a constant factor only.
\end{theorem}

\begin{proof}
After each update we perform $\insertop$ to insert the newly created string to the order-maintenance components.
Its running time $O(\log |\coll|)$ is dominated by the $O(\log t)$ term in the update times.
\end{proof}

\section{Pattern Matching in a Set of Anchored Strings}\label{sec:anchored}
In this short section we relate 2D range reporting with pattern matching.
We adapt this rather folklore technique to our setting,
in which strings can be compared in constant time,
but do not have integer identifiers consistent with the lexicographic order.
Alstrup et al.~\cite{Alstrup} used very similar tools.
Subsequent results of Mortensen~\cite{Mortensen:2003}, however, let us slightly improve the update times.

\begin{definition}
  We call a pair $(s_1,s_2)$ of strings from $\Sigma^+$ an \emph{anchored}
  string.
  An anchored string $(p_1,p_2)$
  is called a \emph{substring}
  of another anchored string $(s_1,s_2)$ if and only if
  $p_1$ is a suffix of $s_1$ and $p_2$ is a prefix of $s_2$.
\end{definition}

Anchored strings can also be seen as strings containing exactly
one special character $|\notin \Sigma$ -- an anchor.
Thus, a pair $(s_1,s_2)$ can be seen as a string $s_1|s_2$ and
$(p_1,p_2)$ is a substring of $(s_1,s_2)$ if and only if
$p_1|p_2$ is a substring of $s_1|s_2$.

Assume that we have a set $\coll$ of strings such that for any two elements $w_1,w_2$
of $\coll$ we can in constant time find out if $w_1$ is lexicographically smaller than $w_2$
and whether $w_1$ is a prefix of $w_2$ (the latter check can be reduced to verifying if $\lcpop(w_1,w_2)=|w_1|$).
Moreover, suppose that $\coll$ is closed under string reversal, i.e., $w^R\in \coll$ for every $w\in \coll$.
We wish to maintain a dynamic multiset of anchored strings $A\subseteq \coll\times \coll$
so that we can find
interesting elements of $A$ containing some given anchored string.
More formally, we want to support the operation $\anchfind(p_1,p_2,occ)$,
where $p_1,p_2\in \coll$, reporting $occ$ elements $a\in A$ such that $(p_1,p_2)$
is a substring of $a$.
If there is less than $occ$ such elements, $\anchfind(p_1,p_2,occ)$ reports
all of them.
In order to efficiently distinguish elements of $A$, we use a unique integer key
for each of them. Keys are provided while we perform an insertion and are later used
to identify the element. In particular, they are reported by $\anchfind$
and deletion uses them as an argument.

\begin{theorem}[Anchored index]\label{thm:ai}
  Let $n$ be the size of the multiset $A$.
  There exists a 
data structure supporting the $\anchfind(*,*,k)$ operation
in worst-case time $O(\log{n}+k)$ and insertion or deletion
to the set $A$ in worst-case time $O(\log{n})$.
\end{theorem}
\begin{proof}
We reduce our case to the problem of maintaining the set
two-dimensional points on plane, so that
we can report a points within some
orthogonal rectangle.
Mortensen~\cite{Mortensen:2003} proposed a data structure
for the problem of dynamic two-dimensional orthogonal
range reporting with worst-case update time of $O(\log{n})$
and query time of $O(\log{n}+k)$, where
$n$ is the size of the currently stored set of points.
Note that the data structure only assumes comparisons on the coordinates
with worst-case $O(1)$ time overhead, as opposed to an explicit
mapping of coordinates to integers.

In the reduction we introduce a new letter $\lastsymb$ larger than all letters in $\Sigma$.
We observe that for two strings $w,w'\in \coll$ we can efficiently lexicographically compare $w$ and $w\lastsymb$ with $w'$ and $w'\lastsymb$.
This is because the order between $w\lastsymb$ and $w'\lastsymb$ is the same as between $w$ and $w'$,
while $w\lastsymb$ is smaller than $w'$ whenever $w$ is smaller than $w'$ but is not a prefix of $w'$.

The reduction of anchored strings to two-dimensional
points works as follows: the element of $(s_1,s_2)\in A$
is mapped to a point $(s_1^R,s_2)$.
In order to answer $\anchfind(p_1,p_2,k)$, where $p_1,p_2\in \coll$, we
query the orthogonal range data structure for
$k$ points in a rectangle
$[p_1^R,p_1^R\lastsymb]\times [p_2,p_2\lastsymb]$.
Clearly, $p_1^R$ is a prefix of $s_1^R$ if $s_1^R$ is both
lexicographically not smaller than $p_1^R$ and not greater than $p_1^R\lastsymb$.
Similarly, $p_2$ is a prefix of $s_2$ if and only if
$s_2$ is not smaller than $p_2$ and not greater than $p_2\lastsymb$.

As the data structure of Mortensen only supports sets of points,
we also maintain a mapping $\mathcal{M}$ from integer keys to
points, and for each distinct point $q=(s_1^R,s_1)$ we maintain
a list of keys $\mathcal{L}_q$ mapped to $q$.
We insert a point $q$ into the range reporting structure only
when $|\mathcal{L}_q|=1$ after the insertion
and delete $q$ when $\mathcal{L}_q$
becomes empty.

For each point $q$ reported by the range-reporting structure
we actually report all keys from $\mathcal{L}_q$, which does
not increase the asymptotic reporting time of $O(\log{n}+k)$.
\end{proof}

\section{Pattern Matching in a Non-Persistent Dynamic Collection}\label{sec:pm}
In this section we develop a dynamic index, i.e., a dynamic collection
of strings, which supports pattern matching queries.
As opposed to the data structure of \cref{sec:collection}, the dynamic index is not persistent,
i.e., $\concop$ and $\splitop$ destroy their arguments.
Formally, we maintain a multiset $\coll$ of strings that is initially empty.
The index supports the following updates:
\begin{itemize}
  \item $\makeop(\str)$, for $\str\in \Sigma^+$, results in $\coll:=\coll\cup \{\str\}$.
      \item $\dropop(\str)$, for $\str\in \coll$, results in $\coll:=\coll\setminus \{\str\}$.
   \item $\concop(\str_1,\str_2)$, for $\str_1,\str_2\in \coll$, results in $\coll:=\coll\cup \{\str_1\str_2\}\setminus \{w_1,w_2\}$.
  \item $\splitop(\str,k)$, for $\str\in \coll$ and $1\leq k<|\str|$, results in
    $\coll:=\coll\cup \{\str[..k], \str[(k + 1)..]\}\setminus\{w\}$.
\end{itemize}

Again, strings in $\coll$ are represented by handles so that input to each operation takes constant space.
Since $\coll$ is a multiset, each string created by an update is assigned a fresh handle.
Hence, as opposed to the data structure of \cref{sec:collection}, equal strings may have different handles.

The data structure supports pattern matching queries for (external) patterns $p$.
More precisely, $\findop(p)$ for $p\in \Sigma^+$ determines
if $p$ occurs in some string $w\in \coll$.
Optionally, it also reports all the occurrences as \emph{stubs}, which can be turned to an actual occurrence
(a handle of a particular $w\in \coll$ plus a position within $w$) at an extra cost.
We could modify the implementation of the $\findop$ operation to report, for a given threshold $\tau$,
no more than $\tau$ arbitrary occurrences. However, the choice of the occurrences reported would depend on the internal randomness
of our data structure. Hence, one would need to make sure that future operations do not depend on this choice.

In the description we assume that $n$ is the total input size of the updates made so far.
Since $\splitop$ and $\makeop$ destroy their arguments, $n$ is also an upper bound on the total
length of strings in $\coll$ (both after and before the current update).
From \cref{lem:small_depth} it follows that a string of length at most $n$ has $\depth(\str) = O(\log n)$ with high probability.
Thus, our update algorithms fail as soon as they encounter a string of larger depth.
Since the running time of each update is $\poly(n)$, each update is successful with high probability.

%

Our data structure consists of separate components which we use for long (longer than $\tau$) and for short patterns (shorter than $\tau$).
The data structure for long patterns described in \cref{sec:long} is similar to the persistent data structure of \cref{sec:collection,sec:om}.
Since our approach in this case results in an additive overhead of $\tau := \log^2 n$ in the query, for short patterns it is
more efficient to use a tailor-made solution, which we present in \cref{sec:short}.
By maintaining both data structures and querying the appropriate one, we obtain a running time, which is only \emph{linear} in the length of the queried string.

\newcommand{\lists}{\mathcal{L}}

\subsection{Permanent Pointers to Characters of $\coll$}\label{sec:lists}
A common ingredient of our solutions for both long and short patterns is a data structure $\lists(\coll)$
which provides pointers to specific characters of strings $w\in \coll$.
These pointers are preserved in the natural way while $\concop$ and $\splitop$ operations are performed.
This is actually the reason why $\concop$ and $\splitop$ destroy their arguments: pointers to characters
are reused.

\begin{lemma}\label{lem:lists}
We can maintain a data structure $\lists(\coll)$ providing pointers to characters of strings $w\in \coll$.
This data structure supports the following updates: $\makeop(w)$ and $\dropop(w)$ in $O(|w|)$ time,
$\splitop(w)$ in $O(\log |w|)$ time, and $\concop(w_1,w_2)$ in $O(\log|w_1w_2|)$ time.
Moreover, given a handle to $w\in \coll$ and a position $i$, $0<i\le |w|$, in $O(\log |w|)$ time it returns
a pointer to $w[i]$. Converting a pointer into a handle and a position, also works in $O(\log |w|)$ time.
\end{lemma}
\begin{proof}
We store an AVL tree for each $w\in \coll$ and a mapping between the handle of strings $w\in \coll$ and roots of the trees. The $i$-th (in the in-order traversal) node of the tree corresponding to $w\in \coll$
represents $w[i]$. After augmenting the trees so that each node stores the size of its subtree, we
can determine the node representing $w[i]$ in $O(\log |w|)$ time.

For $\makeop(w)$ and $\dropop(w)$ we simply create or destroy a tree corresponding to $w$.
Operations $\splitop$ and $\concop$ reduce to splitting and joining AVL trees, which can be done in logarithmic time.
The former also requires finding the $i$-th node, with respect to which we perform the split.

We implement pointers to characters simply as pointers to the corresponding nodes of the AVL tree.
This way determining the handle and the position only requires traversing a path from the pointed node to the root of the tree,
and determining a pointer to the given $w[i]$ reduces to retrieving the $i$-th node of the AVL tree
corresponding to $w\in \coll$.
\end{proof}

We shall use pointers provided by $\lists(\coll)$ to implement occurrence stubs for $\findop$.
More precisely, a stub is going to be a pointer and a relative position with respect to that pointer.
By \cref{lem:lists}, an occurrence in $w\in \coll$ represented by such a stub can be identified in $O(\log |w|)$ time.

\subsection{Long Patterns}\label{sec:long}

We start with describing the main ideas of our solution. In the description we assume $|p|>1$ but we shall actually
use this approach for $|p|>\log^2 n$ only.
Our solution shares some ideas with an index for grammar-compressed strings by Claude and Navarro~\cite{DBLP:journals/fuin/ClaudeN11}
and a dynamic index in compressed space by Nishimoto et al.~\cite{DBLP:journals/corr/NishimotoIIBT15}.

\subsubsection{Combinatorial Background}
Let us consider an occurrence $w[i..j]$ of a pattern $p$ in a string $w\in \coll$.
Following the idea of~\cite{DBLP:journals/fuin/ClaudeN11}, we shall associate it with the lowest node $v$ of $\stree(w)$ which represents a fragment
containing $w[i..j]$ (called the \emph{hook} of the occurrence) and with its signature $s=\ussig(v)$.
Note that the considered occurrence of $p$ in $w$ yields an occurrence of $p$ in $\sstr(s)$
and if $\sstr(s)$ is partitioned according to the production corresponding to $s$, the occurrence
is not contained in any fragment of $\sstr(s)$ created this way.
This yields a non-trivial partition $p=p_1p_2\ldots p_m$ of the pattern.
More formally, if $s\to s_1s_2\ldots s_k$, then for some
$j\leq k-m+1$, $p_1$ is a suffix of $\sstr(s_j)$,
for each $d=1,\ldots,{m-2}$, $\sstr(s_{j+d})=p_{1+d}$
and $p_{m}$ is a prefix of $\sstr(s_{j+m-1})$.
Such an occurrence of $p$ in $\sstr(s)$ is called an \emph{anchored occurrence} of $p$ in $s$
and the anchored string $p_1|p_2\ldots p_m$ is called the \emph{main anchor} of the occurrence.
The main idea of queries to our data structure is to identify the hooks of each occurrence.
For this, we determine all symbols $s$ in $\grammar(\coll)$ which admit an anchored occurrence of $p$.

For each signature $s$ with $\slev(s)>0$ let us define an anchored string $\sanch(s)$ as follows:
if $s=(s_1,s_2)$, then $\sanch(s)=\sstr(s_1)|\sstr(s_2)$, and if $s=(s',k)$, then $\sanch(s)=\sstr(s')|\sstr(s')^{k-1}$.
\begin{lemma}\label{lem:relanch}
  Let $s\in \sigs(\grammar)$.
A pattern $p=p_\ell p_r$ has an anchored occurrence in $s$ with main anchor $p_\ell|p_r$
if and only if $p_\ell|p_r$ is a substring of $\sanch(s)$. If so, the number of anchored occurrences and their relative positions (within $\sstr(s)$)
can be determined in constant time.
\end{lemma}
\begin{proof}
Note that if $p_\ell|p_r$ is a substring of $\sanch(s)=w_\ell|w_r$, then the corresponding occurrence of $p$ in $\sstr(s)$ is clearly anchored
with main anchor $p_\ell|p_r$.
Thus, let us suppose that $p$ has an anchored occurrence in $s$ with main anchor $p_\ell|p_r$ and let $p_1p_2\ldots p_m$ be the partition of $p$ this occurrence yields ($p_1=p_\ell$, $p_2\ldots p_m=p_r$).
Let us consider two cases.
First, if $s=(s_1,s_2)$, then $m=2$ and $p_1|p_2 = p_\ell|p_r$ is clearly a substring of $\sstr(s_1)|\sstr(s_2)=w_\ell|w_r$.
The starting position is $1+|w_\ell|-|p_\ell|$.

Next, suppose that $s=(s',k)$ and let $w'=\sstr(s')$. Note that $m\le k$, $p_2=\ldots=p_{m-1}=w'$, $p_1$ is a suffix of $w'$, and $p_m$ is a prefix of $w'$. Consequently, $p_1|p_2\ldots p_m = p_\ell | p_r$ is a substring of $w'|w'^{m-1}$ and thus of $w'|w'^{k-1}=w_\ell|w_r$.
However, there are actually $k-m+1$, i.e., $1+\lfloor{\frac{|w_r|-|p_r|}{|w_\ell|}}\rfloor$,
occurrences which start at subsequent positions $1+i|w_\ell|-|p_\ell|$ for $1\le i \le k-m+1$.
\end{proof}

\paragraph{Potential Anchors}
\cref{lem:relanch} reduces pattern matching for $p$ to pattern matching of anchored strings $p_\ell|p_r$ (with $p=p_\ell p_r$)
in a set $\sanch(\ussig(v))$ for $v\in \stree(w)$ and $w\in \coll$. What is important, we do not have to consider all $|p|-1$ possible anchored strings that can be obtained by inserting an anchor into $p$.
As stated in the following Lemma, only a few such anchored patterns are necessary.
An analogous observation appeared in a recent paper by Nishimoto et al.~\cite{DBLP:journals/corr/NishimotoIIBT15}.

\begin{lemma}\label{lem:maina}
Consider the run-length encoding $s_1^{\alpha_1}\cdots s_k^{\alpha_k}$ of a context-insensitive decomposition $D$ of a pattern $p$.
If $p=p_\ell p_r$ such that $p_\ell|p_r$ is a main anchor for an anchored occurrence of $p$ in some $s\in\sigs(\grammar)$, then $p_\ell = s_1$ or $p_\ell=s_1^{\alpha_1}\cdots s_i^{\alpha_i}$ for some $i\in\{1,\ldots,k-1\}$.
\end{lemma}
\begin{proof}
Note that $\sstr(s)$ can be seen as an extension of $p$.
Let $v$ be the root of $\ustree(s)$ and let $p=p_1\ldots p_m$ be the partition of $p$ induced by children of $v$, and let $v_1,\ldots,v_m$ be the corresponding children of~$v$.

For $i\in\{1,\ldots,k\}$ and $j\in \{1,\ldots,\alpha_i\}$ let $u_{i,j}$ be the counterpart in $\ustree(s)$ of the node
in $\ustree(p)$ corresponding to the $j$-th copy of $s_i$ in the decomposition.
Observe that each node $u_{i,j}$ is a descendant of some $v_{i'}$, and thus the fragment represented by $u_{i,j}$
is either fully contained in $p_\ell$ or in $p_r$.
Therefore there exists a prefix of $D$ which forms a decomposition of $p_\ell$.

Now, consider nodes $u_{i,j}$ and $u_{i,j+1}$. Observe that $s_i$ is their common signature, and thus by \cref{fct:layer} they have a common parent in $\ustree(s)$.
Unless this parent is $v$ and $u_{i,j}=v_1$ and $u_{i,j+1}=v_2$, the fragment represented in total by $u_{i,j}$ and $u_{i,j+1}$ is either
fully contained in $p_\ell$ or in $p_r$.
On the other hand, if $u_{i,j}=v_1$, then clearly $i=j=1$.
Consequently, the only prefixes of $D$ which may form a decomposition of $p_\ell$ are $s_1$ and $s_1^{\alpha_1}\ldots s_{i}^{\alpha_i}$
for  $i\in\{1,\ldots,k-1\}$.
\end{proof}
\begin{corollary}\label{cor:maina}
Given a pattern $p$, we can in $O(\depth(p))$ time compute a set of \emph{potential main anchors} of $p$,
i.e., a set of $O(\depth(p))$ anchored strings $p_\ell|p_r$ such that $p=p_\ell p_r$, containing (in particular) all main anchors satisfying that condition.
\end{corollary}
\begin{proof}
\cref{lem:layerimpl} lets us compute in $O(\depth(p))$ time the run-length encoding of a context-insensitive decomposition of $p$.
For an encoding of size $k$, \cref{lem:maina} gives $k$ potential main anchors.
This leads to  $O(\depth(p))$ potential main anchors computable in $O(\depth(p))$ time.
\end{proof}

\subsubsection{Data Structure}
Having developed all the necessary combinatorial tools, we are ready to describe its components.

In our non-persistent dynamic collection we store more information than in the persistent data structure described in \cref{sec:persistent_ds}.
As a result, updating the data structure is a bit more involved and less efficient, but we are able to support more powerful queries.
We maintain a grammar $\grammar$, which is implemented the same as way as in the data structure of \cref{sec:persistent_ds}.

Moreover, we maintain the component $\lists(\coll)$ provided by \cref{lem:lists} as well as an anchored index (see \cref{thm:ai}) that contains anchored strings $\sanch(\ussig(v))$ for each node $v\in \stree(w)$.
More precisely, the anchored index contains an entry $\itanch(v)$ whose anchored string is $\sanch(\ussig(v))$
and whose key is a pair consisting of a signature $\ussig(v)$ and a pointer provided by $\lists(\coll)$ to the first character
of the fragment of $w$ represented by $v$. Note that such keys are distinct.
The anchored index requires that certain strings can be quickly compared lexicographically.
In order to do that, we maintain an order-maintenance component $C$ of~\cref{sec:om} for maintaining the lexicographic order
on some strings of $\grammar$. In particular, we shall require that for $w\in \coll$ not only $w$ but also its reverse
is represented by $\grammar$.

To sum up, our data structure maintains a dictionary of signatures $\grammar$,  the component $\lists(\coll)$,
the anchored index, and the order-maintenance component $C$.
\subsubsection{Queries}

Our query algorithm works in several steps. First, we make sure that $\grammar$ contains a signature representing the pattern $p$ and its reverse $p^R$ (using \cref{lem:makeop} in $O(p+\depth(p)) = O(p+\log n)$ time).
Next, we build a context-insensitive decomposition of $p$ in order to determine $O(\depth(p))$ potential main anchors $p_\ell|p_r$ (using \cref{lem:layerimpl,cor:maina}).
For each of the $O(\depth(p))$ potential anchors $p_\ell|p_r$, we make sure that $p_\ell$, $p_r$, $p_\ell^R$, and $p_r^R$
are represented by $\grammar$ (using \cref{lem:splitop}) and add them to $C$.
In total this requires $O(\log^2 n)$ time with high probability.

We process the potential anchors one by one.
Note that $C$ contains all the strings necessary to query the anchored index (see \cref{thm:ai}) with pattern $p_\ell|p_r$.
Hence, we make such queries each of which gives a hook $v$ of an occurrence of $p$.
More precisely, we get the signature $\ussig(v)$ and a pointer to the first character represented by $v$.
For each hook $v$ we determine the number of occurrences and generate relative position of at most $occ$ occurrences (both using \cref{lem:relanch}).
This way we generate stubs of the occurrences in $O(\log n \cdot \depth(p)+occ)=O(\log^2 n + occ)$ time.
Recall that an occurrence represented by a stub can be determined in logarithmic time (using \cref{lem:lists}).
Hence the following result:
\begin{lemma}\label{lem:findop}
Queries $\findop(p)$ can be implemented in $O(|p|+\log^2 n + occ)$ time where $occ$ is the total number of occurrences and $n$ is the total length of strings in the collection.
A stub of an occurrence in $w$ returned by $\findop(p)$ can be located in $O(\log |w|)$ additional time.
The $\findop$ operation may fail with inverse polynomial probability.
\end{lemma}

\subsubsection{Updates}
We now describe how to maintain our data structure.
For maintaining the dictionary $\grammar$ of signatures we may use exactly the same algorithm as in the persistent string collection.
Let us now move to maintaining the anchored index and the data structure $C$.

\begin{lemma}\label{lem:update_ai}
Given a pointer $P$ to a node $v\in \stree(w)$ and handles to $w$ and $w^R$ in $\grammar$,
we may compute $\itanch(v)$ and insert it to or remove it from the anchored index, making sure that $C$ contains all necessary strings.
This takes $O(\log n)$ time with high probability.
\end{lemma}
\begin{proof}
First, observe that $\itrepr(P)$, $s=\itsig(P)$ and its production
rule let us compute the indices of fragments of $w$ equal to
$s_\ell$ and $s_r$ such that $\sanch(s)=(s_\ell,s_r)$.
Consequently, these strings can be added to $\grammar$ via $\splitop$ operations (\cref{lem:splitop}) starting from $w$.
Similarly, their reverses can be extracted from $w^R$.
Finally, all the constructed strings need to be added using $\insertop$ of the order-preserving component $C$.
All this takes $O(\log n)$ time with high probability and lets us add $\sanch(s)$ to the anchored index.

Next, observe that the key can also be constructed in $O(\log n)$ time. This is because $\itrepr(P)$ lets us retrieve
the index of first character of the fragment represented by $v$ and the component $\lists(\coll)$, by \cref{lem:lists},
returns a pointer to a given position in logarithmic time.

Finally, note that the size of the anchored index is linear with respect to the total length of strings in $\coll$.
Hence, by \cref{thm:ai}, updating the anchored index requires $O(\log n)$ time.
\end{proof}

\begin{lemma}\label{lem:makeopn}
For a string $w\in \Sigma^+$, we can execute $\makeop(w)$ and $\dropop(w)$ in $O(|w| \log n)$ time
with high probability.
\end{lemma}

\begin{proof}
First, let us focus on $\makeop(w)$. We use \cref{lem:makeop} to make sure that $\grammar$ contains $w$ and $w^R$.
We also update $\lists(\coll)$ to enable pointers to characters of $w$.
Finally, we traverse $\stree(w)$ and apply~\cref{lem:update_ai} to insert $\itanch(v)$ for all nodes $v$ at positive levels.
All this takes $O(|w|\log n)$ time with high probability.

Implementation of $\dropop(w)$ is even simpler. First, we traverse $\stree(w)$  and apply~\cref{lem:update_ai} to remove $\itanch(v)$
for all nodes $v$ at positive levels. Next, we remove $w$ from $\lists(\coll)$.
\end{proof}

The key observation to efficient implementation of $\splitop$ and $\concop$
is that the entries in the anchored index do not need to be modified for nodes at or below context-insensitive layers.
Thus, we only need to process nodes above it; the following Lemma lets us list them efficiently.
\begin{lemma}\label{lem:list-above}
Let $\decomp$ be a decomposition of $\str$ corresponding to a layer $L$ of $\stree(\str)$.
Given a run-length encoding of $\decomp$ of length $d$ we may list pointers to all nodes of $\stree(\str)$ that are (proper) ancestors of $L$ in
$O(d)$ time.
\end{lemma}
\begin{proof}
Roughly speaking, we would like to run a DFS on $\stree(\str)$ that starts in the root node and is trimmed whenever we encounter a node belonging to $L$.
However, we do not have a list of nodes belonging to $L$, but only the run-length encoding of the corresponding decomposition.
Since nodes on any root-to-leaf path have distinct signatures, this is not a major obstacle.
We still run a DFS that visits the children of each node starting from the leftmost one.
We observe that we encounter the first node belonging to $L$ exactly when we first visit a node $v$ with $\ussig(v)$ being the first symbol in $\decomp$. At this moment we backtrack and continue the DFS until we visit a node with signature equal to the second symbol of $\decomp$.

This way, we visit all nodes of $L$ and their ancestors.
As a result, the algorithm is not efficient if the size of $L$ is much bigger than $d$.
Thus, we introduce the following optimization.

Let $s_1^{\alpha_1}\cdots s_d^{\alpha_d}$ be the run-length encoding of decomposition $\decomp$.
Assume that we detect a node belonging to $L$ with signature $s_i$, which is a child of a node $v$.
By \cref{fct:layer}, we may immediately skip the following $\alpha_i$ children of $v$.
This can only happen $d$ times during the execution of our algorithm.

As a result we visit a set of nodes of $\stree(\str)$ that forms a tree containing $O(d)$ leaves.
Since every vertex of $\stree(\str)$ has at least two children, in total we visit $O(d)$ nodes.
\end{proof}

\begin{lemma}\label{lem:concsplitop}
We can execute $\splitop(w,k)$ and $\concop(\str_1, \str_2)$ operations in $O(\log^2 n)$
time with high probability.
\end{lemma}
\begin{proof}
Let $w=w_1w_2$ and suppose we split $w$ into $w_1$ and $w_2$, or concatenate $w_1$ and $w_2$ into $w$.
First, we make sure that the resulting strings and their reversals are represented by $\grammar$
(using \cref{lem:splitop} or \cref{lem:concop}).

Next, we construct their context-insensitive decompositions $D_1$ and $D_2$ of $w_1$ and $w_2$, respectively (using~\cref{lem:layerimpl2}).
By \cref{fact:two-decompositions} their concatenation $D=D_1D_2$ forms a decomposition of $w$.
We use \cref{lem:list-above} to list nodes of $\stree(w)$ above $D$, nodes of $\stree(w_1)$ above $D_1$, and nodes of $\stree(w_2)$
above $D_2$.

Next, we apply \cref{lem:update_ai} to insert or remove $\itanch(v)$ for all the listed nodes.
First, we process removed nodes, then update the $\lists(\coll)$ components, and finally we process the inserted nodes.
Note that for all nodes $v$ below the context-insensitive decompositions the $\itanch(v)$ entries do not need to be updated.
\end{proof}

\subsubsection{Summary}
\begin{theorem}\label{thm:long}
A collection of non-persistent strings can be maintained under operations $\concop$ and $\splitop$
working in $O(\log^2 n)$ time, $\makeop(w)$ and $\dropop(w)$ requiring $O(|w|\log n)$ time
so that $\findop(p)$ takes $O(|p|+\log^2 n +occ)$ time, where $occ$ is the total
number of occurrences and $n$ is the total input size of updates so far.
Each occurrence returned by $\findop(p)$ requires additional $O(\log n)$ time to locate.
All operations may fail with probability $\frac{1}{n^c}$ where $c$ is an arbitrarily large constant.
\end{theorem}
\begin{proof}
We implement the operations according to \cref{lem:makeopn,lem:concsplitop,lem:findop}, which results
in the running times provided. Recall that the implementations of these operations fail as soon as
they encounter a string of length up to $n$ whose depth is not $O(\log n)$. By~\cref{lem:small_depth},
this happens only with probability $\frac{1}{n^c}$ if the constant at the $O(\log n)$ bound is sufficiently large.
\end{proof}
\subsection{Short Patterns}\label{sec:short}

Given a parameter $\ell$, we design a separate implementation for short patterns, where
both $\concop$ and $\splitop$ work in $O(2^\ell+\log n)$ time, and $\findop(p)$ generates
all $occ$ occurrences in $O(|p|+occ)$ time if $|p|\in [2^{\ell-1},2^{\ell})$.

Intuitively, the main idea is to maintain partitions of all strings in $\coll$ into blocks with lengths from $[2^\ell,2^{\ell+1})$. If $s=s_1 s_2 \ldots s_b$ is such
a partition of a string $s$ with $|s| \geq 2^{\ell}$, every pair of adjacent blocks $s_i s_{i+1}$ (or $s_1$ if there is just one block) is stored in a \emph{generalized suffix tree}.
A generalized suffix tree of a collection of strings
$t_1,t_2,\ldots,t_k$ is a compacted trie storing all strings of the form $t_i[j..|t_i|] \$_i$ for $j=1,2,\ldots,|t_i|$, where
every $ \$_i$ is a unique character. It is known that a generalized suffix tree can be maintained
efficiently under inserting and deleting strings from the collection; see~\cite{dictionary}.
In more detail, it is easy to see that inserting a~string $p$ requires splitting up to $|p|$ edges of the tree, where
splitting an edge means dividing it into two shorter edges by creating a new explicit node,
where we connect a new leaf. All these edges can be retrieved in $O(|p|)$ time if every explicit
node representing a string $as$ stores a \emph{suffix link} to the explicit node representing string $s$.
Deletions are processed in a symmetric manner.
In the original implementation~\cite{dictionary}, it was assumed that the size of the alphabet is constant, and
an additional effort was required to avoid making every $ \$_i$ a unique character. However,
the only reason for this assumption was the following primitive: given a node
$u$ and a character $c$ we want to retrieve the outgoing edge whose label starts with $c$. Such
a~primitive can be implemented in $O(1)$ time with dynamic dictionaries (see \cref{sec:preliminaries}),
hence we obtain the following.

\begin{lemma}[\cite{dictionary}]
\label{lem:GST}
A generalized suffix tree can be maintained in $O(|p|)$ time under inserting or deleting a~string
$p$. Insertion returns a unique identifier, such that every leaf corresponding to a suffix of $p$
is labeled with its identifier. The identifier can be later used to delete $p$.
\end{lemma}

\subsubsection{Queries}
Consider an occurrence of a given pattern $p$ in a string $s=s_1 s_2 \ldots s_b$. Then either $b=1$
and the pattern occurs in $s_1$ or the pattern occurs in a concatenation of two adjacent blocks
$s_i s_{i+1}$ such that $i+1=b$  or the occurrence starts within $s_i$. Consequently, we can
generate the occurrences of $p$ by first navigating the generalized suffix tree storing all pairs
of adjacent blocks (of all strings) and all singleton blocks. After having found the node $u$
corresponding to the whole $p$, we retrieve all leaves in its subtree. To this end, we simply traverse
the whole subtree. Every leaf in the subtree corresponds to an occurrence, but it might happen
that two leaves correspond to the same occurrence. This happens exactly when the occurrence
is fully within one block $s_i$. Then it can be generated either as a substring of $s_{i-1}s_i$ or
as a substring of $s_i s_{i+1}$, and we cannot guarantee which occurrence is found first.
However, we can detect which of these two possibilities holds after generating an occurrence.
Consider the former. Using the additional data stored at the leaf we obtain a pointer to $s_{i-1}$.
Then we can obtain a pointer to $s_i$ and calculate the position of the other occurrence in $s_i s_{i+1}$.
Similarly, in the latter case using the additional data we obtain a pointer to $s_i$ and
then a pointer to $s_{i-1}$ and the position of the other occurrence in $s_{i-1} s_i$.
In both cases, we insert the pointer and the position into a dictionary, which allows us to later ignore
the other occurrence.
We terminate after either traversing the whole subtree or having generated the desired number of occurrences.
Then, because the every inner node has at least two children and the depth of any node of the
tree is at most $2^{\ell+1}\leq 4|p|$, we visit at most $O(|p|+occ)$ nodes.
For each leaf we get an identifier of the string $s_i s_{i+1}$ and an offset of the occurrence
within the string. We convert the identifier to the pointer (provided by $\lists(\coll)$) to the first character
of $s_i$.


\subsubsection{Updates}
Let us first describe how we maintain the partition into blocks.
We store $\lists(\coll)$ and for each string $w\in \coll$ a doubly-linked list of blocks.
For each block we also store a list of pointers to characters within the block. We also store dictionaries so that given a pointer to the character we may identify its block efficiently,
and so that given an identifier provided by \cref{lem:GST} we can retrieve data about the corresponding block.

To implement $\makeop(w)$, we simply partition $w$ into blocks $w_1w_2\ldots w_b$ of desired length and we insert
$w_{i}w_{i+1}$ to the generalized suffix tree using \cref{lem:GST}. Similarly, to implement $\dropop(w)$, where $w = w_1 w_2 \ldots w_b$,
we remove every $w_i w_{i+1}$ (or $w_1$ if $b=1$) from the generalized suffix tree in $O(|w|)$ total time.

Consider $\concop(w,w')$. If $|w|,|w'| \geq 2^{\ell}$ we can simply concatenate
their existing partitions $w = w_1 w_2 \ldots w_b$ and $w' = w'_1 w'_2 \ldots w'_{b'}$ into one,
which requires inserting the new pair of adjacent blocks $w_b w'_1$ into the generalized suffix
tree. Hence the whole update takes $O(2^\ell + \log |ww'|) = O(2^\ell + \log n)$ time. However, it might happen
that one of the strings, say $w'$, is shorter. If $|w| + |w'| < 2^{\ell+1}$, all three strings $w$, $w'$, $ww'$
consist of single blocks, so we delete $w$ and $w'$ from the generalized suffix tree and
insert $ww'$ there.
Otherwise, we concatenate $w_b$ and $w'$ obtaining a block of length
$s\in [2^\ell,3\cdot 2^\ell)$. If the length of the new block is at least $2^{\ell+1}$, we split
it into two blocks of length $2^\ell$ and $s-2^\ell \in [2^\ell,2^{\ell+1})$. In either case,
we change a constant number of adjacent blocks, and therefore can update the generalized
suffix tree in $O(2^\ell)$ time. Then we perform a constant number of insertions or deletions
on the balanced search tree corresponding to $w$ to obtain a balanced search tree corresponding
to $ww'$, which takes $O(\log |w|)=O(\log n)$ time.

Now consider $\splitop(w,k)$. If $|w| < 2^{\ell+1}$, all three strings $w$, $w[..k]$, $w[k+1..]$
consist of single blocks, so we delete $w$ and insert the other two into the generalized suffix tree.
Otherwise, we determine the pointer
to $w[k]$ and use it to use locate the block $w_i$ in the partition $w = w_1 w_2 \ldots w_b$
such that $|w_1 w_2 \ldots w_i| \geq k$ but $|w_1 w_2 \ldots w_{i+1}| > k$. We split
$w$ into three parts $w_1 w_2 \ldots w_{i-1}$, $w_i$, and $w_{i+1} w_{i+2} \ldots w_b$.
This is done by removing $w_{i-1} w_i$ (if $i>1$) and $ w_{i} w_{i+1}$ (if $i<b$) from the generalized suffix tree
in $O(2^\ell)$ time. Then we explicitly construct strings $w_i[..k']$ and $w_i[k'+1..|w_i|]$,
where $k' = k-|w_1 w_2 \ldots w_{i-1}|$. Using $\makeop$ and $\concop$ we can
then obtain the desired two strings $w[..k]$ and $w[(k+1)..]$ in $O(2^\ell+\log n)$ time.

\begin{theorem}
\label{thm:shortpatterns}
A collection of non-persistent strings can be maintained under operations $\concop$ and $\splitop$
working in $O(2^\ell + \log n)$ time, $\makeop(w)$ and $\dropop(w)$ requiring
$O(|w|)$ time, so that $\findop(p)$ for $2^{\ell-1} \le |p|< 2^{\ell}$ takes $O(|p|+occ)$ time,
where $occ$ is the total number of occurrences and $n$ is the total length of the strings
in the collection. Each occurrence returned by $\findop(p)$ requires additional $O(\log n)$ time to locate.
\end{theorem}

\subsection{Summary}

\begin{theorem}
A collection of non-persistent strings can be maintained under operations $\concop$ and $\splitop$
working in $O(\log^2 n)$ time, $\makeop(w)$ and $\dropop(w)$ requiring $O(|w|\log n)$ time
so that $\findop(p)$ takes $O(|p|+occ)$ time, where $occ$ is the total number of occurrences and $n$ is the total input size of updates so far.
Each occurrence returned by $\findop(p)$ requires additional $O(\log n)$ time to locate.
All operations may fail with probability $\frac{1}{n^c}$ where $c$ is an arbitrarily large constant.
\end{theorem}

\begin{proof}
First of all, we maintain the data structure of \cref{thm:long} in order to support queries for long patterns.
Second, we maintain an instance $I_\ell$ of the data structure of \cref{thm:shortpatterns} for each $\ell$ satisfying
$1\le 2^{\ell-1} \le \log^2 n$.
To answer $\findop$ query we use the former data structure if $|p|\ge \log^2 n$ and
one of the latter otherwise. This way the query time is $O(|p|+occ)$ irrespective of $|p|$.
As far as $\makeop(w)$ and $\dropop(w)$ is concerned, \cref{thm:long} requires $O(|w|\log n)$ time
and each $I_\ell$ uses $O(|w|)$ time, which sums up to $O(|w|\log\log n)$ in total.
For $\splitop$ and $\concop$, \cref{thm:long} spends $O(\log^2 n)$ time, while each $I_\ell$ uses $O(2^\ell+\log n)$ time.
The latter sums up to $O(\log^2n + \log n \log \log n) = O(\log^2 n)$ in total.

Notice, however, that as $n$ grows over time we need to introduce instances $I_\ell$ for subsequent values of $\ell$.
More precisely, $I_{\ell+1}$ needs to be ready at latest when $n$ reaches $n_{\ell}:=2^{2^{\ell/2}}$.
Thus, at the moment when $n$ reaches $n_{\ell-1}=2^{2^{(\ell-1)/2}}$ we shall start building that instance.
For this we use a background process which replays the history of updates and performs them on $I_{\ell+1}$.

Note that replaying prior operations takes $O(n_{\ell-1}\cdot (2^\ell+\log n_{\ell-1}))=O(n_{\ell-1}^{\sqrt{2}})=o(n_{\ell}-n_{\ell-1})$ time,
while simulating on $I_{\ell+1}$ operations for $n$ between $n_{\ell-1}$ and $n_{\ell}$ takes time proportional to executing
them on $I_{\ell}$. Consequently, if the background process is faster by a large enough constant than the process
maintaining $I_{\ell}$, $I_{\ell+1}$ is going to be ready early enough.
\end{proof}

\section{Pattern Matching in a Persistent Dynamic Collection}\label{sec:ppm}
In this section we build another dynamic index, in which, as opposed to the index of \cref{sec:pm},
the arguments of $\splitop$ and $\concop$ may be preserved or destroyed, whichever is desired at a given moment.
Formally, we maintain a set $\coll$ of strings that is initially empty.
The index supports the following updates:
\begin{itemize}
  \item $\makeop(\str)$, for $\str\in \Sigma^+$, results in $\coll:=\coll\cup \{\str\}$.
      \item $\dropop(\str)$, for $\str\in \coll$, results in $\coll:=\coll\setminus \{\str\}$.
   \item $\concop(\str_1,\str_2)$, for $\str_1,\str_2\in \coll$, results in $\coll:=\coll\cup \{\str_1\str_2\}$.
   This operation might be followed by $\dropop(w_1)$ and $\dropop(w_2)$ at no extra cost.
  \item $\splitop(\str,k)$, for $\str\in \coll$ and $1\leq k<|\str|$, results in
    $\coll:=\coll\cup \{\str[..k], \str[(k + 1)..]\}$.
       This operation might be followed by $\dropop(w)$ at no extra cost.
\end{itemize}

As in \cref{sec:collection}, each string in $\coll$ has an integer \emph{handle} assigned by the operation which created it.
Moreover, if an update creates a string which has already been added to $\coll$, the original handle is returned.
Otherwise, a fresh handle is assigned to the string. We assume that handles are consecutive integers starting from $0$
to avoid non-deterministic output.

The $\findop(p)$ query for $p\in \Sigma^+$ checks whether $p$ occurs in some string $w\in \coll$.
The user might decide to report all the occurrences at an extra cost depending on the output size.
As opposed to the non-persistent index, every occurrence of $p$ is reported using the handle of the text $w\in \coll$ and the position of the occurrence of $p$ within $w$.

\subsection{Implementation of the Index}\label{sec:internal_index}
The index is built on the top of the data structures of \cref{sec:collection} maintaining a grammar $\grammar$.
It can be seen as a separate component, which sometimes extends $\grammar$ and employs the tools of \cref{sec:om}
to maintain the lexicographic order on some strings represented by $\grammar$.
This leads to the following interpretation, which is more general than the interface presented above.

We maintain a grammar $\grammar$ (as in \cref{sec:collection}, subject to extension using $\makeop$, $\concop$, and $\splitop$) and a subset $\coll$
of strings represented by $\grammar$ being texts for which we are to perform pattern matching.
For technical reasons, we only allow $w\in \coll$ if $\sstr^{-1}(w^R)\in \sigs(\grammar)$,
i.e., we require that the reverse of a string in $\coll$ must be represented by $\grammar$.
Our operations are as follows:

\begin{itemize}
  \item $\actop(s,s')$, for $s,s'\in \sigs(\grammar)$ such that $\sstr(s)=\sstr(s')^R$, results in $\coll := \coll\cup \{\sstr(s)\}$,
  \item $\dactop(s,s')$, for $s,s'\in \sigs(\grammar)$ such that $\sstr(s)=\sstr(s')^R$, results in $\coll := \coll\setminus \{\sstr(s)\}$,
  \item $\findop(p)$, for $p\in \Sigma^+$, checks if $p$ occurs in any string $w\in \coll$ and (optionally) reports all the occurrences.
\end{itemize}

The running times of the $\actop$ and $\dactop$ operations depend on the size of the set $\grammar(\coll)$.
Recall that $\grammar(\coll)$ can be seen as the minimal superset of $\coll$ which is closed with respect
to expanding any production rule; see~\cref{app:formalsymbols}.
This set can be characterized as follows:
\begin{observation}\label{obs:ppm}
  Let us fix a set $\coll\sub \Sigma^+$. A signature $s$ belongs to $\sigs(\grammar(\coll))$
if and only if there exists a string $w\in \coll$ and a node $v\in \stree(w)$ such that $\ussig(v)=s$.
\end{observation}

Like in \cref{sec:collection}, the success probabilities and running times also depend on the total size $t$ of prior updates. This includes all operations triggered by the user,
including those on the grammar ($\makeop$, $\concop$, and $\splitop$) and those on the index ($\actop$ and $\dactop$).
On the other hand, operations called internally by the components of the index are not counted.
\begin{itemize}
  \item $\actop$ runs in $O(1+|\grammar(\coll')\setminus \grammar(\coll)|\cdot (\log n + \log t + \log|\grammar(\coll')|))$ time with probability $1-\frac{1}{t^c}$, where $n=|\sstr(s)|$ and  $\coll'=\coll\cup\{\sstr(s)\}$.
  \item $\dactop$ runs in $O(1+|\grammar(\coll)\setminus \grammar(\coll')|\cdot (\log n + \log t + \log|\grammar(\coll)|))$ time with probability $1-\frac{1}{t^c}$, where $n=|\sstr(s)|$ and $\coll'=\coll \setminus \{\sstr(s)\}$.
  \item $\findop$ runs in $O(|p|+\depth(p)(\log |p|+\log t + \log |\grammar(\coll)|))$ time with probability $1-\frac{1}{t^c}$ plus  $O(\depth(w))$ for each reported occurrence of $p$ in $w\in\coll$.
\end{itemize}

\subsubsection{Data Structure}
Our data structure relies on the combinatorial tools developed in \cref{sec:pm} for the non-persistent index as
well as on the anchored index (\cref{thm:ai}).
However, we do not introduce a separate entry for each $v\in \stree(w)$ for $w\in \coll$,
but instead, we create a single entry for each $s \in \sigs(\grammar(\coll))$. This is possible due to \cref{obs:ppm}.
The entry $\itanch(s)$ consists of the anchored string $\sanch(s)$ and has $s$ as the key. 
Moreover, we store some extra data for each $s\in \sigs(\grammar(\coll))$: the information whether $\sstr(s)\in \coll$ and a linked
list $\spars(s)$ of \emph{parents} of $s$: signatures $s'\in \grammar(\coll)$ whose corresponding
symbol's production rule contains $s$
(formally, either $s'\to s^k$ for some integer $k$, $s'\to s_1s$ for some signature $s_1\in \sigs(\grammar(\coll))$, or $s'\to ss_2$ for some signature $s_2\in \sigs(\grammar(\coll))$).
Additionally, whenever $s'\in \spars(s)$, we store at $s'$ a pointer allowing to access the corresponding element of $\spars(s)$.
Note that any signature $s'$ needs at most two such pointers (one for each distinct signature in the right-hand side of the production rule).
 
\subsubsection{Queries}
By \cref{obs:ppm}, the set of keys is the same as in the non-persistent variant. 
Hence, the algorithm developed to prove \cref{lem:findop} can be used to determine whether $p$ occurs in any string $w\in \coll$.
The running time is $O(|p|+\depth(p))$ to construct the $O(\depth(p))$ potential anchors (see \cref{cor:maina}),
and for each potential anchor $O(\log |p|+\log t)$ (with high probability) to insert the necessary strings to the order-maintenance component
as well as $O(\log |\grammar(\coll)|)$ to query the anchored index.
In total, this gives $O(|p|+\depth(p)(\log |p|+\log t + \log |\grammar(\coll)))$ time with probability $1-\frac{1}{t^c}$.

In the reporting version, the anchored index gives (for each potential anchor $p_\ell|p_r$ with $p=p_\ell p_r$) the list of all signatures $s\in \sigs(\grammar(\coll))$ for which there exists an anchored occurrence of $p$ in $\sstr(s)$ with main anchor $p_\ell|p_r$.
\cref{lem:relanch} lets us identify all these anchored occurrences. Combining this data over all potential anchors,
we obtain (in extra linear time) the list of all anchored occurrences of $p$ in all $\sstr(s)$, where $s\in \sigs(\grammar(\coll))$.
In particular, for every occurrence of $p$ in $w\in \coll$ with hook $v$, the corresponding anchored occurrence of $p$ in $\ussig(v)$ 
is listed. 
Consequently, for every reported signature $s$ it suffices to list the fragments of $w\in \coll$ corresponding to nodes $v$ with $s=\ussig(v)$.
We shall do that in $O(\depth(w))$ time per fragment so that listing an occurrence in $w$ also works in $O(\depth(w))$ extra time, as announced.

More precisely, the following recursive procedure achieves $O(1+\depth(w)-\slev(s))$ time per fragment reported.
First, if $\sstr(s)\in \coll$, we simply report the whole string (in constant time).
Next, we recurse on all parents of $s$ listed in $\spars(s)$.
For each parent $s'$ and each reported fragment, we compute one (if $s'\to s_1s_2$) or $k$ (if $s'\to s^k$) fragments
represented by node with signature $s$.
Note that the inequality $\slev(s')\ge \slev(s)+1$ lets is inductively prove the claimed bound on the running time.
Finally, for each fragment and for each anchored occurrence of $p$ in $s$ we report an occurrence of $p$ in some $w\in \coll$.

\subsubsection{Updates}
The updates are slightly simpler than in the non-persistent version because we just need to maintain
the set $\sigs(\grammar(\coll))$ along with some auxiliary data.

Let us start with the implementation of $\actop$ inserting a string $w$ to $\coll$. 
First, we add $\sstr^{-1}(w)$ to $\sigs(\grammar(\coll))$, recursively adding some further signatures.
During this procedure, we traverse the parse tree $\stree(w)$ in order to add the
needed signatures to $\sigs(\grammar(\coll))$.
Note that for each $v\in\stree(w)$ such that $\ussig(v)$ has to be added, all signatures
corresponding to the ancestors of $v$ in $\stree(w)$ have to be added to $\sigs(\grammar(\coll))$ as well.
We start processing $v$ by adding $\ussig(\spar{w}(v))$ to $\spars(s)$
(unless $v$ is the root). If $s$ is already in $\sigs(\grammar(\coll))$, there is nothing more to do.
Otherwise, we start by recursively processing the children of $v$.
If $s\to s_1s_2$, then we process both children, but if $s\to s_1^k$, we process just the first one (since the further children
have the same signature, and already the first call makes sure that it is added to $\grammar(\coll)$). 
Next, we add $\itanch(s)$ to the anchored index of \cref{thm:ai}. For this, we insert some strings to the order-maintenance
component (we split out the appropriate fragments of $w$ and $w^R$),
and then actually insert the entry. This takes $O(\log n+\log t + \log|\grammar(\coll)|)$ time with high probability for a single symbol.
Clearly, in total, we only process signatures from $\sigs(\grammar(\coll')\sm\grammar(\coll))$, at most two children of every such symbol,
and $\sstr^{-1}(w)$ (on which we spend constant time even if $\sstr^{-1}(w)$ is already in $\grammar(\coll)$).
Finally, we mark that $\sstr^{-1}(w)$ represents a string in $\coll$.

The implementation of $\dactop$ follows the same steps in the reverse order:
First, we unmark $\sstr^{-1}(w)$, and next we recursively remove signatures from $\sigs(\grammar(\coll))$.
While processing $s$ we maintain a path in $\stree(w)$ from the root to a node $v$ such that  $\ussig(v)=s$ and the signatures of nodes above have just been removed from $\sigs(\grammar(\coll))$. We start by removing $s':=\ussig(\spar{w}(v))$ from $\spars(s)$. We use the pointer stored with $s'$ to do it in constant time.
If $\spars(s)$ is still non-empty or if $\sstr(s)\in \coll$, we do not need to remove $s$ from $\sigs(\grammar(\coll))$ and thus there is nothing to do.
Otherwise, we recursively process the children of $v$.
Again, if $s\to s_1s_2$, we process both children, but if $s\to s_1^k$, we process the first one only.
Next, we remove $\itanch(s)$ from the anchored index.
Again, it is easy to see that we spend $O(\log n+\log t + \log|\grammar(\coll)|)$ time per symbol in $\grammar(\coll)\sm\grammar(\coll')$,
which results in the claimed running time.

\subsection{Conclusions}
In this section we describe how the component developed in \cref{sec:internal_index} can be used
to implement the interface of the persistent dynamic index.
The following auxiliary fact, based on the context-insensitive decompositions (\cref{sec:contexti})
and the results of \cref{sec:adding} lets us bound the running times of the $\splitop$ and $\concop$
implementations:

\begin{fact}\label{fct:persistent_index}
  Let $w_1,w_2\in \Sigma^+$ and let $w=w_1w_2$. We have $|\grammar(w)\sm \grammar(\{w_1,w_2\})| = O(\depth(w))$, $|\grammar(w_1)\sm \grammar(w)| = O(\depth(w_1))$,
and $|\grammar(w_2)\sm \grammar(w)| = O(\depth(w_2))$.
\end{fact}
\begin{proof}
Let $D_1$ and $D_2$ be the context-insensitive decompositions of $w_1$ and $w_2$, respectively.
Note that $\grammar(D_i) \sub \grammar(w)\cap \grammar(w_i)$ for $i=1,2$ where the context-insensitive decomposition $D_i$ is identified with the set
of its signatures. \cref{lem:layerimpl2} implies $|\rle(D_i)|=O(\min(\depth(w),\depth(w_i)))$, while \cref{cor:decomposition-to-sig}
yields that ${|\grammar(w)\sm {\grammar(D_1\cdot D_2)}|={O(\depth(w)+|\rle(D_1\cdot D_2)|)}}$ and \linebreak $|\grammar(w_i)\sm \grammar(D_i)|\le O(\depth(w_i)+|\rle(D_i)|)$.
Combining these results, we obtain the claimed bounds.
\end{proof}

We are now ready to prove the main result of this section:

\begin{theorem}\label{thm:persistent_index}
There exists a data structure which maintains a persistent dynamic index $\coll$ and supports $\concop$ and $\splitop$ in $O((\log n + \log t)(\log N + \log t))$
time, $\makeop$ and $\dropop$ in $O(n(\log N + \log t))$ time,
and $\findop$ in $O(|p|+(\log |p|+\log t)(\log N+\log t) + occ(\log N+\log t))$ where $n$ is the total length of strings involved in the operation,
$N$ is the total length of strings in $\coll$, and $t$ is the total input size of the prior updates (linear for each $\makeop$, constant for each $\concop$, $\splitop$, and $\dropop$).
An update may fail with probability $O(t^{-c})$ where $c$ can be set as an arbitrarily large constant.
The data structure assumes that total length of strings in $\coll$ takes at most $\mword$ bits in the binary representation.
\end{theorem}
\begin{proof}
We use the data structure of \cref{sec:collection} to maintain a grammar $\grammar$ so that it contains symbols representing $w$ and $w^R$ for every string
$w$ inserted to $\coll$. By \cref{thm:data_structure}, this takes $O(n+\log t)$ for $\makeop$, and $O(\log n + \log t)$ for $\splitop$ and $\concop$.
Pattern matching is supported using the extension of \cref{sec:internal_index}. Thus, we must maintain $\coll$ by activating
and deactivating the appropriate strings.
The size of $\grammar(\coll)$ is clearly $O(N)$, and what remains to be seen is how much this sets changes due to each operation.
First of all, $|\grammar(s)|\le |\sstr(s)|$, so the running time of the activation during $\makeop$ and deactivation during $\dropop$
can be bounded by $O(n(\log n + \log t + \log N))$. In the implementation of the $\splitop$ and $\concop$, we first activate the newly created string,
and then (optionally) deactivate the input strings. By \cref{fct:persistent_index}, the change in $\grammar(\coll)$ is bounded by the depths of the added and
removed strings, which are $O(\log n + \log t)$ with high probability. Consequently, $\splitop$ and $\concop$ indeed work in $O((\log n + \log t)(\log N + \log t))$ time. The bound on the running time of $\findop$ follows from the facts that $|\grammar(\coll)|=O(N)$ and $\depth(p)\le \log |p|+\log t$
with high probability. Moreover, if $|p|>N$, we can immediately give a negative answer, so we also have $\log |p|\le \log N$.
\end{proof}

\begin{corollary}\label{cor:persistent_index}
Suppose that we use the data structure to build  a dynamic index with $n$ strings, each of length at most $\poly(n)$.
This can be achieved by a Las Vegas algorithm whose running time with high probability is $O(N\log n +M\log^2 n)$ where $N$
is the total length of strings given to $\makeop$ operations and $M$ is the total number of updates ($\makeop$, $\concop$, and $\splitop$).
Pattern matching in such an index takes $O(|p|+\log^2 n + occ\log n)$ time with high probability.
\end{corollary}

\section{An application: Finding Patterns in a History of Edits}\label{sec:timeline}
As an application of the dynamic string collection data structure,
we study the problem of maintaining the history of
edits of a text, so that efficient searching in the history
is possible.

More formally, we consider a text $T$ that is being edited.
At every step an edit of one of the following types is performed:
   (1) insert a single letter somewhere in the text,
   (2) remove a block of the text,
   (3) move a nonempty block of the text to some other position (cut-paste).
We assume that the whole history starts with an empty text $T_1$,
consists of $t$ versions $T_i$, and for each $i\in [2,t]$, $T_i$ is
a result of performing an edit of type (1), (2) or (3) on the text
$T_{i-1}$.

\begin{shortv}
Our data structure supports a pattern matching query.
Given a string $p$ it may determine if $p$ occurs in any version $T_i$.
Moreover, it can then list the occurrences in chronological order.
\end{shortv}

\begin{definition}
Let $T_1,T_2,\ldots,T_t$ be the history of edits of a text $T$.
We say that a pattern $p$ occurs in the history of $T$
if it is a substring of some $T_i$, $1\leq i\leq t$.
Additionally, we say that $p$ first occurs in $T_i$, if
it does not occur in any $T_j$ for $j<i$.
\end{definition}

We would like to be able to list the occurrences a given pattern within the history of edits.
In order to do that, we first need to determine what a \emph{distinct}
occurrence is.

Imagine each inserted character is a given a unique identifier.
The identifier assigned during the operation (1) never changes.
If decided that two occurrences of the pattern are distinct
whenever the sequences of their identifiers are equal,
we would guarantee that moving a block (operation (3)) does
not introduce additional occurrences within that block.
However, such definition would also result in the identification
of the occurrences of some pattern that are not contiguous
in time.
For instance, if $T_1=\eps$, $T_2=a$, $T_3=ac$, $T_4=abc$ and $T_5=ac$,
then the pattern $ac$ cannot be found in $T_4$, but still the
occurrences in $T_3$ and $T_5$ would be identified.

\begin{definition}
Let $T_1,\ldots,T_t$ be the history of edits of $T$.
For pattern $p$, let $o_1$ and $o_2$
be the occurrences
of $p$ within the history $T$.
The occurrence $o_i$ is located in version $T_{\tau(o_i)}$ at position $\delta(o_i)$.
Denote by $id(o_i)$ the sequence of letter identifiers of the fragment $T_{\tau(o_i)}[\delta(o_i)..(\delta(o_i)+|p|-1)]$
which is the occurrence of $o_i$.
We say that $o_1$ and $o_2$ such that $\tau(o_1)\leq \tau(o_2)$ are distinct if and only if
$id(o_1)\neq id(o_2)$ or there exists no such occurrence $o_3$ of $p$ for which both
$\tau(o_1)<\tau(o_3)<\tau(o_2)$ and $id(o_3)=id(o_1)$ hold.
\end{definition}


Below we also show how to list the occurrences
in the chronological order of their first appearances.

\subsection{Listing the Occurrences in Order -- the Static Case}\label{sec:timelinestatic}

Let again $T_1,\ldots,T_t$ be the history of edits of the text $T$.
Our goal is to support efficient listing of occurrences of some pattern
in the history, in chronological order.
We describe the occurrences $o$ by the pairs $(\tau(o),\delta(o))$, as defined above.
If the occurrences $o_1,o_2,\ldots$ are to be listed, we require that
$\tau(o_i)\leq \tau(o_{i+1})$ for every $i$.

Assume temporarily that we only need to support searching for patterns
of length at least $2$.

The strategy is to reduce the problem to the case of pattern matching in
a set of anchored strings.
The appropriate multiset $A_t$ of anchored strings can be defined inductively:
If $t=1$, then $A_1=\emptyset$.
Otherwise (if $t>1$), we assume that $A_{t-1}$ is given. We have the following cases:
  \begin{enumerate}
    \item $T_t$ is the text $T_{t-1}$ with one letter inserted. Then
      $T_t=XcY$ and $T_{t-1}=XY$ for some strings $X$, $Y$ and a letter $c$.
      If $T_t$ introduces any new occurrences of the pattern $p$, then $p$ is
      of the form $p_1cp_2$.
      Assume that $p_1\neq\eps$.
      Then $X\neq\eps$ as well.
      Moreover, $p_1$ is a suffix of $X$, whereas $cp_2$ is a prefix of $cY$.
      Equivalently, an anchored string $p_1|cp_2$ is a substring of an
      anchored string $X|cY$.
      As the choice of $p_1$ and $p_2$ was arbitrary, the anchored string $X|cY$
      may contain any anchored string $p_1|cp_2$ such that $p=p_1cp_2$.
      Otherwise, if $p_1=\eps$, then $Y\neq\eps$ and we include $c|Y$ in our set $A_t$.
      To sum up, $A_t=A_{t-1}\cup\{X|cY,c|Y\}$.

    \item $T_t$ is the text $T_{t-1}$ with some block of text deleted.
      Then $T_t=XY$ and $T_{t-1}=XBY$ for some string $X$, $Y$ and a nonempty string
      $B$.
      If $T_t$ introduces any new occurrences of the pattern $p$, then both $X$ and $Y$
      are nonempty and $p$ is of the form $p_1p_2$, where $p_1$ and $p_2$ are nonempty,
      $p_1$ is a suffix $X$, while $p_2$ is a prefix of $Y$.
      Hence, we set $A_t=A_{t-1}\cup\{X|Y\}$.

    \item $T_t$ is created by moving some block of text within $T_{t-1}$.
    	Then $T_{t-1}=XABY$ and $T_t = XBAY$ for some non-empty strings $A,B$
      and some strings $X,Y$.
      We add to $A_t$ anchored strings $XB|AY$, $X|BAY$ (if $X\ne \eps$), and $XBA|Y$ (if $Y\ne \eps$).
  \end{enumerate}

As any subsequent edit operation increases the size of the set of anchored
strings by at most $3$, we have $|A_t|=O(t)$.
The above construction can be simulated efficiently using the dynamic string
collection data structure from \cref{sec:collection,sec:om}:
we begin with an empty collection $\coll$ and
each step consists of at most one $\makeop$ operation and a constant
number of $\splitop$ and $\concop$ operations to assemble the needed
anchored strings.
Thus, the total time needed to construct $A_t$ is $O(t\log t)$.

In order to list the occurrences chronologically, we keep the timestamp
function $\tau:A_t\to \{1,\ldots,t\}$ such that for each
$a\in A_j\setminus A_{j-1}$ we have $\tau(a)=j$.
Additionally, for each $a=X|Y$ in $A_{j}\setminus A_{j-1}$ we remember
the offset $\delta(a)$ containing the position following the
anchor between $X$ and $Y$ in the version $T_j$.
Note that we always have $T_j[\delta(a)-|X|..\delta(a)+|Y|-1]=XY$.
The values $\delta(a)$ can be easily computed while inserting $a$ to $A_j$.

Once we build the set $A=A_t$, we follow the same general strategy of
\cref{sec:anchored} to map its elements to two-dimensional points.
However, as we need our points to be reported chronologically according
to the timestamp function $\tau$, the construction is more subtle.
Fix the history $T_1,\ldots,T_t$.

We first recall the data structure of Willard \cite{Willard:1985}, which stores
a static set of $n$ two-dimensional points $(x,y)$ and allows to report $k$
points in an orthogonal rectangle in time $O(\log{n}+k)$.
We construct a balanced binary tree with
the individual points stored in the leaves in the $x$-order.
Then, in each internal node $v$ we store the list $Y_v$ of points
represented in the leaf descendants of $v$, sorted by the $y$-coordinate.
Additionally, for each point $p$ in $Y_v$ we store the pointers to the
predecessors/successors (with respect to weak inequalities)
of $p$ (in the $y$-order) in both $Y_l$ and $Y_r$, where $l$ and $r$
is the left and the right child of $v$, respectively.
The list $Y_v$ along with the required pointers can be constructed
from $Y_l$ and $Y_r$ in time $O(|Y_l|+|Y_r|)$.
Thus, the total time needed to build the data structure is $O(n\log{n})$.

In the static case, it is convenient to store the lists $Y_v$ as arrays.
In order to report points in a rectangle $[x_1,x_2]\times [y_1,y_2]$,
we first find the successor (predecessor) of $y_1$ ($y_2$) in $Y_{\text{root}}$
with binary search.
Then, we traverse the tree until we end up in $O(\log{n})$ nodes
responsible for $x$-intervals completely contained in $[x_1,x_2]$.
Following the predecessor/successor pointers allows us to avoid repeating
binary search in the subsequently visited nodes.
Finally, we report $k$ points by listing $k$ points from
$O(\log{n})$ subarrays of some of the arrays $Y_v$.

Now, we extend this data structure so that the points are reported
in increasing order of their timestamps, given by a function $\tau$.
We additionally build an \emph{range minimum query} data structure
for each of the lists $Y_v$.
As $Y_v$ is stored as an array,
such structure can be constructed in linear time and allows
searching for an element of minimum timestamp in a subarray in
$O(1)$ time \cite{LCA}.
As discussed above, all we need to do is to report $k$ points
with smallest weights among $O(\log{n})$ subarrays.
In order to do that, we initialize a priority queue with
$O(\log{n})$ triples of the form $(Y_v,a,b)$.
We assume the minimum weight of subarray $Y_v[a..b]$ to be
a key of such triple -- such a key can be computed in $O(1)$ time.
Next, we do the following until $k$ points have been reported:
remove the triple $(Y_v,a,b)$ with a minimum key $\tau(Y_v[c])$ from the queue,
report $Y_v[c]$ and add the triples $(Y_v,a,c-1)$ and $(Y_v,c+1,b)$
to the queue.
It is clear, that each subsequent point can be reported with
$O(\log{n})$ time overhead, so reporting $k$ points takes
$O(k\log{n})$ time in total.

Let $p$ be a pattern that we want to search for in the history
of edits.
Since our data structure supports only searches for anchored strings,
we need to turn our pattern into an anchored string as well.
Unfortunately, in general we do not know how to partition the
pattern in order to find its first occurrence in the history.
As a result, we need to try every possible partition.
All the possible partitions $p=p_1p_2$ into nonempty substrings
$p_1,p_2$ can be temporarily inserted into $\coll$ by
issuing $|p|$ single-letter $\makeop$ operations and no more than $|p|-1$ $\splitop$
operations on the collection $\coll$.
Fortunately, the sets of occurrences reported for different
partitions do not intersect.
A reported anchored string $a$ corresponds to the occurrence
$T_{\tau(a)}[\delta(a)-|p_1|..\delta(a)+|p_2|-1]$.
Since we have $O(|p|)$ possible partitions, $O(|p|)$ $\anchfind$
queries need to be issued to the anchored strings data structure
to find the first occurrence of $p$ in the history of edits,
obtaining the query time $O(|p|\log{t})$.
Moreover, if we want to report $k$ occurrences chronologically,
we can again employ a priority queue to obtain time $O(|p|\log{t}+k\log{t})$.

We handle the case of patterns of length $1$ separately.
For each letter $c\in\Sigma$ we keep the identifiers of
inserted letters $c$ along with the time intervals when
the identifiers were present in the text.
Listing the occurrences of a letter $c$ in the history
of edits in chronological order is trivial.

\begin{theorem}
Let $T_1,T_2,\ldots,T_t$ be the history of edits of the text $T$.
There exists a data structure that can preprocess the history in $O(t \log t)$ time w.h.p. and then answer pattern matching queries against all versions of $T$. A query for a pattern $p$ is answered in $O(|p|\log t)$ time w.h.p. (with respect to $t$).
After that the occurrences can be listed in chronological order, each with a delay of $O(\log t)$.
\end{theorem}

\subsection{Online Indexing of the Edit History}
We now sketch how to drop the assumption
that the whole history of edits $T_1,\ldots,T_t$ is given beforehand.
Instead, we assume that the insert/remove/cut operations
come in an online fashion.
As discussed previously, all we have to do is to design a data structure
maintaining a multiset $A$ of two-dimensional points subject to the following
operations:
\begin{itemize}
\item add a point to $A$,
\item report $k$ points from $A\cap([x_1,x_2]\times [y_1,y_2])$
  in the chronological order of their additions.
\end{itemize}
We proceed similarly as in the static case.
The dynamization of Willard's data structure has already been
studied~\cite{Dietz:1991,Mehlhorn:1990,Mortensen:2003}.
However, to our best knowledge, no effort has been made
to study the case when chronological reporting is required.
In order to allow for efficient chronological reporting,
we employ the arguably simplest way to dynamize
the static structure, without much loss in asymptotic update/query times
and worst-case performance guarantees.

We use the following lemma to design a replacement for
the static pointers from
$Y_v$ to the lists $Y_{c_1},Y_{c_2},\ldots$, where $c_1,c_2,\ldots$
are the children of the node $v$ in the top-level search tree.

\begin{lemma}[\cite{Dietz:1991}]\label{lem:mark}
Let $L$ be an initially empty list whose elements can be either
\emph{marked} or \emph{unmarked}.
There exist a data structure that supports each of the following
operations on $L$ in $O(\log{\log{n}})$ time.
\begin{itemize}
\item $insert(x,y)$ -- insert an unmarked element $x$ before $y\in L$.
\item $mark(x)$ -- mark an element $x\in L$.
\item $succ(x)$ -- for $x\in L$, return the next marked element in $L$,
  if such element exists.
\end{itemize}
\end{lemma}

We also replace the top-level binary search tree from the static
structure with a weight-balanced B-tree of \cite{Arge:2003}.
In a weight balanced B-tree $\Gamma$ with branching parameter $a\geq 8$,
the elements are stored as leaves.
For an internal node $v\in\Gamma$ we define its weight $w(v)$ to be
the number of leaves among the descendants of $v$.
Also denote by $c^v_1,c^v_2,\ldots$ the children of a node $v$.

The following invariants are fulfilled by a weight-balanced B-tree $\Gamma$:
\begin{itemize}
\item All leaves of $\Gamma$ are at the same level.
\item
  Let \emph{height} of a node $v$ be the number of edges
on the path from $v$ to any leaf.
An internal node of height $h$ has weight less than $2a^h$.
\item Except for the root, an internal node of height $h$ has weight
  larger than $\frac{1}{2}a^h$.
\end{itemize}
\begin{lemma}[\cite{Arge:2003}]
For a weight-balanced B-tree $\Gamma$ with branching parameter $a\geq 8$,
the following hold:
\begin{itemize}
    \item All internal nodes of $\Gamma$ have at most $4a$ children.
    \item Except for the root, all internal nodes of $\Gamma$ have at least $a/4$ children.
    \item If the number of elements stored in $\Gamma$ is $n$, $\Gamma$
      has height $O(\log_a{n})$.
\end{itemize}
\end{lemma}
To insert an element $e$ into $\Gamma$, we append it as a leaf
at an appropriate position.
This may result in some node $v$ getting out of balance,
i.e., $w(v)>2a^h$, and in that case rebalancing is needed.
For reasons that will become clear later, we realize rebalancing
in the following manner.
Once $w(v)$ first becomes at least $\frac{3}{2}a^h$ for a node $v$ with
$k(v)$ children, we find the index $j(v)$ such that
$\max\left(\sum_{i=0}^{j(v)}w(c^v_i),\sum_{i=j(v)+1}^{k(v)} w(c^v_i)\right)$
is the smallest.
Note that this value is less than $\frac{3}{4}a^h+2a^{h-1}$
at that time.
The value $j(v)$ is then kept (or incremented if a child $c^v_i$
with $i\leq j(v)$ splits into nodes ${c^v_i}',{c^v_i}''$) for
future split of $v$.
The intuition here is that we decide what the split of $v$ will look
like before it actually happens.
This allows us to rebuild any secondary data structures stored in $v$,
so that once nodes $v',v''$ are created, their secondary structures
are ready to use.
The rebuilding work is distributed among the following updates,
and as a result, we obtain good worst-case update bounds.

At the very moment when $w(v)$ becomes $2a^h$, we split $v$ into nodes $v'$ and $v''$ such that
$v'$ has children $c^v_1,\ldots,c^v_{j(v)}$ and $v''$ has
children $c^v_{j(v)+1},\ldots,c^v_{k(v)}$.
We have $$\max(w(v'),w(v''))=\\
\max\left(\sum_{i=0}^{j(v)}w(c^v_i),\sum_{i=j(v)+1}^{k(v)} w(c^v_i)\right)<\\
\tfrac{3}{4}a^h+2a^{h-1}+\tfrac{1}{2}a^h<a^h\left(\tfrac{5}{4}+\tfrac{2}{a}\right)\leq \tfrac{3}{2}a^h,$$
$$\min(w(v'),w(v''))=w(v)-\max(w(v'),w(v''))>\tfrac{1}{2} a^h.$$
Thus, the weight invariants hold after the split.
The worst-case time needed to perform an insertion is $O(\log{n})$.

Suppose $\Gamma$ stores $n$ elements and each node $v\in \Gamma$ maintains a secondary
data structure $D_v$ on the elements stored in the subtree of $v$
such that:
\begin{itemize}
  \item An empty structure $D_v$ is initialized in $O(1)$ time.
  \item We can add an element to $D_v$ in worst-case time $O(U(n))$.
  \item Query on $D_v$ can be performed in worst-case time $O(Q(n))$.
  \item A list $D_v.L$ containing a history of insertions to $D_v$
    is a part of $D_v$.
\end{itemize}
\begin{lemma}
Additional worst-case time needed to adjust the secondary structures
$D_v$ while performing an insertion to a weight-balanced B-tree $\Gamma$
is $O(U(n)\log n)$.
\end{lemma}
\begin{proof}
When inserting an element $e$ to $\Gamma$, a new leaf is created.
The leaf always contains a single element, hence all we need to
is to insert $e$ into a newly created secondary data structure.

Now, assume that $v$ is an internal node of $\Gamma$ of height $l$.
We insert $e$ to $D_v$ in time $O(U(n))$ and append this insertion
to $D_v.L$ in $O(1)$ time.
If after the insertion $w(v)>\frac{3}{2}a^l$ for the first time,
we initialize two secondary structures $D_{v'}$ and $D_{v''}$.
During the next $\frac{1}{2}a^l$ insertions of elements
in the subtree of $v$ we replay the insertions to $D_v$,
distributing the elements between $D_{v'}$ and $D_{v''}$,
four insertions at once.
From the insertion algorithm it follows that we already
know how the split of $v$ into $v',v''$ will look like
and thus whether the insertion past should be
performed on $D_{v'}$ or $D_{v''}$.
The additional work per insertion in the subtree of $v$
is $O(U(n))$.

When $w(v)=2a^l$, the secondary structures $D_{v'},D_{v''}$
are ready for the nodes $v'$ and $v''$.
\end{proof}

In our application, the secondary data structure $D_v$
can be seen as an enhanced linear list of points.
It supports the following operations:
\begin{itemize}
\item $insert(x,y,w,x')$ -- given an element $y\in D_v$ and a child
  $w$ of $v$, insert $x$ before $y$.
  Mark $x$ as stored in $D_w$ as $x'$.
  If $y=\perp$, $x$ is inserted at the end of the list.
  If $w=\perp$, $x$ is not assigned to any child of $v$.
\item $next(x,w)$ ($prev(x,w)$) -- for $x\in D_v$, return the next (previous)
  element in the list assigned to the subtree $w$.
\item $first(x,y)$ -- return the element among $x,\ldots,y$
  that was the first to be added to $D_v$.
\end{itemize}

\begin{lemma}
Let $v\in\Gamma$.
Operations $insert$, $next$, and $prev$, on $D_v$ can be performed in $O(\log{\log{n}})$
worst-case time. Operation $first$ can be performed in $O(1)$ worst-case time.
\end{lemma}
\begin{proof}
Let $w_1,w_2,\ldots,w_k$ be the children of $v$.
Recall that $k\leq 4a=O(1)$.
We keep a single structure from \cref{lem:mark}, denoted by $L_{w_i}$
for each $w_i$.
When $insert(x,y,w,x')$ is performed,
a copy of $x$ is inserted into each $L_{w_i}$.
We additionally mark $x$ in $L_w$.
By \cref{lem:mark}, the structures $L_{w_i}$ alone allow us to perform
$insert$, $next$ and $prev$ within desired time bounds.

\newcommand{\tmlca}{\mathtt{lca}}

In order to answer $first(x,y)$ queries, we need to introduce
another ingredient of $D_v$.
Let $\tau(x)$ be the time, when $insert(x,*,*,*)$ was performed.
We maintain a \emph{cartesian tree} $\mathcal{CT}$ corresponding to the list
$e_1,\ldots,e_l$ represented by $D_v$.
Denote by $\mathcal{CT}(e)$ the node of $\mathcal{CT}$ corresponding
to the element $e$.
The cartesian tree can be defined inductively.
Let $e_j$ be the unique element $e_j$ such that $\tau(e_j)$ is minimal.
The root $\mathcal{CT}(e_j)$ contains $e_j$.
If $j>1$ ($j<k$), then the left (right) child of $\mathcal{CT}(e_j)$ is the root of
the cartesian tree corresponding
to the sublist $e_1,\ldots,e_{j-1}$ ($e_{j+1},\ldots,e_k$).
\begin{claim}[folklore, \cite{LCA}]\label{cla:cartesian}
The minimum element of the sublist $e_a,\ldots,e_b$ is $e_c$ such
that $\mathcal{CT}(e_c)$
is the lowest common ancestor ($\tmlca$) of the nodes of $\mathcal{CT}(e_a)$
and $\mathcal{CT}(e_b)$.
\end{claim}

In our case, each newly added element $e$ has greater timestamp $\tau(e)$
than any of the previously inserted elements.
Clearly, the greatest element is always a leaf of the cartesian tree.
If $e$ is the only element of $D_v$, we make it a root of $\mathcal{CT}$.
Otherwise, suppose $e'$ ($e''$) is the left (right) neighbor of $e$ in $D_v$.
If $e'$ ($e''$) does not exist, we set $\tau(e')=-\infty$ ($\tau(e'')=\infty$).
Assume wlog that $\tau(e')<\tau(e'')$.

By \cref{cla:cartesian}, we have $\tmlca(\mathcal{CT}(e'),\mathcal{CT}(e''))=\mathcal{CT}(e')$,
$\tmlca(\mathcal{CT}(e),\mathcal{CT}(e'))=\mathcal{CT}(e')$ and \linebreak $\tmlca(\mathcal{CT}(e),\mathcal{CT}(e''))=\mathcal{CT}(e'')$.
Thus, $\mathcal{CT}(e)$ is a descendant of $\mathcal{CT}(e'')$ and $\mathcal{CT}(e'')$ is a descendant of $\mathcal{CT}(e')$.
Assume there exists such node $\mathcal{CT}(e^*)$ that is a descendant
of $\mathcal{CT}(e'')$ and an ancestor of $\mathcal{CT}(e)$.
Thus $\tau(e)>\tau(e^*)>\tau(e'')>\tau(e')$.
If $e^*$ was to the left (right) of $e'$ ($e''$), both
$\mathcal{CT}(e^*)$ and $\mathcal{CT}(e)$
would be descendants of $\mathcal{CT}(e')$ ($\mathcal{CT}(e'')$).
However, as they lie on the opposite sides of $e'$ ($e''$), one
of them would end up in the left subtree of $\mathcal{CT}(e')$ ($\mathcal{CT}(e'')$)
and the other in the right subtree.
It contradicts the fact that $\mathcal{CT}(e)$ is a descendant of $\mathcal{CT}(e^*)$.
We conclude that $\mathcal{CT}(e)$ is in fact the left son of $\mathcal{CT}(e'')$.

We now use the data structure of Cole et al.~\cite{Cole:2005} to maintain the cartesian tree
$\mathcal{CT}$.
The data structure allows to maintain
a dynamic tree subject to leaf insertions and $\tmlca$ queries.
Both insertions and $\tmlca$ queries can be handled in worst-case $O(1)$ time.
Upon insertion, all we have to do is to make $e$ a child
of the node $\mathcal{CT}(e'')$.
\end{proof}

We now describe how to use the weight-balanced B-tree $\Gamma$ along with the
secondary data structures $D_v$ to maintain a
growing set of points $A$, so that efficient orthogonal chronological
reporting is possible.
Recall that leaves of $\Gamma$ (from left to right) contain the points of $A$ ordered
by the first coordinate, whereas the secondary structures order the points
by the second coordinate.
Let $r$ be the root of $\Gamma$.
For convenience we additionally keep the points of $A$
along with a mapping to the entries of $D_r$
in a separate balanced tree $B_y$.
$B_y$ orders points by the second coordinate, so that
we can efficiently find the (strict or non-strict) predecessor/successor
of some point in $A$ in $y$-order.

In order to insert the point $(x,y)$, we
first use $B_y$ to find the point $(x',y')$ with the smallest $y$ coordinate
(corresponding to $s_r$ in $D_r$) greater than $y$ in worst-case time $O(\log{n})$.
Let $r=v_1,\ldots,v_b$ be the ancestors of the leaf representing $(x',y')$.
While we visit $v_1,v_2,\ldots$ in that order we insert $(x,y)$ into $D_{v_i}$
and query $D_{v_i}$ to find an appropriate position
in $D_{v_{i+1}}$, where $(x,y)$ should be inserted.
Then a new leaf containing $(x,y)$ is created as a child of $v_b$
and any needed rebalancing is done in a bottom-up fashion.
The worst case running time of the insertion is $O(\log{n}\log{\log{n}})$.

To report $k$ points from a rectangle $[x_1,x_2]\times [y_1,y_2]$ in
the insertion order, we follow the same strategy as in the static
case: we first find $O(\log{n})$ nodes of $\Gamma$ responsible
for the interval $[x_1,x_2]$.
For each of these nodes $v$ we find the two elements $a_v,b_v$ of
the list $D_v$ such that $a_v,\ldots,b_v$ are exactly the points in the
subtree of $v$ with the second coordinates in the interval $[y_1,y_2]$.
This can be done using the $prev$ and $next$ operations of secondary structures
during the top-down traversal.
Next, a priority queue storing the chronologically first points
of those subtrees is formed and each time we pop a point from subtree $v$ off the queue,
at most two additional points from $v$ are added.
The total time to report $k$ points is thus $O(\log{n}\log{\log{n}}+k\log{n})$.

As a simple corollary of the above discussion and \cref{sec:timelinestatic},
we obtain the following theorem.

\begin{theorem}
There exists a data structure maintaining the history of edits of text $T$,
which can be updated in $O(\log{t}\log{\log{t}})$ worst-case time, when a new edit is performed.
The data structure supports pattern matching queries against all version of $T$ in $O(m \log{t}\log\log{t})$ with high probability, where $m$ is the length of the pattern.
After that the occurrences of the pattern can be listed in chronological order, each with a worst-case delay of $O(\log t)$ with high probability.
\end{theorem}

\begin{remark}
The above construction of the incremental data structure for chronological
orthogonal reporting can be further tweaked so that the insertion and reporting take
$O(\log{n})$ and $O(\log{n})$ time in the worst-case correspondingly.
This can be done by using the methods of Mortensen \cite{Mortensen:2003}
along with the maintenance of a cartesian tree.
Rather easy but technical changes are required at every step of
Mortensen's construction.
\end{remark}

\appendix

\section{Model of Computation}\label{app:model}
We use the standard word RAM model (with multiplication); see e.g.~\cite{DBLP:conf/stacs/Hagerup98} for a reference.
We assume that the machine word is $\mword$ bits long, which implies that each memory cell can store an integer from the range $[0, 2^{\mword})$.

In this version of the paper we do not analyze the space usage of our data structures.
We only note that an update taking $t$ units of time increases the space by at most $O(t)$.
Queries could potentially also increase memory consumption, but (since none of our bounds are amortized)
we may keep track of memory cells modified by a query algorithm and afterwards revert them to the original values.

We use randomization, which is incorporated in the word RAM model in a standard
way: given an integer $x$ fitting a machine word in constant time we may obtain a uniformly random integer between 0 and $x$ inclusive.
We never gamble with the correctness of the results, and we guarantee the running times bounds to hold worst-case, but we do allow each operation to fail (notifying about its failure).
Once an instance of the data structure fails, it cannot be used any more, i.e., all subsequent operations also fail.
A standard transformation lets us guarantee no failures at the cost of gambling with the running time: we store a history of the updates
and once we encounter a failure, we initialize a new instance and replay the history (restarting every time we face a further failure).

We provide a guarantee that all operations succeed \emph{with high probability} (w.h.p.).
This is the strongest model in which dynamic dictionaries with worst-case (not amortized) constant-time updates have been developed.
Formally, we say that an event holds with high probability if its probability is at least $1-\frac{1}{2n^c}$ where $c>0$ is a parameter that can be chosen by the user and $n$ is the total input size of the prior and current updates (we assume input size of each update to be a positive integer).
Provided that all operations are supported in $n^{O(1)}$ time (which is the case in our work),
the aforementioned transformation results in operations working without failures where the same (up to a constant factor)
running time bounds hold with high probability (again, with respect to the input size of prior updates).
Moreover, the running time of a sequence of operations is w.h.p. (with respect to the input size of all updates)
proportional to the sum of the worst-case bounds in the model with failures, despite the fact that the algorithm needs to confront a failure with constant probability. %
Both these statements are due to the following result.
\begin{lemma}\label{lem:restart}
Let us fix a sequence of updates $u_1,\ldots,u_m$ to an initially empty randomized data structure $D$.
Let $n_i$ be the the total input size of $u_1,\ldots,u_i$. Suppose that for some $c>2$, each update $u_i$ works in time $t_i$, $1\le t_i \le n_i^{\sqrt{c}}$, with probability
at least $1-\frac{1}{2n_i^c}$ and fails otherwise.
Consider an algorithm which performs these updates and restarts (with an independent instance of $D$) whenever it faces a failure.
Its running time is $O(\sum_{i=1}^m t_i)$ with probability at least $1-\frac{1}{n_m^{\omega(\sqrt{c})}}$.
\end{lemma}
\begin{proof}
Let $T_i=\sum_{j=1}^i t_j$ and note that $n_i \le T_i\le n_i^{1+\sqrt{c}}$.
Also let $N=n_m$ and let $l$ be the smallest index such that $T_l \ge \sqrt{N}$.
Let us modify the algorithm so that at steps $i\ge l$  it does not restart but fails permanently instead.
This may only decrease the probability of timely success.

Observe that with constant probability (at least $a := 1-\sum_{k=1}^\infty \frac{1}{2k^2}=\frac{12-\pi^2}{12}$)
the sequence of operations is performed without any need to restart.
Hence, the probability of the event that $k$ restarts are needed before the success is bounded by $a^k$. For $k=\omega(\log N)$
this is $\frac{1}{N^{\omega(1)}}$.

On the other hand the time, between two restarts is bounded by $T_{l-1}$ and thus if there are at most $O(\log N)$ of them,
in time $T+O(T_{l-1}\log N)=T+O(\sqrt{N}\log N)=O(T)$ we either have a success or a permanent failure.
However, the probability of permanent failure ever happening is bounded $\sum_{i=l}^m \frac{1}{2n_i^c} \le \frac{m}{2n_l^c}$.
Since $N\ge m$ and $n_l^{1+\sqrt{c}} \ge T_l \ge \sqrt{N}$ this is at most
$$\frac{N}{2N^{\frac{c}{2(1+\sqrt{c})}}} =\frac{1}{N^{\omega(\sqrt{c})}}$$
as claimed.
\end{proof}

We use randomization to implement efficient dynamic dictionaries.\footnote{%
As opposed to Alstrup et al.~\cite{Alstrup}, we also directly use randomization in the core ideas of our technique.
However, due to the usage of dictionaries in both solutions, our models of computation are exactly the same.
}
Such dictionaries map keys to values and, both of which fit $O(1)$ machine words,
and provide membership queries (determine if there is an entry for a given key) and retrieval queries
(return the entry for a given key).
To implement such dictionaries, we use the result of Dietzfelbinger and Meyer auf der Heide~\cite{Dietzfelbinger:1990},
who supports all operations in $O(1)$ time w.h.p.
However, because we need to maintain multiple (possibly small) dictionaries,
we need to make sure that the w.h.p. bounds for the dictionaries coincide with the w.h.p.
bounds for the data structures we construct.
For this, we simulate all dictionaries throughout our data structure using a single dictionary, whose
keys additionally contain an identifier of the simulated instance they come from.
Moreover, for every update of our data structure, we insert to this dictionary a number of dummy entries
proportional to the input size of this update.

\section{Formal Definitions of $\symbols$, $\sstr$ and $\sstr^{-1}$}\label{app:formalsymbols}
For every (possibly infinite) alphabet $\Sigma$ we define an infinite set of symbols $\symbols$.
This set is defined inductively so that $\Sigma \subseteq \symbols$, $(S_1,S_2)\in \symbols$ for every $S_1,S_2\in \symbols$
and $(S,k)\in \symbols$ for every $S\in \symbols$ and $k\in \mathbb{Z}_{\ge 2}$.
Note that each symbol $S\in \symbols$ generates exactly one string $s\in \Sigma^+$, which we denote by $\sstr(S)$.
We also define $\slength(S)=|\sstr(S)|$.

For every instance of the data structure we fix random bits $\hs_i(S)$ for every $S\in \symbols$ and $i\in \mathbb{Z}_{\ge 1}$.
We assume that these random bits are drawn before an instance of the data structure is initialized, although they are hidden until it decides to reveal them.
This conceptual trick lets us unambiguously construct a straight-line grammar $\grammar(w)$ for every string $w\in \Sigma^+$
using the $\shrink{i}$ functions defined in \cref{sec:single_string}.
A grammar $\grammar(\coll)$ defined for a collection of nonempty strings is simply the union of such straight-line grammars
over $w\in \coll$.

Note that not necessarily all symbols of $\symbols$ might appear in straight-line grammars $\grammar(w)$ for $w\in \Sigma^+$
In particular, it easy to see that a single string might be generated by multiple symbols but, as we shall see soon,
exactly one of this symbols is present in $\grammar(\Sigma^+)$.

Before, let us define recursively the level $\slev(S)$ of a symbol $S$.
If $S\in \Sigma$ is a terminal, then $\slev(S)=0$. If $S=(S_1,k)$ for $k\in \mathbb{Z}_{\ge 2}$, then $\slev(S)=\slev(S_1)+1$.
Otherwise, $S=(S_1,S_2)$ and $\slev(S)$ is the smallest even number $l$ greater than $\max(\slev(S_1), \slev(S_2))$ such that $\hs_{l/2}(S_1) = 0$ and $\hs_{l/2}(S_2) = 1$. The following easy observation provides a characterization of this function.
\begin{observation}
If $S$ is a symbol of $\grammar(\Sigma^+)$, then $S$ may only be generated by $\shrink{\slev(S)}$.
\end{observation}

We are now ready to show the aforementioned bijection between $\Sigma^+$ and symbols of $\grammar(\Sigma^+)$.
We note that such a property did not hold for previous approaches~\cite{Alstrup,Mehlhorn}. Although its presence does not seem crucial for our solution, it simplifies our concepts.

\begin{restatable}{lemma}{startsymbol}\label{lem:startsymbol}
If $S$ is a symbol of $\grammar(\Sigma^+)$, then $S$ is the start symbol of $\grammar(\sstr(S))$.
Consequently, $\sstr(\cdot)$ bijectively maps symbols of $\grammar(\Sigma^+)$ to $\Sigma^+$.
\end{restatable}
\begin{proof}
Let $l=\slev(S)$.
The fact that $S\in \grammar(w)$ for some $w$ implies that $w$ has a fragment equal to $\sstr(S)$ which
in $\cshrink{l}(w)$ is represented by $S$.
Consequently, for each level $0\le i \le l$, this fragment is represented by a fragment of $\cshrink{i}(w)$.
We shall inductively prove that these fragments are equal to $\cshrink{i}(\sstr(S))$.
For $i=0$ this is clear, so let us take $0\le i < l$ such that this property holds at level $i$
to conclude that the same property holds at level $i+1$.
It suffices to prove that $\shrink{i+1}$ produces the same blocks both in the fragment of $\cshrink{i}(w)$ and in $\cshrink{i}(\sstr(S))$.
Since the corresponding fragment of $w$ is represented by a fragment of $\cshrink{i+1}(w)$, then $\shrink{i+1}$ places block boundaries
before and after the fragment of $\cshrink{i}(w)$. Moreover, because $\shrink{i+1}$ places a block boundary between two positions
solely based on the characters of these positions, the blocks within the fragment are placed in the analogous
positions as in $\shrink{i}(\sstr(S))$.
Consequently, the blocks are the same and this concludes the inductive proof.
In particular we get $\shrink{l}(\sstr(S))=S$ and thus $S$ in indeed the starting symbol of  $\grammar(\sstr(S))$.

For the second part, note that $\sstr(S_1)=w=\sstr(S_2)$ implies that $S_1=S_2$ is the start symbol of $\grammar(w)$.
\end{proof}

The function inverse to $\sstr$ mapping $\Sigma^+$ to symbols of $\grammar(\Sigma^+)$ is denoted by $\sstr^{-1}$.
We now restate the properties already outlined in~\cref{sec:persistent_ds}.
We say that $\grammar$ \emph{represents} a string $w$ if $\sstr^{-1}(w) \in \grammar$.
A grammar $\grammar$ is called \emph{closed} if all symbols appearing on the right-hand side of a production rule
also appear on the left-hand side.
Our data structure maintains a grammar $\grammar$ which satisfies the following invariant:
\begin{invariant}
$\grammar \subseteq \grammar(\Sigma^+)$ is closed and represents each $w\in \coll$.
\end{invariant}

\bibliographystyle{plainurl}
\bibliography{paper}

\begin{thebibliography}{10}

\bibitem{Alstrup}
Stephen Alstrup, Gerth~St{\o}lting Brodal, and Theis Rauhe.
\newblock Pattern matching in dynamic texts.
\newblock In David~B. Shmoys, editor, {\em 11th Annual {ACM-SIAM} Symposium on
  Discrete Algorithms, {SODA} 2000}, pages 819--828. {ACM/SIAM}, 2000.
\newblock URL: \url{http://dl.acm.org/citation.cfm?id=338219}.

\bibitem{Alstrup:2000}
Stephen Alstrup and Jacob Holm.
\newblock Improved algorithms for finding level ancestors in dynamic trees.
\newblock In {\em Automata, Languages and Programming, {ICALP} 2000}, pages
  73--84, 2000.
\newblock \href {http://dx.doi.org/10.1007/3-540-45022-X_8}
  {\path{doi:10.1007/3-540-45022-X_8}}.

\bibitem{dictionary}
Amihood Amir, Martin Farach, Zvi Galil, Raffaele Giancarlo, and Kunsoo Park.
\newblock Dynamic dictionary matching.
\newblock {\em Journal of Computer and System Sciences}, 49(2):208--222, 1994.
\newblock \href {http://dx.doi.org/10.1016/S0022-0000(05)80047-9}
  {\path{doi:10.1016/S0022-0000(05)80047-9}}.

\bibitem{Arge:2003}
Lars Arge and Jeffrey~Scott Vitter.
\newblock Optimal external memory interval management.
\newblock {\em SIAM Journal on Computing}, 32(6):1488--1508, June 2003.
\newblock \href {http://dx.doi.org/10.1137/S009753970240481X}
  {\path{doi:10.1137/S009753970240481X}}.

\bibitem{Bender:2002}
Michael~A. Bender, Richard Cole, Erik~D. Demaine, Martin Farach{-}Colton, and
  Jack Zito.
\newblock Two simplified algorithms for maintaining order in a list.
\newblock In Rolf~H. M{\"{o}}hring and Rajeev Raman, editors, {\em Algorithms,
  {ESA} 2002}, volume 2461 of {\em {LNCS}}, pages 152--164. Springer, 2002.
\newblock \href {http://dx.doi.org/10.1007/3-540-45749-6_17}
  {\path{doi:10.1007/3-540-45749-6_17}}.

\bibitem{LCA}
Michael~A. Bender and Martin Farach{-}Colton.
\newblock The {LCA} problem revisited.
\newblock In Gaston~H. Gonnet, Daniel Panario, and Alfredo Viola, editors, {\em
  Latin American Symposium on Theoretical Informatics, {LATIN 2000}}, volume
  1776 of {\em LNCS}, pages 88--94. Springer Berlin Heidelberg, 2000.
\newblock \href {http://dx.doi.org/10.1007/10719839_9}
  {\path{doi:10.1007/10719839_9}}.

\bibitem{BrodalLMTT03}
Gerth~St{\o}lting Brodal, George Lagogiannis, Christos Makris, Athanasios~K.
  Tsakalidis, and Kostas Tsichlas.
\newblock Optimal finger search trees in the pointer machine.
\newblock {\em Journal of Computer and System Sciences}, 67(2):381--418, 2003.
\newblock \href {http://dx.doi.org/10.1016/S0022-0000(03)00013-8}
  {\path{doi:10.1016/S0022-0000(03)00013-8}}.

\bibitem{Charikar}
Moses Charikar, Eric Lehman, Ding Liu, Rina Panigrahy, Manoj Prabhakaran, Amit
  Sahai, and Abhi Shelat.
\newblock The smallest grammar problem.
\newblock {\em {IEEE} Transactions on Information Theory}, 51(7):2554--2576,
  2005.
\newblock \href {http://dx.doi.org/10.1109/TIT.2005.850116}
  {\path{doi:10.1109/TIT.2005.850116}}.

\bibitem{treei}
Siu{-}Wing Cheng and Moon{-}Pun Ng.
\newblock Isomorphism testing and display of symmetries in dynamic trees.
\newblock In {\'{E}}va Tardos, editor, {\em 7th Annual {ACM-SIAM} Symposium on
  Discrete Algorithms, {SODA} 1996}, pages 202--211. {ACM/SIAM}, 1996.
\newblock URL: \url{http://dl.acm.org/citation.cfm?id=313852.313924}.

\bibitem{DBLP:journals/fuin/ClaudeN11}
Francisco Claude and Gonzalo Navarro.
\newblock Self-indexed grammar-based compression.
\newblock {\em Fundamenta Informaticae}, 111(3):313--337, 2011.
\newblock \href {http://dx.doi.org/10.3233/FI-2011-565}
  {\path{doi:10.3233/FI-2011-565}}.

\bibitem{Cole:2005}
Richard Cole and Ramesh Hariharan.
\newblock Dynamic {LCA} queries on trees.
\newblock {\em SIAM Journal on Computing}, 34(4):894--923, April 2005.
\newblock \href {http://dx.doi.org/10.1137/S0097539700370539}
  {\path{doi:10.1137/S0097539700370539}}.

\bibitem{DeterministicTossing}
Richard Cole and Uzi Vishkin.
\newblock Deterministic coin tossing with applications to optimal parallel list
  ranking.
\newblock {\em Information and Control}, 70(1):32--53, July 1986.
\newblock \href {http://dx.doi.org/10.1016/S0019-9958(86)80023-7}
  {\path{doi:10.1016/S0019-9958(86)80023-7}}.

\bibitem{Dietz:1991}
Paul~F. Dietz and Rajeev Raman.
\newblock Persistence, amortization and randomization.
\newblock In {\em 2nd Annual {ACM-SIAM} Symposium on Discrete Algorithms,
  {SODA} 1994}, pages 78--88, Philadelphia, PA, USA, 1991. Society for
  Industrial and Applied Mathematics.

\bibitem{Dietz:1987}
Paul~F. Dietz and Daniel~Dominic Sleator.
\newblock Two algorithms for maintaining order in a list.
\newblock In Alfred~V. Aho, editor, {\em 19th Annual {ACM} Symposium on Theory
  of Computing, {STOC} 1987}, pages 365--372. {ACM}, 1987.
\newblock \href {http://dx.doi.org/10.1145/28395.28434}
  {\path{doi:10.1145/28395.28434}}.

\bibitem{Dietzfelbinger:1990}
Martin Dietzfelbinger and Friedhelm {Meyer auf der Heide}.
\newblock A new universal class of hash functions and dynamic hashing in real
  time.
\newblock In Mike Paterson, editor, {\em Automata, Languages and Programming,
  {ICALP} 1990}, volume 443 of {\em {LNCS}}, pages 6--19. Springer, 1990.
\newblock \href {http://dx.doi.org/10.1007/BFb0032018}
  {\path{doi:10.1007/BFb0032018}}.

\bibitem{Ferragina:1997}
Paolo Ferragina.
\newblock Dynamic text indexing under string updates.
\newblock {\em Journal of Algorithms}, 22(2):296--328, February 1997.
\newblock \href {http://dx.doi.org/10.1006/jagm.1996.0814}
  {\path{doi:10.1006/jagm.1996.0814}}.

\bibitem{DBLP:journals/siamcomp/FerraginaG98}
Paolo Ferragina and Roberto Grossi.
\newblock Optimal on-line search and sublinear time update in string matching.
\newblock {\em {SIAM} Journal on Computing}, 27(3):713--736, 1998.
\newblock \href {http://dx.doi.org/10.1137/S0097539795286119}
  {\path{doi:10.1137/S0097539795286119}}.

\bibitem{Gu:1994}
Ming Gu, Martin Farach, and Richard Beigel.
\newblock An efficient algorithm for dynamic text indexing.
\newblock In Daniel~Dominic Sleator, editor, {\em 5th Annual {ACM-SIAM}
  Symposium on Discrete Algorithms, {SODA} 1994}, SODA '94, pages 697--704,
  Philadelphia, PA, USA, 1994. {ACM/SIAM}.
\newblock URL: \url{http://dl.acm.org/citation.cfm?id=314464.314675}.

\bibitem{DBLP:conf/dcc/GasieniecKPS05}
Leszek Gąsieniec, Roman~M. Kolpakov, Igor Potapov, and Paul Sant.
\newblock Real-time traversal in grammar-based compressed files.
\newblock In {\em Data Compression Conference, {DCC} 2005}, page 458. {IEEE}
  Computer Society, 2005.
\newblock \href {http://dx.doi.org/10.1109/DCC.2005.78}
  {\path{doi:10.1109/DCC.2005.78}}.

\bibitem{DBLP:conf/stacs/Hagerup98}
Torben Hagerup.
\newblock Sorting and searching on the word {RAM}.
\newblock In Michel Morvan, Christoph Meinel, and Daniel Krob, editors, {\em
  Symposium on Theoretical Aspects of Computer Science, {STACS} 1998}, volume
  1373 of {\em LNCS}, pages 366--398. Springer, Berlin Heidelberg, 1998.
\newblock \href {http://dx.doi.org/10.1007/BFb0028575}
  {\path{doi:10.1007/BFb0028575}}.

\bibitem{Hirshfeld}
Yoram Hirshfeld, Mark Jerrum, and Faron Moller.
\newblock A polynomial algorithm for deciding bisimilarity of normed
  context-free processes.
\newblock {\em Theoretical Computer Science}, 158(1{\&}2):143--159, 1996.
\newblock \href {http://dx.doi.org/10.1016/0304-3975(95)00064-X}
  {\path{doi:10.1016/0304-3975(95)00064-X}}.

\bibitem{JezRecompression}
Artur Jeż.
\newblock Recompression: a simple and powerful technique for word equations.
\newblock In Natacha Portier and Thomas Wilke, editors, {\em Symposium on
  Theoretical Aspects of Computer Science, STACS 2013}, volume~20 of {\em
  Leibniz International Proceedings in Informatics (LIPIcs)}, pages 233--244.
  Schloss Dagstuhl--Leibniz-Zentrum für Informatik, 2013.
\newblock \href {http://dx.doi.org/10.4230/LIPIcs.STACS.2013.233}
  {\path{doi:10.4230/LIPIcs.STACS.2013.233}}.

\bibitem{JezFully}
Artur Jeż.
\newblock Faster fully compressed pattern matching by recompression.
\newblock {\em {ACM} Transactions on Algorithms}, 11(3):20:1--20:43, 2015.
\newblock \href {http://dx.doi.org/10.1145/2631920}
  {\path{doi:10.1145/2631920}}.

\bibitem{Lehre:2013}
Per~Kristian Lehre and Carsten Witt.
\newblock General drift analysis with tail bounds, 2013.
\newblock \href {http://arxiv.org/abs/1307.2559} {\path{arXiv:1307.2559}}.

\bibitem{LohreySurvey}
Markus Lohrey.
\newblock Algorithmics on {SLP}-compressed strings: {A} survey.
\newblock {\em Groups Complexity Cryptology}, 4(2):241--299, 2012.
\newblock \href {http://dx.doi.org/10.1515/gcc-2012-0016}
  {\path{doi:10.1515/gcc-2012-0016}}.

\bibitem{McCreight:1976}
Edward~M. McCreight.
\newblock A space-economical suffix tree construction algorithm.
\newblock {\em Journal of the {ACM}}, 23(2):262--272, April 1976.
\newblock \href {http://dx.doi.org/10.1145/321941.321946}
  {\path{doi:10.1145/321941.321946}}.

\bibitem{Mehlhorn:1990}
Kurt Mehlhorn and Stefan N{\"{a}}her.
\newblock Dynamic fractional cascading.
\newblock {\em Algorithmica}, 5(2):215--241, 1990.
\newblock \href {http://dx.doi.org/10.1007/BF01840386}
  {\path{doi:10.1007/BF01840386}}.

\bibitem{Mehlhorn}
Kurt Mehlhorn, Rajamani Sundar, and Christian Uhrig.
\newblock Maintaining dynamic sequences under equality tests in polylogarithmic
  time.
\newblock {\em Algorithmica}, 17(2):183--198, 1997.
\newblock \href {http://dx.doi.org/10.1007/BF02522825}
  {\path{doi:10.1007/BF02522825}}.

\bibitem{TreeContraction}
Gary~L. Miller and John~H. Reif.
\newblock Parallel tree contraction and its application.
\newblock In {\em 26th Annual {IEEE} Symposium on Foundations of Computer
  Science, {FOCS} 1985}, pages 478--489. IEEE Computer Society, 1985.
\newblock \href {http://dx.doi.org/10.1109/SFCS.1985.43}
  {\path{doi:10.1109/SFCS.1985.43}}.

\bibitem{Mortensen:2003}
Christian~Worm Mortensen.
\newblock Fully-dynamic two dimensional orthogonal range and line segment
  intersection reporting in logarithmic time.
\newblock In {\em 14th Annual {ACM-SIAM} Symposium on Discrete Algorithms,
  {SODA} 2003}, pages 618--627, Philadelphia, PA, USA, 2003. Society for
  Industrial and Applied Mathematics.
\newblock URL: \url{http://dl.acm.org/citation.cfm?id=644108.644210}.

\bibitem{DBLP:journals/corr/NishimotoIIBT15}
Takaaki Nishimoto, Tomohiro I, Shunsuke Inenaga, Hideo Bannai, and Masayuki
  Takeda.
\newblock Dynamic index, {LZ} factorization, and {LCE} queries in compressed
  space, 2015.
\newblock \href {http://arxiv.org/abs/1504.06954} {\path{arXiv:1504.06954}}.

\bibitem{Okasaki:1999}
Chris Okasaki.
\newblock {\em Purely functional data structures}.
\newblock Cambridge University Press, 1999.

\bibitem{logarithmic}
Mihai Patrascu and Erik~D. Demaine.
\newblock Logarithmic lower bounds in the cell-probe model.
\newblock {\em {SIAM} Journal on Computing}, 35(4):932--963, 2006.
\newblock \href {http://dx.doi.org/10.1137/S0097539705447256}
  {\path{doi:10.1137/S0097539705447256}}.

\bibitem{Plandowski}
Wojciech Plandowski.
\newblock Testing equivalence of morphisms on context-free languages.
\newblock In Jan van Leeuwen, editor, {\em Algorithms, {ESA} 1994}, volume 855
  of {\em {LNCS}}, pages 460--470. Springer, 1994.
\newblock \href {http://dx.doi.org/10.1007/BFb0049431}
  {\path{doi:10.1007/BFb0049431}}.

\bibitem{Pugh}
William Pugh and Tim Teitelbaum.
\newblock Incremental computation via function caching.
\newblock In {\em 16th Annual {ACM} Symposium on Principles of Programming
  Languages, POPL 1989}, pages 315--328. ACM, 1989.
\newblock \href {http://dx.doi.org/10.1145/75277.75305}
  {\path{doi:10.1145/75277.75305}}.

\bibitem{Rytter}
Wojciech Rytter.
\newblock Application of {L}empel-{Z}iv factorization to the approximation of
  grammar-based compression.
\newblock {\em Theoretical Computer Science}, 302(1-3):211--222, 2003.
\newblock \href {http://dx.doi.org/10.1016/S0304-3975(02)00777-6}
  {\path{doi:10.1016/S0304-3975(02)00777-6}}.

\bibitem{SV1995}
Süleyman~Cenk Sahinalp and Uzi Vishkin.
\newblock Data compression using locally consistent parsing.
\newblock Technical report, University of Maryland, Department of Computer
  Science, 1995.

\bibitem{Sahinalp:1996}
Süleyman~Cenk Sahinalp and Uzi Vishkin.
\newblock Efficient approximate and dynamic matching of patterns using a
  labeling paradigm.
\newblock In {\em 37th {IEEE} Annual Symposium on Foundations of Computer
  Science, {FOCS} 1996}, pages 320--328. IEEE Computer Society, 1996.
\newblock \href {http://dx.doi.org/10.1109/SFCS.1996.548491}
  {\path{doi:10.1109/SFCS.1996.548491}}.

\bibitem{sundartarjan}
Rajamani Sundar and Robert~E. Tarjan.
\newblock Unique binary-search-tree representations and equality testing of
  sets and sequences.
\newblock {\em SIAM Journal on Computing}, 23(1):24--44, 1994.
\newblock \href {http://dx.doi.org/10.1137/S0097539790189733}
  {\path{doi:10.1137/S0097539790189733}}.

\bibitem{onlinesuffix}
Esko Ukkonen.
\newblock On-line construction of suffix trees.
\newblock {\em Algorithmica}, 14(3):249--260, 1995.
\newblock \href {http://dx.doi.org/10.1007/BF01206331}
  {\path{doi:10.1007/BF01206331}}.

\bibitem{Vishkin}
Uzi Vishkin.
\newblock Randomized speed-ups in parallel computation.
\newblock In {\em 16th Annual {ACM} Symposium on Theory of Computing, {STOC}
  1984}, pages 230--239, New York, NY, USA, 1984. ACM.
\newblock \href {http://dx.doi.org/10.1145/800057.808686}
  {\path{doi:10.1145/800057.808686}}.

\bibitem{SuffixTree}
Peter Weiner.
\newblock Linear pattern matching algorithms.
\newblock In {\em 14th Annual Symposium on Switching and Automata Theory,
  {SWAT} 1973}, pages 1--11, Washington, DC, USA, 1973. IEEE Computer Society.
\newblock \href {http://dx.doi.org/10.1109/SWAT.1973.13}
  {\path{doi:10.1109/SWAT.1973.13}}.

\bibitem{Willard:1985}
Dan~E. Willard.
\newblock New data structures for orthogonal range queries.
\newblock {\em SIAM Journal on Computing}, 14, 1985.
\newblock \href {http://dx.doi.org/10.1137/0214019}
  {\path{doi:10.1137/0214019}}.

\end{thebibliography}

\end{document}